\documentclass[a4paper,english,cleveref, autoref, thm-restate]{lipics-v2021}

\newcommand{\ignore}[1]{}
\newcommand{\ceq}{\coloneqq}
\newcommand{\enex}{\hfill{$\triangle$}}
\newcommand{\enrem}{\hfill{$\triangle$}}
\usepackage{xspace,todonotes,stmaryrd} 
\usepackage{stmaryrd}
\usepackage{mathtools,hyperref}

\usepackage{tikz-cd}

\newcommand{\Dom}{C}
\newcommand{\atuple}[1]{{\langle #1 \rangle}}

\newtheorem{question}[theorem]{Question}

\usepackage{tikz,hyperref}

\title{Datalog-Expressibility for Monadic and Guarded Second-Order Logic} 

\titlerunning{Datalog for Guarded Second-Order Logic} 

\relatedversion{}
\relatedversiondetails{Previous Version}{https://arxiv.org/abs/2010.05677}

\author{Manuel Bodirsky}{TU Dresden, Institut f\"ur  Algebra, Germany 
\and \url{https://tu-dresden.de/mn/math/algebra/bodirsky}}{manuel.bodirsky@tu-dresden.de}{https://orcid.org/0000-0001-8228-3611}{
Funded by the European Union (project POCOCOP, ERC Synergy grant No. 101071674). Views and opinions expressed are however those of the author(s) only and do not necessarily reflect those of the European Union or the European Research Council Executive Agency. Neither the European Union nor the granting authority can be held responsible for them.}

\author{Simon Kn\"auer}{TU Dresden, Institut f\"ur  Algebra, Germany 
\and \url{https://tu-dresden.de/mn/math/algebra/das-institut/beschaeftigte/simon-knaeuer}}{simon.knaeuer@tu-dresden.de}{}{The author was supported by DFG Graduiertenkolleg 1763 (QuantLA).}

\author{Sebastian Rudolph}{TU Dresden, 
Computational Logic Group, Germany
\and \url{http://clgroup.inf.tu-dresden.de}}{sebastian.rudolph@tu-dresden.de}{https://orcid.org/0000-0002-1609-2080}{
The author has received funding from the European Research Council (Grant Agreement no. 771779, DeciGUT), the Bundesministerium für Bildung und Forschung (BMBF, Federal Ministry of Education and Research) in the Center for Scalable Data Analytics and Artificial Intelligence (ScaDS.AI), and by BMBF and DAAD (German Academic Exchange Service) in project 57616814 (SECAI, School of Embedded and Composite AI).}

\authorrunning{M. Bodirsky, S. Kn\"auer, S. Rudolph} 

\Copyright{Manuel Bodirsky and Simon K\"auer and Sebastian Rudolph} 



\category{} 

\relatedversion{} 

\supplement{}

\nolinenumbers 

\hideLIPIcs  

\EventEditors{John Q. Open and Joan R. Access}
\EventNoEds{2}
\EventLongTitle{42nd Conference on Very Important Topics (CVIT 2016)}
\EventShortTitle{CVIT 2016}
\EventAcronym{CVIT}
\EventYear{2016}
\EventDate{December 24--27, 2016}
\EventLocation{Little Whinging, United Kingdom}
\EventLogo{}
\SeriesVolume{42}
\ArticleNo{23}

\newcommand{\red}[1]{\textcolor{red}{#1}}

\newcommand{\bA}{\mathfrak{A}}
\newcommand{\bB}{\mathfrak{B}}
\newcommand{\bC}{\mathfrak{C}}
\newcommand{\bD}{\mathfrak{D}}

\newcommand{\bI}{\mathfrak{I}}
\newcommand{\bP}{\mathfrak{P}}
\newcommand{\bH}{\mathfrak{H}}
\newcommand{\bS}{\mathfrak{S}}
\newcommand{\bT}{\mathfrak{T}}

\DeclareMathOperator{\GSO}{GSO}

\DeclareMathOperator{\dom}{dom}

\DeclareMathOperator{\Aut}{Aut}

\newcommand{\ar}{\mathop{ar}}
\newcommand{\goal}{\ensuremath{{\sf goal}}\xspace}
\DeclareMathOperator{\Csp}{CSP}

\begin{document}

\maketitle

\begin{abstract}
We characterise the sentences in Monadic Second-order Logic (MSO) that are over finite structures equivalent to a Datalog program, in terms of an existential pebble game.  
We also show that for every class ${\mathcal C}$ of finite structures that can be expressed in MSO and is closed under homomorphisms, 
and for all $\ell,k \in {\mathbb N}$, there exists a \emph{canonical} Datalog program $\Pi$ of width $(\ell,k)$ in the sense of Feder and Vardi. 
The same characterisations also hold for Guarded Second-order Logic (GSO), which properly extends MSO. To prove our results, we show that every
class ${\mathcal C}$ in GSO whose complement is closed under homomorphisms is a finite union of constraint satisfaction problems (CSPs) of $\omega$-categorical structures.  
The intersection 
of MSO and Datalog is known to contain
the class of \emph{nested monadically defined queries (Nemodeq)}
; likewise, we show that the intersection 
of GSO and Datalog contains all problems that can be expressed by the more expressive language of 
\emph{nested guarded queries (GQ$^+$)}. 
Yet, by exploiting our results, we can show that neither of the two query languages can serve as a characterization, as we exhibit a CSP whose complement corresponds to a query in the intersection 
of MSO and Datalog 
that is not expressible in 
GQ$^+$.
\end{abstract}

\tableofcontents

\section{Introduction}
\emph{Monadic Second-order Logic (MSO)} is a logic of great importance in theoretical computer science. While it significantly exceeds the expressive capabilities of First-order Logic (allowing for expressing crucial structural properties like reachability or connectedness), it is still computationally reasonably well-behaved: 
By B\"uchi's theorem (see, e.g.,~\cite{Libkin}), the formal languages definable in MSO are precisely the regular ones; by Courcelle's theorem~\cite{EngelfrietCourcelle}, MSO sentences can be evaluated
in polynomial time on classes of structures whose treewidth is bounded by a constant. 
The latter result even holds for the more expressive logic of \emph{Guarded Second-order Logic (GSO)}~\cite{GraedelHirschOtto,Grohe-GSO}, which extends First-order Logic by second-order quantifiers over \emph{guarded relations}. 

Another fundamental formalism in theoretical computer science, 
which is particularly heavily studied in database theory and logic-based knowledge representation, is \emph{Datalog} (see, e.g.,~\cite{DantsinEGV01,Libkin}). Every fixed Datalog program can be evaluated on given finite structures in polynomial time. If a linear order is present and if the available relations are closed under complement, Datalog can even express \emph{all} polytime-computable queries \cite{AbiteboulHV95}.     
That is, like MSO and GSO, Datalog strikes a good balance between expressivity and good mathematical and computational properties.


Neither of the two formalisms subsumes the other in terms of what can be expressed. 
Yet, in various scenarios, we are interested in simultaneously having
the good computational properties of expressibility in Datalog 
\emph{and} having the good computational properties of expressibility in MSO or GSO. 
A wide variety of popular query formalisms (among them (unions of) conjunctive queries, (2-way conjunctive) regular path queries, monadic Datalog, the recently introduced almost monadic Datalog ~\cite{CateDalmauOprsal}, guarded Datalog, monadically defined queries, or nested monadically defined queries) are known to be both in Datalog and GSO~\cite{RK2013}. Query languages expressible both in Datalog and MSO have been shown to warrant decidable query entailment over logical theories exhibiting universal models of finite treewidth and cliquewidth~\cite{FLOR23a,FLORb}.
Also, all these formalisms have favourable properties when it comes to static analysis, most notably decidable query containment~\cite{RK2013,BourhisKroetzschRudolph}. 
Note that on the contrary, query containment in unrestricted Datalog is undecidable, as is query containment in unrestricted MSO / GSO\footnote{This follows from Trakhtenbrot's theorem which states that satisfiability of first-order logic in the finite is undecidable.}. The same holds for query entailment under logical theories, even if the latter are extremely simple. So it is really the interplay of the restrictions imposed by both formalisms that is required to ensure decidability of  central tasks in databases and knowledge representation. This makes the semantic intersection of Datalog and MSO / GSO so interesting and worthwhile investigating.

\medskip

In this article, we investigate two questions, which will 
turn out to be closely related: 
\begin{enumerate}
\item Which classes of finite structures are expressible both in MSO/GSO and in Datalog? 
\item Which \emph{constraint satisfaction problems (CSPs)} 
can be expressed in MSO/GSO? 
\end{enumerate}
Indeed, as our investigation reveals, the versatile and vibrant discipline of constraint satisfaction problems offers many suitable notions and tools for our endeavour to understand ``MSO/GSO $\cap$ Datalog''.
While this might come as a surprise at the first glance, interesting correspondences connecting CSPs with expressivity characterizations of logical formalisms from database theory and knowledge representation have been observed before \cite{BienvenuCLW14,FeierKL19}.

\medskip

We recall that, for a structure $\bB$ with a finite relational signature $\tau$, the \emph{constraint satisfaction problem for $\bB$} is the class of all finite $\tau$-structures that homomorphically map to~$\bB$. 
It is well known that, whenever $\bB$ is finite, its constraint satisfaction problem can already be expressed in a small fragment of MSO, called monotone monadic SNP (MMSNP,~\cite{FederVardi}). 
Yet, this correspondence does not hold for infinite $\bB$. 

\begin{example}\label{expl:Q}
The constraint satisfaction problem for the structure $({\mathbb Q};<)$, which is the class of all finite acyclic digraphs $(V;E)$, 
 cannot be expressed in MMSNP~\cite{Book}. It can, however, be expressed in MSO over the class of all finite digraphs by the sentence
$$\forall X \neq \emptyset \; \exists x \in X \; \forall y \in X \colon \neg E(x,y).$$
To see this, note that if $(V;E)$ is a finite digraph such that there exists a non-empty $X \subseteq V$ such that for every $x \in X$ there is a $y \in Z$ with $E(x,y)$, then $X$ contains a directed cycle, and hence has no homomorphism to $({\mathbb Q};<)$. Conversely, if $(V;E)$ 
is a finite digraph and contains no directed cycle, then every non-empty subset of $V$ must contain a sink, i.e., a vertex $x$ with no outgoing $E$-edges, and hence satisfies the given sentence.
\enex
\end{example}

The class of CSPs of arbitrary infinite structures $\bB$ is quite large; 
the following is easy to see (see, e.g.,~\cite[Lemma 1.1.8]{Book}). 

\begin{remark}\label{rem:CSP} 
For every finite relational signature $\tau$, a class $\mathcal D$ of finite 
$\tau$-structures is the CSP of some countably infinite structure if and only if 
 \begin{itemize}
 \item it is closed under disjoint unions, and 
 \item it contains any $\bA$ that 
maps homomorphically to some $\bA' \in \mathcal D$.  
\end{itemize} 
A class of finite relational structures that satisfies these two conditions is simply called \emph{a CSP}. 
\end{remark}

The second condition in Remark~\ref{rem:CSP}, which is sometimes also referred to as \emph{closure under inverse homomorphisms}, can be equivalently rephrased by requiring that the \emph{complement}\footnote{Whenever in this paper we refer to the complement of some class $\mathcal D$ of finite $\tau$-structures, we mean the class of all finite $\tau$-structures not contained in $\mathcal D$ -- we will make $\tau$ explicit or it will be clear from the context.} of ${\mathcal D}$ is 
 \emph{closed under homomorphisms}, where
 a class $\mathcal C$ is closed under homomorphisms 
if for any structure $\bA \in \cal C$ that maps homomorphically to some $\bC$ we have $\bC \in \cal C$.
Examples of classes of structures that are closed under homomorphisms
naturally arise from Datalog. 
We say that a class $\mathcal C$ of finite $\tau$-structures \emph{is definable in Datalog} (or shorter: \emph{is in Datalog})\footnote{Warning: 
Feder and Vardi~\cite{FederVardi} say that a CSP is in Datalog if its \emph{complement} in the class of all finite $\tau$-structures is definable in Datalog in our sense.}
if there exists a Datalog program $\Pi$ with a distinguished nullary predicate 
 \goal\  such that $\Pi$ derives  \goal\  on a finite $\tau$-structure if and only if the structure is in $\mathcal C$; in this case, we may denote $\mathcal C$ by $\llbracket \Pi \rrbracket$. 
%
While each class of $\tau$-structures that is definable in Datalog is closed under homomorphisms, not every such class corresponds to the complement of a CSP. 


\begin{example}
Consider, for unary predicates $R$ and $B$, the class $\mathcal C_{R,B}$ of finite $\{R,B\}$-structures $\bA$ such that $R^{\bA}$ is empty or $B^{\bA}$ is empty. 
Clearly, ${\mathcal C}_{R,B}$ is not closed under disjoint unions, so it cannot correspond to the CSP of any structure. Yet, it is easy to see that a finite structure is in ${\mathcal C}_{R,B}$ if and only if the Datalog program consisting of just the one rule 
$$ \goal \; {:}{-} \; R(x), B(y) $$
does not derive \goal\  on that structure.
\enex
\end{example}

A noteworthy subclass of all CSPs are the CSPs of structures $\bB$ that
are countably infinite and \emph{$\omega$-categorical}. 
A structure $\bB$ is \emph{$\omega$-categorical} if all countable models of $\bB$'s first-order theory are isomorphic to $\bB$. A well-known example of an $\omega$-categorical structure is $({\mathbb Q};<)$, a result going back to Cantor~\cite{CantorDicht}. 
We would like to mention that 
constraint satisfaction problems
of $\omega$-categorical structures
 can be evaluated in polynomial time on any class of structures whose treewidth is bounded by some constant $k \in {\mathbb N}$, by a result of Bodirsky and Dalmau~\cite{BodDalJournal}, using Datalog programs; however, the concept of bounded treewidth is not needed in the present article. 
 
Two important parameters of a Datalog program $\Pi$ are the maximal arity $\ell$ of its auxiliary predicates (IDBs), and the maximal number $k$ of variables per rule in $\Pi$ -- we then say that $\Pi$ has 
\emph{width $(\ell,k)$}, following the terminology of Feder and Vardi~\cite{FederVardi}. 
The special case of $\ell=1$ is typically referred to as \emph{monadic Datalog}. 
The parameters $\ell$ and $k$ are important in theory, but they also bear some practical relevance:
When evaluating $\Pi$ on a given structure $\bA$ with domain $A$, the memory needed is bounded by $O(|A|^\ell)$ and the number of computation steps by $O(|A|^k)$. 
The polynomial-time algorithm for instances of an $\omega$-categorical CSP of treewidth at most $k$
 presented by Bodirsky and Dalmau is in fact a Datalog program of width $(k-1,k)$.  
 A Datalog program $\Pi$ is called \emph{sound} for a class of $\tau$-structures ${\mathcal C}$ if 
 $\llbracket \Pi \rrbracket \subseteq {\mathcal C}$. 
 Bodirsky and Dalmau showed that if ${\mathcal C}$ is the complement of the CSP of an $\omega$-categorical $\tau$-structure~$\bB$
 then there exists, for any $\ell,k \in {\mathbb N}$, a Datalog program $\Pi$ of width $(\ell,k)$ such that 
 \begin{itemize}
 \item $\Pi$ is sound for ${\mathcal C}$, and 
\item $\llbracket  \Pi' \rrbracket  \subseteq \llbracket \Pi \rrbracket$ for 
every Datalog program $\Pi'$ of width $(\ell,k)$
which is sound for ${\mathcal C}$. 
\end{itemize}
This $\Pi$ (which is unique up to immaterial syntactic variations) is then referred to as the \emph{canonical Datalog program of width $(\ell,k)$ for ${\mathcal C}$}.

Interestingly and conveniently, there is a known game-theoretic characterisation capturing whether the canonical Datalog program of width $(\ell,k)$ for ${\mathcal C}$ derives \goal\ on a given $\tau$-structure $\bA$~\cite{BodDalJournal}. 
This characterisation is based on the existential pebble game from finite model theory, which is played on a pair $(\bA,\bB)$ of structures. In more detail, the \emph{existential $(\ell,k)$-pebble game} is played by two players, called \emph{Spoiler} and \emph{Duplicator} (see, e.g., \cite{DalmauKolaitisVardi,FederVardi,KolaitisVardiDatalog}). Spoiler starts by placing $k$ pebbles on elements $a_1,\dots,a_k$ of $\bA$, and Duplicator responds by placing $k$ pebbles $b_1,\dots,b_k$ on $\bB$. If the map that sends $a_1,\dots,a_k$ to $b_1,\dots,b_k$ is not a partial homomorphism from $\bA$ to $\bB$, then the game is over and Spoiler wins. Otherwise, 
Spoiler removes all but at most $\ell$ pebbles from $\bA$, and Duplicator has to respond by removing the corresponding pebbles from $\bB$. Then Spoiler can again place all his pebbles on $\bA$, and Duplicator must again respond by placing her pebbles on $\bB$. If the game continues forever,
then Duplicator wins. 
If $\bB$ is a finite, or more generally a countable $\omega$-categorical structure then Spoiler has a winning strategy for the existential $(\ell,k)$-pebble game on $(\bA,\bB)$ if and only if the canonical Datalog program of width $(\ell,k)$ for $\Csp(\bB)$ derives \goal\ on $\bA$ (Theorem~\ref{thm:BD}). 
This connection played an essential role in proving Datalog inexpressibility results, for example for the class of finite-domain CSPs~\cite{AtseriasBulatovDawar} (leading to a complete classification of those finite structures $\bB$ such that the complement of $\Csp(\bB)$ can be expressed in Datalog~\cite{BoundedWidthJournal}).


\subsection*{Results and Consequences}
In this article, we prove that 
every class of finite structures in GSO 
whose complement is closed under 
homomorphisms is a finite union 
of CSPs that can also be expressed in GSO (Theorem~\ref{thm:decomp};
an analogous statement holds for MSO). 
Moreover, every CSP in GSO is the CSP of  
a countable $\omega$-categorical structure (Corollary~\ref{cor:omega-cat}). 
We highlight that this result elegantly generalises known results: 
\begin{itemize}
\item our result generalises the fact that every CSP in FO is the CSP of a countable $\omega$-categorical structure, which can be seen from combining Rossmann's theorem~\cite{Rossman08} with a generalisation of the theorem of Cherlin, Shelah, and Shi~\cite{CherlinShelahShi} from graphs to general relational structures. 
\item Our result also generalises the fact that every CSP in the logic MMSNP (for \emph{monotone monadic strict NP}~\cite{FederVardi}) is the CSP of a countable $\omega$-categorical structure~\cite{BodDalJournal}.
\end{itemize} 
It follows that every class 
of finite structures that can be expressed both  
in GSO and in Datalog 
is an intersection of \emph{finitely many} complemented CSPs of $\omega$-categorical structures. 
In contrast, it is \emph{not} generally true that a Datalog program describes a finite intersection of complements of CSPs (we present a counterexample in Example~\ref{expl:datalog-completement}). 

Next, we use the connection between GSO and CSPs to present a characterisation of those GSO sentences 
$\Phi$ that are over finite structures equivalent to a Datalog program (Section~\ref{sect:main}). 
 Our characterisation involves a variant of the existential pebble game from finite model theory, which we call the \emph{$(\ell,k)$-game}.
This game is defined for a homomorphism-closed class ${\mathcal C}$ of finite $\tau$-structures, and it is played by the two players Spoiler and Duplicator on a finite $\tau$-structure $\bA$ as follows. 
\begin{itemize}
\item Duplicator picks a countable $\tau$-structure $\bB$ such that $\Csp(\bB) \cap {\mathcal C} = \emptyset$. 
\item The game then continues as the existential $(\ell,k)$-pebble game played by Spoiler and Duplicator on $(\bA,\bB)$, as described above. 
\end{itemize}
We use results from~\cite{BodDalJournal} to show that a GSO sentence $\Phi$ is over finite structures equivalent to a Datalog program of width $(\ell,k)$ if and only if 
\begin{itemize}
\item $\llbracket \Phi \rrbracket$ is closed under homomorphisms, and 
\item Spoiler wins the existential $(\ell,k)$-game for $\llbracket \Phi \rrbracket$ on $\bA$ if and only if $\bA \models \Phi$. 
\end{itemize}
We also show that for every class of finite models ${\mathcal C}$ which is closed under homomorphisms and expressible in GSO, 
and for all $\ell,k \in {\mathbb N}$, there exists a canonical Datalog program $\Pi$ of width $(\ell,k)$ for $\mathcal C$ (Theorem~\ref{thm:main}). 

Our next series of results concern the most expressive syntactically defined formalism known from the literature that is contained in both Datalog and MSO, namely \emph{nested monadically defined queries}~\cite{RK2013}, as well as 
the most expressive syntactically defined formalism known from the literature that is 
 contained in both Datalog and GSO, namely \emph{nested guarded queries}~\cite{BourhisKroetzschRudolph}.

We prove that there are problems in the intersection of Datalog and GSO that cannot be expressed as a nested guarded query (Section~\ref{sect:GQ}). 
To prove this result, we introduce a modified version of the existential pebble game, which we call the \emph{nested guarded pebble game}, and which captures precisely the expressiveness of nested guarded queries (Theorem~\ref{thm:multi}). 
We also present an example of a problem which lies in the intersection of Datalog and MSO that cannot be expressed by nested monadically defined queries (and not even by nested guarded queries;  Corollary~\ref{cor:gqplus}).

Finally, we provide evidence that the class of CSPs that can be expressed in MSO is \emph{not} contained in the intensively studied class of CSPs for reducts of finitely bounded homogeneous structures: 
in Section~\ref{sect:coNP} we present an example of a CSP which is expressible in MSO and coNP-complete, and hence not the CSP of a reduct of a finitely bounded homogeneous structure, unless NP=coNP  (Proposition~\ref{prop:couterexpl}). As an illustration of the results of Section~\ref{sect:main}, we also prove that the given MSO sentence is not equivalent to a Datalog program.  

Some of the results of this article until Section~\ref{sect:main} have been announced  in a conference paper with the title ``Datalog-Expressibility for Monadic and Guarded Second-Order Logic'' in the proceedings of ICALP'21~\cite{BKR}; the results in Section~\ref{sect:GQ} about nested monadic and nested guarded queries were not yet present in the conference version.

\section{Preliminaries}
In the entire text, $\tau$ denotes a finite signature containing relation symbols and sometimes also constant symbols. 
If $R \in \tau$ is a relation symbol, we write
$\ar(R)$ for its arity. 
If $\bA$ is a $\tau$-structure we use the corresponding roman capital letter $A$ to denote the domain of $\bA$; the domains of structures are assumed to be 
non-empty.  If $R \in \tau$, then $R^{\bA} \subseteq A^{\ar(R)}$ denotes the corresponding relation of $\bA$.

A \emph{primitive positive $\tau$-formula} (in database theory also referred to as \emph{conjunctive query}) is a 
first-order $\tau$-formula without disjunction, negation, and universal quantification. Every 
primitive positive formula is equivalent to a formula
of the form $$\exists x_1,\dots,x_n (\psi_1 \wedge \cdots \wedge \psi_m)$$
where $\psi_1,\dots,\psi_m$ are atomic $\tau$-formulas, i.e., formulas built from relation symbols in $\tau$ or equality. 
An \emph{existential positive $\tau$-formula} 
is a first-order $\tau$-formula without negation and universal quantification. 
We write $\psi(x_1\dots,x_n)$ if the free variables of $\psi$ are from $x_1,\dots,x_n$. 
If $\bA$ is a $\tau$-structure
and $\psi(x_1,\dots,x_n)$ is a $\tau$-formula, then the relation $$R \ceq \{(a_1,\dots,a_n) \mid \bA \models \psi(a_1,\dots,a_n) \}$$ is
called the relation \emph{defined by $\psi$ over $\bA$}; if $\psi$ can be chosen to be primitive positive (or existential positive)
then $R$ is called \emph{primitively positively definable} (or \emph{existentially positively definable}, respectively).

For all logics over the signature $\tau$ considered in this text, we say that 
two formulas $\phi(x_1,\dots,x_n)$ and $\psi(x_1,\dots,x_n)$ are \emph{equivalent (over finite structures)}  if for all (finite) $\tau$-structures $\bA$ and all $a_1,\dots,a_n \in A$ we have 
$$ \bA \models \phi(a_1,\dots,a_n) \Leftrightarrow \bA \models \psi(a_1,\dots,a_n).$$
It is easy to see that every
existential positive $\tau$-formula is a disjunction of primitive positive $\tau$-formulas (and hence referred to as a \emph{union of conjunctive queries} in database theory). 
If $\phi$ is a primitive positive $\tau$-formula
without equality for a relational signature $\tau$, 
then the \emph{canonical database} is the $\tau$-structure $\bA$ whose domain consists of all the variables of $\phi$, and where 
$R^{\bA}$, for $R \in \tau$ of arity $k$, consists of all tuples $(x_1,\dots,x_k)$ for which $\phi$ contains the conjunct $R(x_1,\dots,x_k)$.

Formulas without free variables are called \emph{sentences}; in database theory,
formulas are often called \emph{queries}
and sentences are often called \emph{Boolean queries}.  If $\phi$ is a sentence, we let $\llbracket \phi \rrbracket$ denote the class of all finite models of $\phi$.

A \emph{reduct} of a relational structure 
$\bA$ is a structure $\bA'$ obtained from $\bA$ by dropping some of the relations. $\bA$ is then also called an \emph{expansion} of $\bA'$.

\subsection{Datalog}
\label{sect:datalog}
In this section we consider a finite set $\tau$ of constant and relation symbols, the latter being referred to as \emph{EDBs} (for \emph{extensional database predicates}). 
Let $\rho$ be a finite set of new relation symbols, called the \emph{IDBs} (for \emph{intensional database predicates}). 
A Datalog program is a set of rules of the form 
$$\psi_0 \; {:}{-} \; \psi_1,\dots,\psi_n$$
where $\psi_0$ (called the \emph{head} of the rule) is an atomic $\rho$-formula
and $\psi_1,\dots,\psi_n$ (jointly referred to as the \emph{body} of the rule)
are atomic $(\rho \cup \tau)$-formulas;
we also assume that all rules are \emph{safe}, i.e., that every variable that appears in the head also appears in the body. 
If $\bA$ is a $\tau$-structure, and $\Pi$ is a
Datalog program with EDBs from $\tau$ and IDBs $\rho$, then a $(\tau \cup \rho)$-expansion $\bA'$ of $\bA$ is called a \emph{fixed point of $\Pi$ 
on $\bA$} if $\bA'$ satisfies the sentence
$$ \forall \bar x (\psi_0 \vee \neg \psi_1 \vee \dots \vee \neg \psi_n) $$
for each rule $\psi_0 \; {:}{-} \; \psi_1,\dots,\psi_n$. 
If $\bA_1$ and $\bA_2$ are two $(\rho \cup \tau)$-structures with the same domain $A$, we let $\bA_1 \cap \bA_2$ denote the $(\rho \cup \tau)$-structure with domain $A$ such that 
\begin{itemize}
\item $R^{\bA_1 \cap \bA_2} \ceq R^{\bA_1} \cap R^{\bA_2}$, and
\item $c^{\bA_1 \cap \bA_2} = c^{\bA_1} = c^{\bA_2}$.
\end{itemize}
Note that if $\bA_1$ and $\bA_2$ are two fixed points of $\Pi$ on $\bA$, then $\bA_1 \cap \bA_2$ is a fixed point of $\Pi$ on $\bA$, too. 
Hence, there exists a unique smallest (with respect to inclusion) fixed point of $\Pi$ on $\bA$, which we denote by $\Pi(\bA)$. 

\begin{remark}\label{rem:op-dat} 
There is an equivalent `operational' definition of the semantics of Datalog, and in particular of $\Pi(\bA)$; roughly speaking, we compute $R^{\Pi(\bA)}$ for $R \in \rho$  `bottom up', starting from the empty relation, and adding tuples that must be contained in all fixed points, until we reach a (the) smallest fixed point (for details see, e.g., Libkin's book~\cite[Section 10.5]{Libkin}; an explicit treatment of the equivalence can e.g.\ be found in~\cite[Theorem 8.1.6]{Book}). 
From this equivalent description of $\Pi(\bA)$
it is apparent that if $\bA$ is a finite structure then $\Pi(\bA)$ can be computed in polynomial time in the size of $\bA$.
\enrem 
\end{remark}

If $R \in \rho$, 
we also say that $\Pi$ \emph{defines $R^{\Pi(\bA)}$ on $\bA$}. A Datalog program together with a distinguished predicate $R \in \rho$ may also be viewed as a formula, which we also call a \emph{Datalog query}, and which over a given $\tau$-structure $\bA$ denotes the relation 
$R^{\Pi(\bA)}$. 
If the distinguished predicate has arity $0$, 
we often call it the \emph{goal predicate}; we say that \emph{$\Pi$ derives \goal\ on $\bA$} if 
$\goal^{\Pi(\bA)} = \{()\}$. 
The class ${\mathcal C}$ of finite $\tau$-structures $\bA$ such that $\Pi$ derives \goal\ on $\bA$ is called \emph{the class of finite $\tau$-structures defined by $\Pi$}, and denoted by $\llbracket \Pi \rrbracket$. 
Note that this class ${\mathcal C}$ is definable in  universal second-order logic (we have to express that in every expansion of the input by relations for the IDBs that satisfies all the rules of the Datalog program the goal predicate is non-empty). 

\subsection{Second-Order Logic and Monadic Second-Order Logic}
\emph{Second-order logic} is the extension of first-order logic which additionally allows 
existential and universal quantification over relations;
that is, if $R$ is a relation symbol and $\phi$ is a second-order $\tau \cup \{R\}$-formula, then $\exists R. \phi$ and $\forall R. \phi$
are second-order $\tau$-formulas. 
 If $\bA$ is a $\tau$-structure and $\Phi$ is a second-order $\tau$-sentence, we write $\bA \models \Phi$ (and say that $\bA$ is a model of $\Phi$) if
 $\bA$ satisfies $\Phi$, which is defined in the usual Tarskian style. 
 We write 
  $\llbracket \Phi \rrbracket$ for the class of all finite models of $\Phi$. 
   A second-order formula is called 
 \emph{monadic} if all second-order variables are unary. We use syntactic sugar and also write 
 $\forall x \in X \colon \psi$ instead of 
 $\forall x (X(x) \Rightarrow \psi)$
 and $\exists x \in X \colon \psi$ instead of 
 $\exists x (X(x) \wedge \psi)$. 
 
\emph{Monadic second-order logic (MSO)} is surprisingly powerful for expressing CSPs, which we illustrate with the following example. 

\begin{example}\label{expl:eq-mso}
Let $E := \{(x,x) \mid x \in {\mathbb N}\}$ be the binary equality relation on ${\mathbb N}$ and let 
$D := \{(x,y) \in {\mathbb N}^2 \mid x \neq y\}$ be the binary disequality relation on ${\mathbb N}$. 
The problem $\Csp({\mathbb N};E,D)$ is expressible in MSO. To see this, let $\phi_X$, for a unary relation symbol $X$, be the following $\{E,D,X\}$-sentence.
$$ \exists x. X(x) \wedge \forall x,y \big (X(x) \wedge E(x,y) \Rightarrow X(y) \big) \wedge (\big (X(x) \wedge E(y,x) \Rightarrow X(y) \big)$$
For unary relation symbols $X$ and $Y$ we write $Y \subsetneq X$ as a  shortcut for $$\forall x (Y(x) \Rightarrow X(x)) \wedge \exists x (X(x) \wedge \neg Y(x)).$$ 
Let $\psi_X$ be the sentence that states that $X$ is a smallest set which satisfies $\phi_X$: 
$$ \phi_X \wedge \forall Y (Y \subsetneq X \Rightarrow \neg \phi_Y)$$
Note that if $\bA$ is an $\{E,D\}$-structure
and $X \subseteq A$ satisfies $\psi_X$,
then any homomorphism from $\bA$ to $(\mathbb{N};E,D)$ 
must be constant on $X$. Also note that if
$X,Y \subseteq A$ satisfy $\psi_X$ and $\psi_Y$ 
and $X \cap Y \neq \emptyset$, then $Z := X \cap Y$ satisfies $\psi_Z$. Hence, for every element $a \in A$ there exists a unique smallest set $X = X_a$ that satisfies $\psi_X$ and contains $a$.

We claim that the CSP above can be expressed by
\begin{align}
 \forall X \big (\psi_X 
 \Rightarrow \forall x,y (D(x,y) \Rightarrow (\neg X(x) \vee \neg X(y)) \big).  \label{eq:eq-mso}
 \end{align} 
Indeed, suppose that $\bA$ is an $\{E,D\}$-structure and let $X \subseteq A$
be smallest with the property that it satisfies $\phi_X$. If there are $x,y \in X$ such that $(x,y) \in D^{\bA}$ then clearly there is no homomorphism from $\bA$ to $({\mathbb N};E,D)$. On the other hand, if $\bA$ satisfies~\eqref{eq:eq-mso}, then any map $h \colon A \to {\mathbb N}$ such that  $h^{-1}(h(a)) = X_a$ for all $x \in A$ is a homomorphism from $\bA$ to $({\mathbb N};E,D)$. Clearly, there exists such a map $h$, which proves the claim.
\enex 
\end{example}


\subsection{Guarded Second-Order Logic} 
\emph{Guarded Second-order Logic (GSO)}, introduced by Gr\"adel, Hirsch, and Otto~\cite{GraedelHirschOtto}, 
is the extension of \emph{guarded first-order logic} by second-order quantifiers.  
Guarded (first-order) $\tau$-formulas are defined inductively by the following rules~\cite{Guarded}:
\begin{enumerate}
\item all atomic $\tau$-formulas are guarded $\tau$-formulas;
\item if $\phi$ and $\psi$ are guarded $\tau$-formulas, then so are $\phi \wedge \psi$, $\phi \vee \psi$, and $\neg \phi$. 
\item if $\psi(\bar x, \bar y)$ is a guarded $\tau$-formula and $\alpha(\bar x, \bar y)$ is an atomic $\tau$-formula such that all free variables of $\psi$ occur in $\alpha$ then 
$\exists \bar y \big ( \alpha(\bar x, \bar y) \wedge \psi(\bar x, \bar y) \big )$
and $\forall \bar y \big ( \alpha(\bar x, \bar y) \Rightarrow \psi(\bar x, \bar y) \big )$ 
are guarded $\tau$-formulas. 
\end{enumerate} 
Guarded second-order formulas are defined similarly, but we additionally allow (unrestricted) second-order quantification; GSO generalises Courcelle's logic MSO$_2$ from graphs to general relational structures. 

\begin{definition}\label{def:guarded}
A second-order $\tau$-formula is called \emph{guarded} if it is defined inductively by the rules (1)-(3) for guarded first-order logic and additionally by second-order quantification. 
\end{definition}


There are many semantically equivalent ways of introducing GSO~\cite{GraedelHirschOtto}. 
Let $\bB$ be a $\tau$-structure. Then a tuple $(t_1,\dots,t_n) \in B^n$
is called \emph{guarded in $\bB$} if there exists 
an atomic $\tau$-formula $\phi$ and $b_1,\dots,b_k$
such that $\bB \models \phi(b_1,\dots,b_k)$ and $\{t_1,\dots,t_n\} \subseteq \{b_1,\dots,b_{k} \}$. 
 A relation $R \subseteq B^n$ is called \emph{guarded} if all tuples in $R$ are guarded. 
Note that (for $n = 1$) every element of $B$ is guarded (because of the valid atomic formula $x=x$), and consequently 
all unary relations are guarded. 
If $\Psi$ is an arbitrary second-order sentence, 
we say that a structure $\bB$ \emph{satisfies $\Psi$ with guarded semantics}, in symbols $\bB \models_g \Psi$, if 
all second-order quantifiers in $\Psi$ are evaluated 
over guarded relations only. Note that for MSO sentences, the usual semantics and the guarded semantics coincide.

\begin{proposition}[see Proposition 3.9 in~\cite{GraedelHirschOtto}]
\label{prop:guarded-semantics-equiv}
Guarded Second-order Logic and 
full Second-order Logic with guarded semantics are equally expressive. 
\end{proposition}

It follows that GSO is at least as expressive as MSO. 
There are many Datalog programs that are equivalent to a GSO sentence, but not to an MSO sentence. However, since MSO is surprisingly expressive (see Example~\ref{expl:eq-mso}) it may be quite challenging to \emph{prove} that a specific problem cannot be expressed in MSO.
For example, it is an open problem whether $\Csp({\mathbb Q};B)$
is expressible in MSO, 
where $$B = \{(x,y,z) \in {\mathbb Q}^3 \mid x < y < z \vee z < y < x\}$$ is the so-called \emph{betweenness relation}~\cite{EngelfrietCourcelle}. 
One such MSO-inexpressibility result is part of the following proposition whose proof  is based on a variant of an example of a Datalog query in GSO given in~\cite{BourhisKroetzschRudolph} (Example 2). 

\begin{proposition}\label{prop:not-mso}
There is a Datalog query that can be expressed in  GSO but not in MSO. 
\end{proposition}
\begin{proof}
Let $\tau$ be the signature consisting of the binary relation symbols $S,T,R,N$, and let $\mathcal C$ be the class of finite $\tau$-structures such that the following Datalog program with one binary IDB $U$ derives \goal.
\begin{align*}
U(x,y) & \; {:}{-} S(x,y) \\
U(x',y') & \; {:}{-} U(x,y), N(x,x'), N(y,y'), R(x',y') \\
\goal & \; {:}{-} U(x,y), T(x,y) 
\end{align*}
On the left of Figure~\ref{fig:leiter} one can find an example of a 
$\{S,T,R,N\}$-structure $\bB$ where the given Datalog program derives \goal. 
Note that the Datalog program derives \goal on a given finite structure if and only if a `ladder' structure as shown on the left of Figure~\ref{fig:leiter} 
homomorphically maps to the structure.  

\begin{figure}
\begin{subfigure}[b]{5.8cm}
	\begin{tikzpicture}[scale=2.6]
		
		\node (A1) at (0,0) [circle,inner sep=1pt, outer sep=2pt, label=center:{$v_1$}] {};
		\node (A2) at (0.5,0)[circle,inner sep=1pt, outer sep=2pt, label=center:{$v_2$}] {};
		\node (A3) at (1,0)[circle,inner sep=1pt, outer sep=2pt, label=center:{$v_3$}] {};
		\node (A4) at (1.5,0)[circle,inner sep=1pt , outer sep=2pt, label=center:{$v_4$}] {};
		\node (a1) at (0,0.5)[circle,inner sep=1pt , outer sep=2pt, label=center:{$w_1$}] {};
		\node (a2) at (0.5,0.5)[circle,inner sep=1pt, outer sep=2pt, label=center:{$w_2$}] {};
		\node (a3) at (1,0.5)[circle,inner sep=1pt , outer sep=2pt, label=center:{$w_3$}] {};
		\node (a4) at (1.5,0.5)[circle,inner sep=1pt , outer sep=2pt, label=center:{$w_4$}] {};
		
		\node () at (0,-0.23)[] {};
		
		%
		%
		%
		%
		%
		
		\node () at (-0.12,0.25)[] {{$S$}};
		\node () at (1.38,0.25)[] {{$T$}};
		\node () at (0.38,0.25)[] {{$R$}};
		\node () at (0.88,0.25)[] {{$R$}};
		\node () at (0.25,0.6)[] {{$N$}};
		\node () at (0.75,0.6)[] {{$N$}};
		\node () at (1.25,0.6)[] {{$N$}};
		\node () at (0.25,-0.1)[] {{$N$}};
		\node () at (0.75,-0.1)[] {{$N$}};
		\node () at (1.25,-0.1)[] {{$N$}};

		\draw[->, >=latex,dotted,  line width=1pt, 	shorten >=4pt, shorten <=4pt]  (A1) to (a1);
		\draw[->, >=latex, line width=0.5pt, shorten >=4pt, shorten <=4pt]  (A2) to (a2);
		\draw[->, >=latex, line width=0.5pt, shorten >=4pt, shorten <=4pt]  (A3) to (a3);
		\draw[->, >=latex, dashed ,  line width=1pt,shorten >=4pt, shorten <=4pt]  (A4) to (a4);

		\draw[->, >=latex, line width=1.2pt, shorten >=5pt, shorten <=5pt]  (A1) to (A2);
		\draw[->, >=latex, line width=1.2pt, shorten >=5pt, shorten <=5pt]  (A2) to (A3);
		\draw[->, >=latex, line width=1.2pt, shorten >=5pt, shorten <=5pt]  (A3) to (A4);
		
		\draw[->, >=latex, line width=1.2pt, shorten >=5pt, shorten <=5pt]  (a1) to (a2);
		\draw[->, >=latex, line width=1.2pt, shorten >=5pt, shorten <=5pt]  (a2) to (a3);
		\draw[->, >=latex, line width=1.2pt,shorten >=5pt, shorten <=5pt]  (a3) to (a4);

	\end{tikzpicture}
	\subcaption{Structure $\mathfrak{B}$}
\end{subfigure}
\begin{subfigure}[b]{5.5cm}
	\begin{tikzpicture}[scale=2.4]
		
		\node (A1) at (0,0) [circle,inner sep=1pt, outer sep=2pt, label=center:{$v_1$}] {};
		\node (A2) at (0.5,0)[circle,inner sep=1pt, outer sep=2pt, label=center:{$v_2$}] {};
		\node (A3) at (1,0)[circle,inner sep=1pt, outer sep=2pt, label=center:{$v_3$}] {};
		\node (A4) at (1.5,0)[circle,inner sep=1pt , outer sep=2pt, label=center:{$v_4$}] {};
		\node (a1) at (0,0.5)[circle,inner sep=1pt , outer sep=2pt, label=center:{$w_1$}] {};
		\node (a2) at (0.5,0.5)[circle,inner sep=1pt, outer sep=2pt, label=center:{$w_2$}] {};
		\node (a3) at (1,0.5)[circle,inner sep=1pt , outer sep=2pt, label=center:{$w_3$}] {};
		\node (a4) at (1.5,0.5)[circle,inner sep=1pt , outer sep=2pt, label=center:{$w_4$}] {};
		
		%
		\node () at (0.25,0)[] {{\Large $<$}};
		\node () at (0.75,0)[] {{\Large $<$}};
		\node () at (1.25,0)[] {{\Large $<$}};
		
		\node () at (0.25,0.5)[] {{\Large $<$}};
		\node () at (0.75,0.5)[] {{\Large $<$}};
		\node () at (1.25,0.5)[] {{\Large $<$}};
		
		\node (x) at (0.75,0.25)[] {{\Large $>$}};

		\draw[
		shorten >=3pt, shorten <=3pt
		]  (x) to[out=180, in=-90]  (a1);
		
		\draw[
		shorten >=4pt, shorten <=3pt
		]  (x) to[out=0, in=90]  (A4);

		\node () at (0,0.7)[] {\textcolor{blue}{$P_b$}};
		\node () at (0.5,0.7)[] {\textcolor{blue}{$ P_b$}};
		\node () at (1.5,0.7)[] {\textcolor{blue}{$ P_b$}};
		\node () at (1.0,0.7)[] {\textcolor{blue}{$ P_b$}};
		
		\node () at (0,-0.2)[]{\textcolor{red}{$P_a$}}; 
		\node () at (0.5,-0.2)[]{\textcolor{red}{$P_a$}}; 
		\node () at (1,-0.2)[]{\textcolor{red}{$P_a$}}; 
		\node () at (1.5,-0.2)[]{\textcolor{red}{$P_a$}};

	\end{tikzpicture}
	\subcaption{Structure $\mathfrak{A}$}
\end{subfigure}
\begin{subfigure}[b]{2cm}
	\begin{tikzpicture}[scale=2.2]
		
		\node () at (0,0)[] {{\Large {\textcolor{red}{$aaaa$}\textcolor{blue}{$bbbb$}}}};

		\node () at (0,-0.7)[] {};

	\end{tikzpicture}
	\subcaption{Word $w_\mathfrak{A}$}
\end{subfigure}

\caption{An example of an $\{S,T,R,N\}$-structure $\bB$ in the class $\mathcal C$ of Proposition~\ref{prop:not-mso}.} 
\label{fig:leiter}
\end{figure}
To show that ${\mathcal C}$ is not MSO-expressible, suppose toward a contradiction
suppose for contradiction that there exists an MSO sentence $\Phi$ such that $\llbracket \Phi \rrbracket = {\mathcal C}$. We use $\Phi$ to construct an MSO sentence $\Psi$ 
which holds on a finite word $w \in \{a,b\}^*$ (represented as a structure with the signature $\{P_a,P_b,<\}$ in the usual way~\cite{Libkin}) if and only if $w \in \{a^nb^n \mid n \geq 1\}$; this contradicts the theorem of B\"uchi-Elgot-Trakhtenbrot (see, e.g.,~\cite{Libkin}),
which states that a language of finite words can be expressed in MSO if and only if the language is regular. 
Let $\Psi_1$ be the MSO sentence obtained from $\Phi$ by replacing all subformulas of $\Phi$ of the form 
\begin{itemize}
\item $S(x,y)$ by a formula $\phi_S(x,y)$ that states that $x$ is the smallest element with respect to $<$, that $P_b(y)$ holds, and that there is no $z < y$ in $P_b$;
\item $T(x,y)$ by a formula $\phi_T(x,y)$ that states that $P_a(x)$ holds, that 
there is no $z > x$ in $P_a$, and that $y$ is the largest element with respect to $<$;
\item $R(x,y)$ by the formula $\phi_R(x,y)$ given by $x < y$;
\item $N(x,y)$ by a formula $\phi_N(x,y)$ stating that $y$ is the next element after $x$ w.r.t.~$<$. 
\end{itemize}
The resulting MSO sentence $\Psi_1$ has the signature
$\{P_a,P_b,<\}$.
Let $\Psi$ be the conjunction of $\Psi_1$ with the sentence $\Psi_2$ which states that 
for all $x,y \in A$, if $x<y$ and $P_a(y)$ then
$P_a(x)$. 

{\bf Claim.} If $\bA$ is a $\{<,P_a,P_b\}$-structure that represents a word $w_\bA \in \{a,b\}^*$, then $\bA \models \Psi$ if and only if $w_\bA$ is of the form $a^nb^n$ for some $n \geq 1$. 

Let $\bB$ be the $\{S,T,R,N\}$-structure with the same domain as $\bA$ such that for $X \in \{S,T,R,N\}$ we have $X^{\bB} := \{(x,y) \mid \bA \models \phi_X(x,y)\}$. 
See Figure~\ref{fig:leiter} for an example of a structure $\bA$ such that $w_{\bA} = a^4b^4$
and the corresponding $\{S,T,R,N\}$-structure $\bB$. 
Note that $\bB \models \Phi$ if and only if $\bA \models \Psi_1$. 
If $w_\bA$ is of the form $a^nb^n$ for some $n \geq 1$, then $\bA$ clearly satisfies $\Psi_2$. To show that it also satisfies $\Psi_1$, let $v_1,\dots,v_n,w_1,\dots,w_n \in A$ be such that $\{v_1,\dots,v_n\} = P_a^{\bA}$
and $\{w_1,\dots,w_n\} = P_b^{\bA}$ such that 
for all $i,j \in \{1,\dots,n\}$, if $i<j$ then $v_i<^{\bA} v_j$ and $w_i <^{\bA} w_j$. Then 
\begin{align}
(v_1,w_1) \in S^{\bB} & , \quad \quad (v_n,w_n) \in T^{\bB}, \nonumber \\
(v_i,w_i) \in R^{\bB} & \text{ for all } i \in \{2,\dots,n-1\}, \label{eq:leiter} \\
(v_i,v_{i+1}),(w_i,w_{i+1}) \in N^{\bB} & \text{ for all } i \in \{1,\dots,n-1\}. \nonumber 
\end{align} 
It follows that $\bB$ satisfies $\Phi$ and therefore $\bA \models \Psi_1$. 

For the converse direction, suppose that $\bA \models \Psi$. 
Clearly, $w_{\bA} \in a^*b^*$ because $\bA \models \Psi_2$. Moreover, since $\bA \models \Psi_1$ we have that $\bB \models \Phi$, 
and hence there exist $n \in {\mathbb N}$ and elements $v_1,\dots,v_n,w_1,\dots,w_n \in A$ such that 
$\bB$ satisfies (\ref{eq:leiter}). We first prove that $P_a^{\bA} = \{v_1,\dots,v_n\}$ and $|P_a^{\bA}|=n$. 
Since $(v_n,w_n) \in T^{\bB}$ we have $\phi_T(v_n,w_n)$ and hence $v_n \in P_a^{\bA}$. 
Since $\bB \models N(v_1,v_2),\dots,N(v_{n-1},v_n)$ we have that $v_1 < v_2 < \cdots < v_{n-1} < v_n$ holds in $\bA$ and it also follows that $|P_a^{\bA}| \geq n$. 
Then for every $i \in n$ we have that 
$v_i \in P_a^{\bA}$ because $v_i \leq v_n$, 
$v_n \in P_a^{\bA}$, and
$w_{\bA} \in a^*b^*$. 
Now suppose for contradiction that there exists $x \in P^{\bA}_a \setminus \{v_1,\dots,v_n\}$; choose $x$ largest with respect to $<^{\bA}$. Since $(v_n,w_n) \in T^{\bB}$ and $x \in P_a^{\bA}$ we must have $x \leq v_n$, and hence $x < v_n$ since $x \notin \{v_1,\dots,v_n\}$. 
Then there exists $y \in A$ such that $\phi_N(x,y)$ holds in $\bA$. 
Since $y \leq v_n$, $v_n \in P_a^{\bA}$,
and $w_{\bA} \in a^*b^*$, we must have $y \in P_a^{\bA}$. 
By the maximal choice of $x$ we get that $y=v_i$ 
for some $i \in \{1,\dots,n\}$. But then $\phi_N(x,v_i)$ implies that $x \in \{v_1,\dots,v_{n-1}\}$, a contradiction. 
Similarly, one can prove that $P^{\bA}_b = \{w_1,\dots,w_n\}$ and that $|P^{\bA}_b| = n$. 
This implies that $w_{\bA} = a^nb^n$.

We finally have to prove that ${\mathcal C}$ is in GSO. Let $\Phi$ be the GSO $\{S,T,R,N\}$ sentence with existentially quantified unary relations $V,W$, and existentially quantified binary relations $R' \subseteq R$ and $N' \subseteq N$,
which states that
\begin{itemize}
\item there are elements $v_1,v_n \in V$ and $w_1,w_n \in W$ such that $S(v_1,w_1)$ and $T(v_n,w_n)$ hold;
\item for every $x \in V \setminus \{v_1\}$ there is a unique element $y \in V \setminus \{v_n\}$ such that $N'(y,x)$ holds;
\item for every $x \in V \setminus \{v_n\}$ there is a unique element $y \in V \setminus \{v_1\}$ such that $N'(x,y)$ holds;
\item for every $x \in W \setminus \{w_1\}$ there is a unique element $y \in W \setminus \{w_n\}$ such that $N'(y,x)$ holds;
\item for every $x \in W \setminus \{w_n\}$ there is a unique element $y \in W \setminus \{w_1\}$ such that $N'(x,y)$ holds;
\item for all $v \in V$ and $w \in W$ we have that $N'(v_1,v) \wedge N'(w_1,w)$ implies $R'(v,w)$. 
\item for all $v,v' \in V \setminus \{v_1,v_n\}$ and $w,w' \in W \setminus \{w_1,w_n\}$ we have that $R'(v,w) \wedge N'(v,v') \wedge N'(w,w')$ implies $R'(v,w)$. 
\item For all $v \in V$ and $w \in W$ we have that $N'(v,v_n) \wedge N'(w,w_n)$ implies $R'(v,w)$. 
\end{itemize}
Then $\Phi$ holds on a finite $\{S,T,R,N\}$-structure $\bB$ if and only if $B$ has elements $v_1,\dots,v_n$, $w_1,\dots,w_n$ satisfying 
$(\ref{eq:leiter})$, which is the case if and only if $\bB \in {\mathcal C}$, which can be seen from the inductive computation of the fixed point by the Datalog program, see Remark~\ref{rem:op-dat}.
\end{proof}

Sometimes, we will also use the term GSO (MSO, Datalog) to denote all problems (i.e., all classes of structures) that can be expressed in the formalism. In particular, this justifies to say that (the complement of) a certain CSP is \emph{in} GSO (MSO, Datalog).


\section{Homomorphism-Closed GSO}
\label{sect:mgso}
In this section, we prove that the class of finite models of a GSO sentence is a finite union of CSPs of $\omega$-categorical structures whenever its complement is closed under homomorphisms.
In particular, every 
CSP in GSO (and therefore every CSP in MSO) is the CSP of an 
$\omega$-categorical structure $\bB$. 
This result will be essential for our proof of the existence of canonical Datalog programs for GSO in Section~\ref{sect:main}. 
In our proof, we use an equivalent  characterisation 
of the CSPs that can be formulated
as the CSP of an $\omega$-categorical structure from~\cite{BodHilsMartin}; this characterisation will be recalled in Section~\ref{sect:csps}. 

\begin{remark} 
This remark is addressed to readers familiar with the universal-algebraic approach to CSPs (see~\cite{Book} for an introduction), and can be skipped. Recall   
that whether the
CSP of an $\omega$-categorical structure $\bB$ is in Datalog only depends on the polymorphism clone of $\bB$ (see~\cite[Corollary 8.3.4]{Book} for an even stronger statement). 
Therefore, our result opens up the repertoire of methods from universal algebra for further studies of the intersection of GSO and Datalog. 
\end{remark} 

\subsection{CSPs for Countably Categorical Structures} 
\label{sect:csps}
Given permutation group $G$ on a set $B$, an \emph{orbit} of $G$
	is a set of the form $\{g(b) \mid g \in G\}$ for some $b \in B$. 
By the theorem of Ryll-Nardzewski, 
a countable structure 
 $\bB$ is $\omega$-categorical if and only if
 for every $n \in {\mathbb N}$ 
there are finitely many orbits of the componentwise action 
of the automorphism group $\Aut(\bB)$ of $\bB$ on $B^n$ (see, e.g.,~\cite{Hodges}). 
It follows that if a structure is $\omega$-categorical, then all of its reducts are $\omega$-categorical as well~\cite{Hodges}.

There is a known condition that characterises classes of structures that are CSPs of $\omega$-categorical structures. 
Let $\tau$ be a relational signature and let ${\mathcal C}$ be a class of finite $\tau$-structures. Let $\Lambda_n$ be the class of primitive positive $\tau$-formulas with 
free variables $x_1,\dots,x_n$ whose canonical
database is in ${\mathcal C}$. 
We define $\sim^{\mathcal C}_n$ to be the equivalence relation on $\Lambda_n$ such that 
$\phi_1 \sim^{\mathcal C}_n \phi_2$ holds if for
all primitive positive $\tau$-formulas $\psi(x_1,\dots,x_n)$ we have that
$\phi_1(x_1,\dots,x_n) \wedge \psi(x_1,\dots,x_n)$
is satisfiable in a structure from ${\mathcal C}$
if and only if 
$\phi_2(x_1,\dots,x_n) \wedge \psi(x_1,\dots,x_n)$
is satisfiable in a structure from ${\mathcal C}$.\footnote{A closely related concept was considered for graph classes in parametrised complexity~\cite[Definition 6.7.5]{DowneyFellows}; we thank Antoine Mottet for pointing this out to us.} 
The \emph{index} of an equivalence relation is the number of its equivalence classes.
If $\sim$ is an equivalence relation on a set $A$, and $a \in A$, then $[a]_{\sim}$ denotes the equivalence class of 
$a$; if the reference to $\sim$ is clear, we omit $\sim$ in this notation. 
 
 \begin{theorem}[Bodirsky, Hils, Martin~\cite{BodHilsMartin}, Theorem 4.27]\label{thm:bodhilsmartin}
 Let ${\mathcal C}$ be a constraint satisfaction problem. Then there is 
 an $\omega$-categorical structure $\bB$ such that
 ${\mathcal C} = \Csp(\bB)$ iff 
 $\sim^{\mathcal C}_n$ has finite index for all $n$. 
 Moreover, the structure $\bB$ can be chosen so that for all $n \in {\mathbb N}$ 
the orbits of the componentwise action 
of 
$\Aut(\bB)$ 
on $B^n$ are primitively positively definable in $\bB$. 
 \end{theorem}


\begin{example}
The structure $\bB_1 \ceq ({\mathbb Z};<)$ is not $\omega$-categorical. However, $\sim_n^{\Csp(\bB_1)}$ has finite index for all $n$, and indeed 
$\Csp({\mathbb Z};<) = \Csp({\mathbb Q};<)$ and $({\mathbb Q};<)$ is $\omega$-categorical. On the other hand, for
$\bB_2 \ceq ({\mathbb Z};\text{Succ})$ 
we have that the index 
$\sim_2^{\Csp(\bB_2)}$ is infinite, 
and it follows that there is no $\omega$-categorical structure $\bB$ such that $\Csp(\bB_2) = \Csp(\bB)$; see~\cite{Book}. 
\enex
\end{example}

\begin{remark}
Since we make essential use of Theorem~\ref{thm:bodhilsmartin}, it might be helpful for the reader to have the following idea about its proof. We use the equivalence relation $\sim_n^{\mathcal C}$ to expand the structures in ${\mathcal C}$ by new relations; for the resulting class we can 
use Fra\"iss\'e's construction of homogeneous structures (Definition~\ref{def:homog}). The assumption that $\sim_n^{\mathcal C}$ has finite index is the reason that the resulting structure is also $\omega$-categorical. 
By taking a reduct to the original signature, we might lose homogeneity, but keep $\omega$-categoricity, and obtain the desired structure $\bB$. 
\enrem
\end{remark} 

\subsection{Quantifier Rank}
In order to construct $\omega$-categorial structures for a given CSP in GSO, we need to verify the condition given in Theorem~\ref{thm:bodhilsmartin}. In this context, 
 it will be convenient to work with signatures that also contain constant symbols. 
The \emph{quantifier rank}
of a second-order $\tau$-formula $\Phi$ is
the maximal number of nested (first-order or second-order) quantifiers in $\Phi$; for this definition, we view $\Phi$ as a second-order sentence with guarded semantics, just as in~\cite{Blumensath-MSO}.  
If $\bA$ and $\bB$ are $\tau$-structures and $q \in {\mathbb N}$ we write $\bA \equiv^{\GSO}_q \bB$ if $\bA$ and $\bB$ satisfy the same GSO $\tau$-sentences of quantifier rank at most $q$. 
The following lemma is a natural extension of a classical result of Fr\"aiss\'e and Hintikka on first-order logic (see, e.g.,~\cite{Hodges}). 

\begin{lemma}[Proposition 3.4  in~\cite{Blumensath-MSO-2}]\label{lem:finite}
Let $q \in {\mathbb N}$ and $\tau$ be a finite signature with relation and constant symbols. Then $\equiv^{\GSO}_q$ is an equivalence relation with finite
index on the class of all finite $\tau$-structures. 
Moreover, every class of $\equiv^{\GSO}_q$ can 
be defined by a single GSO sentence with quantifier rank $q$. The analogous statements hold for MSO as well. 
\end{lemma} 

If $\bA$ is a $\tau$-structure and $\bar a$ is a $k$-tuple of elements of $A$, then 
 we write $(\bA,\bar a)$ for a $\tau \cup \{c_1,\dots,c_k\}$-structure expanding $\bA$ where $c_1,\dots,c_k$ denote fresh constant symbols being mapped to the corresponding entries of $\bar a$. 
If $\bA$ and $\bB$ are $\tau$-structures and $\bar a \in A^k$, $\bar b \in B^k$, and when writing 
 $(\bA,\bar a) \equiv^{\GSO}_q (\bB,\bar b)$, we implicitly assume that we have chosen the same constant symbols for $\bar a$ and for $\bar b$. The following lemma is a natural extension of classical results of Ehrenfeucht and Fr\"aiss\'e on first-order logic (see, e.g.,~\cite{Hodges}). 


\begin{lemma}[Proposition 3.6 in~\cite{Blumensath-MSO-2}]\label{lem:bnf}
Let $q \in {\mathbb N}$ and let $\bA$ and $\bB$ be $\tau$-structures. 
Then $\bA \equiv^{\GSO}_{q+1} \bB$ 
if and only if the following properties hold: 
\begin{itemize}
\item (first-order forth) For every $a \in A$, there exists $b \in B$ such that $(\bA,a) \equiv^{\GSO}_{q} (\bB,b)$. 
\item (first-order back) For every $b \in B$, there exists $a \in A$ such that $(\bA,a) \equiv^{\GSO}_{q} (\bB,b)$. 
\item (second-order forth) For every expansion $\bA'$ of $\bA$ by a guarded relation, there exists an expansion $\bB'$ of $\bB$ by a guarded relation  such that $\bA' \equiv^{\GSO}_{q} \bB'$. 
\item (second-order back) 
For every expansion $\bB'$ of $\bB$ by a guarded relation, there exists an expansion $\bA'$ of $\bA$ by a guarded relation  such that $\bA' \equiv^{\GSO}_{q} \bB'$. 
\end{itemize} 
\end{lemma} 

In the following, $\tau$ denotes a finite relational signature. 
Theorem~\ref{thm:disjoint-union} below states that the $\equiv_q^{\text{GSO}}$-equivalence class of the union of two structures which might overlap on some elements that are named by constants can be recovered from the equivalence classes of its components. 

\begin{definition}\label{def:Dn}
Let $\rho \ceq \{c_1,\dots,c_n\}$ be a finite set of constant symbols. 
Then $\mathcal D_n$ is defined to be the set of all pairs
$(\bA,\bB)$ of finite $(\tau \cup \rho)$-structures such that 
\begin{itemize}
\item $c^{\bA} = c^{\bB}$ for all constant symbols $c \in \rho$; 
\item $\{c_1^\bA,\dots,c_n^\bA\} = A \cap B = \{c_1^\bB,\dots,c_n^\bB\}$. 
\end{itemize}
Given a $(\bA,\bB) \in \mathcal D_n$, we  write $\bA \uplus \bB$ for the structure with domain $A \cup B$ such that $R^{\bA \uplus \bB} \ceq R^{\bA} \cup R^{\bB}$
for each relation symbol $R \in \tau$
and $c^{\bA \uplus \bB} = c^{\bA} = c^{\bB}$ for each constant symbol $c \in \rho$. 
\end{definition}

The special case of the following theorem for $n=0$ (i.e., for disjoint unions with no overlap) is Proposition 4.2 in~\cite{Blumensath-MSO-2}.

\begin{theorem}\label{thm:disjoint-union}
Let $q,n,r,s \in {\mathbb N}$, let $(\bA_1,\bB_1),(\bA_2,\bB_2) \in {\mathcal D}_n$, and let 
$\bar a_1 \in (A_1)^r$, $\bar a_2 \in (A_2)^r$, $\bar b_1 \in (B_1)^s$, $\bar b_2 \in (B_2)^s$ 
be such that 
 $(\bA_1,\bar a_1) \equiv^{\GSO}_q (\bA_2, \bar a_2)$ and 
$(\bB_1,\bar b_1) \equiv^{\GSO}_q (\bB_2, \bar b_2)$. Then $$(\bA_1 \uplus \bB_1, \bar a_1, \bar b_1)  \equiv^{\GSO}_q (\bA_2 \uplus \bB_2, \bar a_2, \bar b_2).$$ 
\end{theorem}
\begin{proof}
Our proof is by induction on $q$. 
Every quantifier-free formula is a Boolean combination of atomic formulas, so for $q = 0$ it suffices to consider atomic formulas $\phi$. By symmetry, it suffices to show that 
if $(\bA_1 \uplus \bB_1,\bar a_1,\bar b_1) \models \phi$ then $(\bA_2 \uplus \bB_2, \bar a_2, \bar b_2) \models \phi$. 
Then $\phi$ is built using a relation symbol $R \in \tau$, and the tuple that witnesses the truth of $\phi$ in $\bA_1 \uplus \bB_1$ must be from
$R^{\bA_1}$ or from $R^{\bB_1}$, by the definition of $\bA_1 \uplus \bB_1$. 
We first consider the former case; the
latter case can be treated similarly. 
 If a constant that appears in $\phi$ is
 from $A_1 \cap B_1$, then by the definition of ${\mathcal D}_n$ this element is denoted by a constant symbol $c \in \rho$,
 and therefore we may assume without loss of generality that $\phi$ is a formula over the signature of $(\bA_1,\bar a_1)$. 
Hence, $(\bA_1, \bar a_1) \models \phi$
and by assumption $(\bA_2, \bar a_2) \models \phi$. 
 This in turn implies that  $(\bA_2 \uplus \bB_2, \bar a_2, \bar b_2) \models \phi$.
 
 For the inductive step, suppose that the claim holds for $q$, and that 
 $(\bA_1,\bar a_1) \equiv^{\GSO}_{q+1} (\bA_2, \bar a_2)$ and $(\bB_1,\bar b_1) \equiv^{\GSO}_{q+1} (\bB_2, \bar b_2)$.
  By symmetry and Lemma~\ref{lem:bnf} 
it suffices to verify the properties (first-order forth) and (second-order forth).  Let $c_1 \in A_1 \cup B_1$. We may assume that $c_1 \in A_1$; the case that $c_1 \in B_1$ can be shown similarly. 
By Lemma~\ref{lem:bnf}, there exists $c_2 \in A_2$ such that
$(\bA_1,\bar a_1, c_1) \equiv^{\GSO}_{q} (\bA_2,\bar a_2, c_2)$. By the inductive assumption, this implies that $$(\bA_1 \uplus \bB_1,\bar a_1,c_1,\bar b_1) \equiv^{\GSO}_{q} (\bA_2 \uplus \bB_2, \bar a_2, c_2,\bar b_2)$$ and concludes the proof of (first-order forth). 

Now let $R$ be a guarded relation of $\bA_1 \uplus \bB_1$ of arity $k$. Let $\bA'_1$ be the expansion of $\bA_1$ by the guarded relation $R \cap A_1^k$,
and $\bB_1'$ be the expansion of $\bB_1$ by the guarded relation $R \cap B_1^k$. 
By Lemma~\ref{lem:bnf}
there are expansions $\bA_2'$ of $\bA$ and
$\bB_2'$ of $\bB_2$ by guarded relations 
such that 
$(\bA_1',\bar a_1) \equiv_q^{\GSO} (\bA_2',\bar a_2)$ and 
$(\bB_1' ,\bar b_1) \equiv_q^{\GSO} (\bB_2',\bar b_2)$. By the inductive assumption, this implies that  
$(\bA_1' \uplus \bB_1',\bar a_1, \bar b_1) \equiv_q^{\GSO} (\bA_2' \uplus \bB_2',\bar a_2,\bar b_2)$, which completes the proof of (second-order forth). 
\end{proof}

\begin{corollary}\label{cor:omega-cat}
Let ${\mathcal C}$ be a CSP that can be expressed in GSO. Then there exists a countable $\omega$-categorical structure $\bB$ such that ${\mathcal C} = \Csp(\bB)$. 
\end{corollary}
\begin{proof}
Let $\tau$ be the signature of ${\mathcal C}$, 
and let $\Phi$ be a GSO $\tau$-formula with quantifier rank $q$ 
such that ${\mathcal C} = \llbracket \Phi \rrbracket$. 
By Theorem~\ref{thm:bodhilsmartin} it suffices to show that the equivalence relation $\sim^{\mathcal C}_n$ has finite index for every $n \in {\mathbb N}$. 
Let $\rho \ceq \{c_1,\dots,c_n\}$ be a set of new constant symbols. 
By Lemma~\ref{lem:finite}, there exists an $m \in {\mathbb N}$ such that 
$\equiv_q^{\GSO}$ has $m$ equivalence classes
on $(\tau \cup \rho)$-structures. 
If $\phi(x_1,\dots,x_n)$ is a primitive positive $\tau$-formula, then define ${\mathfrak S}_{\phi}$ to be the following 
$(\tau \cup \rho)$-structure:
\begin{itemize} 
\item the elements
are the equivalence classes of the smallest equivalence relation on the variables of $\phi$ that contains all pairs $x,y$ such that $\phi$ contains the conjunct $x=y$, and 
\item
 $(C_1,\dots,C_n) \in R^{\bS}$ for $R \in \tau$ if and only if 
 there are $y_1 \in C_1,\dots,y_n \in C_{n}$ such that $R(y_1,\dots,y_n)$ is a conjunct of $\phi$;
 \item 
finally, we set $c_i^{\bS_\phi} \ceq [x_i]$ for all $i \in \{1,\dots,n\}$. 
\end{itemize} 

We claim that if $\bS_\phi \equiv_q^{\GSO} \bS_\psi$, then $\phi \sim_n^{\mathcal C} \psi$. 
Let $\theta(x_1,\dots,x_n)$ be a primitive positive $\tau$-formula; we may assume that the existentially quantified variables of $\theta$ are disjoint from the existentially quantified variables of $\phi$ and of $\psi$, so that $(\bS_{\phi},\bS_{\theta}), (\bS_\psi,\bS_{\theta}) \in {\mathcal D}_n$. 
Since $\bS_\phi \equiv_q^{\GSO} \bS_\psi$ and $\bS_\theta \equiv_q^{\GSO} \bS_\theta$, we have 
$\bS_\phi \uplus \bS_\theta \equiv_q^{\GSO} \bS_\psi \uplus \bS_\theta$ by Theorem~\ref{thm:disjoint-union}. 
Now suppose that $\phi \wedge \theta$ is satisfiable in a model of $\Phi$.
This is the case if and only if $\bS_\phi \uplus \bS_\theta$ satisfies $\Phi$, which in turn implies that $\bS_\psi \uplus \bS_\theta$ satisfies $\Phi$ since $\Phi$ has quantifier rank $q$. 
This in turn is the case 
 if and only if
$\psi \wedge \theta$ is satisfiable in a model of $\Phi$, which proves the claim.

The claim implies that $\sim_n^{\mathcal C}$ has at most $m$ equivalence classes, concluding the proof. 
\end{proof}

\begin{example}\label{expl:ext}
Let $\Phi$ be the following MSO sentence $\Phi$. 
\begin{align*}
\forall X \big (( \exists x. X(x)) \Rightarrow
 \exists x,y \in X \; \forall z \in X (\neg E(x,z) \vee \neg E(y,z)) \big )
\end{align*}
Note that a finite digraph does not satisfy $\Phi$ if it has a non-empty subgraph with the following extension property:
 for any $x,y \in X$ there exists $z \in X$ such that $E(x,z)$ and $E(y,z)$. 

\begin{figure}
\begin{center}
\includegraphics[scale=0.5]{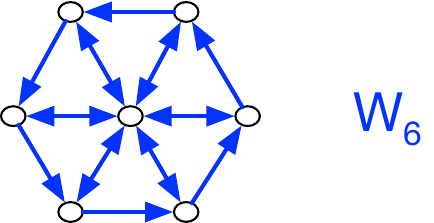}
\end{center}
\caption{An illustration of the digraph $W_6$ from Example~\ref{expl:ext}.}
\label{fig:W6}
\end{figure}

It is easy 
to see that $\llbracket \Phi \rrbracket$ is closed under disjoint unions and
that its complement is closed under homomorphisms. 
Corollary~\ref{cor:omega-cat} implies that there exists a countable $\omega$-categorical structure
with $\Csp(\bB) = \llbracket \Phi \rrbracket$.
\enex
\end{example}

\begin{remark} 
We would like to point out that
the existence of an $\omega$-categorical structure $\bB$ with the properties
from Example~\ref{expl:ext} does not simply follow from the theorem of Cherlin, Shelah, and Shi~\cite{CherlinShelahShi},
which states that for every finite set of finite connected structures $\mathcal F$ there exists an $\omega$-categorical structure $\bC$ such that a finite structure $\bA$ homomorphically maps to $\bC$ if and only if there is no homomorphism from a structure in $\mathcal F$ into $\bA$. 

Formally, we claim that there is no finite set $\mathcal F$ of finite structures ${\mathcal F}$ such that a finite structure $\bA$ has a homomorphism to $\bB$ if and only if no structure from ${\mathcal F}$ homomorphically maps to $\bA$. 
To see this, first note that the following graphs $W_i$, for $i \geq 2$, satisfy the mentioned extension property: the vertex set of $W_i$ is $\{0,1,\dots,i\}$, there is an edge from $p$ to $q$, for $p,q \in \{0,\dots,i-1\}$, if $p-q = 1 \mod i$, and additionally we have the edges $(p,i)$ and $(i,p)$ for every $p \in \{0,\dots,i-1\}$. See Figure~\ref{fig:W6} for an illustration of $W_6$. 
Note that for every $i \geq 2$ the graph $W_i$ is a core, and that they are pairwise homomorphically incomparable. Moreover, every subgraph of $W_i$ does not satisfy the extension property. 
\enrem
\end{remark} 


\subsection{Finite Unions of CSPs} 
In this section we prove that every class in GSO whose complement is closed under homomorphisms is a  
finite union of CSPs in GSO  (Theorem~\ref{thm:decomp}); the statement announced at the beginning of Section~\ref{sect:mgso} then follows (Corollary~\ref{cor:main}). 
Throughout this section, let $\tau$ be a relational signature and let ${\mathcal C}$ be a non-empty class of finite $\tau$-structures whose complement is closed under homomorphisms. In particular, ${\mathcal C}$ contains the structure $\bI$ with only one element where all relations are empty. 

\begin{definition}\label{def:sim}
Let $\sim$ be the equivalence relation defined on $\mathcal C$ by letting 
$\bA \sim \bB$ if for every $\bC \in \mathcal C$ we have $\bA \uplus \bC \in \mathcal C$ if and only if $\bB \uplus \bC \in \mathcal C$; here $\uplus$ denotes the usual disjoint union of structures, which is a special case of Definition~\ref{def:Dn} for $n=0$. 
\end{definition}

Note that the equivalence classes of $\sim$
are in one-to-one correspondence to  
the equivalence classes of $\sim_0^{\mathcal C}$. Also note that ${\mathcal C}$ is closed under disjoint unions if and only if 
$\sim$ has only one equivalence class.


If $\bA \in {\mathcal C}$, then we write
$[\bA]$ for the equivalence class of $\bA$ with respect to $\sim$. The following observations are immediate consequences from the definitions: 
\begin{enumerate}
\item 
each $\sim$-equivalence class is closed under homomorphic equivalence,\footnote{Two structures are said to be  \emph{homomorphically equivalent} if there are homomorphisms from $\bA$ to $\bB$ and from $\bB$ to $\bA$.}
\item each $\sim$-equivalence class is closed under disjoint unions. 
\item $\bA \in [\bI]$ if and only if 
$\bA \uplus \bB \in {\mathcal C}$ for all $\bB \in {\mathcal C}$. 
\end{enumerate}

\begin{lemma}\label{lem:general}
Let $\bA \in {\mathcal C}$ and 
let $\mathcal D$ be the smallest subclass of $\mathcal C$ that contains $[\bA]$ and whose complement is closed under homomorphisms.
Then 
\begin{enumerate}
\item $\mathcal D$ is a union of equivalence classes of $\sim$, and 
\item if $\sim$ has more than one equivalence class, then 
${\mathcal C} \setminus \mathcal D$ is non-empty. 
\end{enumerate}
\end{lemma}
\begin{proof}
Let $\bC \in [\bA]$, let $\bB$ be a finite structure with a homomorphism to $\bC$, and let $\bB' \in [\bB]$. Since $\bB \uplus \bC$ and $\bC$ are homomorphically equivalent, we have that $\bB \uplus \bC \sim \bC$. We claim that 
$\bB' \uplus \bC \sim \bC$. To see this, let $\bD \in {\mathcal C}$. Then 
\begin{align*}
 \bC \uplus \bD \in {\mathcal C} 
\Leftrightarrow \; & (\bB \uplus \bC) \uplus \bD \in {\mathcal C} && \text{(since $\bB \uplus \bC \sim \bC$)} \\
\Leftrightarrow \; & \bB \uplus (\bC \uplus \bD) \in {\mathcal C} \\
\Leftrightarrow \; & \bB' \uplus (\bC \uplus \bD) \in {\mathcal C} && \text{(since $\bB \sim \bB'$)} \\
\Leftrightarrow \; & (\bB' \uplus \bC) \uplus \bD \in {\mathcal C}
\end{align*}
which shows the claim. So $\bB' \uplus \bC \in [\bC] = [\bA]$. Since $\bB'$ has a homomorphism to $\bB' \uplus \bC$ we obtain that $\bB' \in {\mathcal D}$; this proves the first statement. 

To prove the second statement,
first observe that the statement is clear if $\bA \in [\bI]$, since the complement of $[\bI]$ is closed under homomorphisms. The statement therefore follows from the assumption that $\sim$ has more than one equivalence class. 
Otherwise, if $\bA \notin [\bI]$, then there exists a structure 
$\bB \in {\mathcal C}$ such that $\bA \uplus \bB \notin {\mathcal C}$. 
Then $\bB \in \mathcal C \setminus {\mathcal D}$ can be shown indirectly as follows:
otherwise $\bB$ would have a homomorphism to a structure $\bA' \in [\bA]$.
Since $\bB \uplus \bA'$ is homomorphically equivalent to $\bA'$, we have 
$\bB \uplus \bA' \sim \bA' \sim \bA$ 
and in particular $\bB \uplus \bA' \in {\mathcal C}$. 
But $\bB \uplus \bA' \in {\mathcal C}$ if
and only if $\bB \uplus \bA \in {\mathcal C}$ 
since $\bA \sim \bA'$. 
This is in contradiction to our assumption on $\bB$. 
\end{proof}

\begin{example} We consider a signature $\tau \ceq \{R_1,R_2,R_3\}$ of unary relation symbols. 
Define for every $i \in \{1,2,3\}$ the $\tau$-structure 
$\bS_i$ to be a one-element structure
where $R_i$ is non-empty and $R_j$, for $j \neq i$, is empty. Let
$${\mathcal C} \ceq \Csp(\bS_1 \uplus \bS_2) \cup \Csp(\bS_2 \uplus \bS_3) \cup \Csp(\bS_3 \uplus \bS_1).$$
Clearly, the complement of ${\mathcal C}$ is closed under homomorphisms. 
The equivalence classes of $\sim$ can be described as follows. For distinct $i,j \in \{1,2,3\}$, 
\begin{align*}
[\bS_i \uplus \bS_j] & = \Csp(\bS_i \uplus \bS_j) \setminus (\Csp(\bS_i) \cup \Csp(\bS_j)) \\
[\bS_i] & = \Csp(\bS_i) \setminus [\bI]  \\
[\bI] & = \Csp(\bI) . 
 \end{align*} 

\vspace{-.8cm} 
\enex
\end{example}

Recall that Lemma~\ref{lem:finite} asserts that
the equivalence relation 
$\equiv^{\GSO}_q$ on the class of finite $\tau$-structures has finitely many equivalence classes $\mathcal C_1,\dots,\mathcal C_m$, and that each of the equivalence classes $\mathcal C_i$ can be defined by a single GSO $\tau$-sentence 
$\Psi_i$ with quantifier rank $q$; we write $T^{\tau}_q \ceq \{\Psi_1,\dots,\Psi_m\}$ for this set of GSO sentences. 

\begin{definition} 
Let $\Phi$ be a GSO $\tau$-sentence  of quantifier rank $q$. 
Let $J\subseteq \{1,\ldots,m\}$ be such that $\{\Psi_j \in T^{\tau}_q  \mid j \in J\}$ is exactly the set of all sentences in $T^{\tau}_q$ that imply $\Phi$.
 Then $|J|$ is called the \emph{degree} of $\Phi$. 
 \end{definition}
 
 It is easy to see that the degree of $\Phi$ is exactly the index of $\equiv^{\GSO}_q$ restricted to $\llbracket \Phi \rrbracket$.

	\begin{lemma}\label{lem:link}
	Let $\Phi$ be a GSO $\tau$-sentence of quantifier rank $q$ 
	such that 	the complement of ${\mathcal C} \ceq \llbracket \Phi \rrbracket$ is closed under homomorphisms. 
	Let $\sim$ be the equivalence relation from Definition~\ref{def:sim} for ${\mathcal C}$.
		Then
		for every $\sim$-class $\mathcal{D}$
		there exists $I\subseteq \{1,\ldots,m\}$ such  that ${\mathcal D} = \bigcup_{i \in I} \llbracket \Psi_i \rrbracket$. 
	\end{lemma}
	\begin{proof}
		As in the proof of Corollary~\ref{cor:omega-cat}
		one can use Theorem~\ref{thm:disjoint-union} 
		to show for all finite $\tau$-structures $\bA,\bB$
		that 
		if $\bA \equiv_q^{\GSO} \bB$, then $\bA \sim \bB $. 
		This means that $\mathcal{D}$ is a union of  $\equiv_q^{\GSO} $-classes and therefore there exists $I \subseteq J\subseteq \{1,\ldots,m\}$ such that ${\mathcal D} = \bigcup_{i \in I} \llbracket \Psi_i \rrbracket$. 
\end{proof}

From this lemma we get immediately the following corollary.

	\begin{corollary}\label{cor: index bound}
	Let $\Phi$ and $\sim$ be as in Lemma~\ref{lem:link}. 
		Then the index of $\sim$ is smaller than or equal to the degree of $\Phi$.
\end{corollary}

\begin{theorem}\label{thm:decomp}
Let $\Phi$ be a GSO sentence.
If the complement of $\llbracket \Phi \rrbracket$ is closed under homomorphisms, then there are 
finitely many GSO $\tau$-sentences $\Phi_1,\dots,\Phi_t$ 
each of which describes a CSP 
such that $\Phi$ is equivalent to
$\Phi_1 \vee \cdots \vee \Phi_t$. 
If $\Phi$ is an MSO sentence,
then $\Phi_1,\dots,\Phi_t$ can be chosen to be MSO sentences as well. 
\end{theorem} 
\begin{proof}
	We prove the statement by induction on 
	the degree $n$ of $\Phi$. 
	By Corollary~\ref{cor: index bound} the equivalence relation $\sim$ 
		has at most $n$ equivalence classes 
		on $\tau$-structures. Hence, if $n=1$, 
	then $\llbracket \Phi \rrbracket$ is closed under disjoint unions, and we are done by Remark~\ref{rem:CSP}. 
	
	Let $\bA_1,\dots,\bA_s$ be $\tau$-structures such that $\{[\bA_1],\dots,[\bA_s]\}$ is the set of all equivalence classes of $\sim$ that are distinct from $[\bI]$ and contained in $\llbracket \Phi \rrbracket$. 
	Let $\mathcal D_i$ be the smallest subclass of 
	$\llbracket \Phi \rrbracket$ that contains $[\bA_i]$ and whose complement is closed under homomorphisms.
	Note that $\llbracket \Phi \rrbracket = \bigcup_{i \leq s} {\mathcal D}_i$ since $[\bI]$ is contained in ${\mathcal D}_i$ for all $i \leq s$. 
{By Lemma~\ref{lem:general} (1), each $\mathcal D_i$ is a union of $\sim$-classes which are themselves a union of $\equiv_q^{\GSO} $-classes by Lemma~\ref{lem:link}.
		It follows that there exists $I_i\subseteq \{1,\ldots,m\}$ such that ${\mathcal D}_i = \bigcup_{j \in I_i} \llbracket \Psi_j\rrbracket$. We define $\Phi_i:=\bigvee_{j \in I_i} \Psi_j
		$. Note that the GSO sentence $\Phi_i$ is of quantifier rank $q$ such that ${\mathcal D}_i = \llbracket \Phi_i \rrbracket$.} Hence, $\Phi$ is equivalent to $\bigvee_{i \leq s} \Phi_i$.  
	Lemma~\ref{lem:general} (2) asserts that 
	$\llbracket \Phi \rrbracket \setminus {\mathcal D}_i$ is non-empty, and hence the degree of $\Phi_i$ must be strictly smaller than $n$ for all $i\in \{1,\ldots,s\}$. The statement now follows from the inductive assumption. The same argument applies to MSO as well. 
\end{proof}

By taking complements,  
Theorem~\ref{thm:decomp} together with Corollary~\ref{cor:omega-cat} implies the following. 

\begin{corollary}\label{cor:main}
Every GSO sentence whose finite models are closed under homomorphisms is equivalent to a finite conjunction of GSO sentences each of which describes the complement of a CSP of a countable $\omega$-categorical structure. The analogous statement holds for MSO. 
\end{corollary}

At this point, it seems worthwhile to stress that this property of GSO does not extend to full SO: Not every homomorphism-closed class of structures that can be expressed in Second-order Logic is a finite intersection of complements of CSPs. 
We present an example of a class of finite $\tau$-structures that can even be expressed in Datalog 
but cannot be written in this form. 

\begin{example}\label{expl:datalog-completement}
Let $S$ and $T$ be unary, and let $R$ be a binary relation symbol. 
Let ${\mathcal C}$ be the class of all finite $\{S,T,R\}$-structures $\bA$ such that the following Datalog program $\Pi$ 
with the binary IDB $E$ derives $\goal$ on $\bA$. 
\begin{align*}
E(x,y) \; & {:}{-} \;  S(x),S(y) \\
 E(x,y) \; & {:}{-} \; E(x',y'),R(x',x),R(y',y) \\
 \goal \; & {:}{-} \; T(x),E(x,x'),R(x',y) 
\end{align*} 
For $n \in {\mathbb N}$, let $\bP_n$ be the $\{S,T,R\}$-structure on
the domain $\{1,\dots,n\}$ with
\begin{align*}
S^{\bP_n} & \ceq \{1\} & T^{\bP_n} & \ceq \{n\} & R^{\bP_n} & \ceq \big \{(i,i+1) \mid i \in \{1,\dots,n-1\} \big \} . 
\end{align*}
It is easy to see that each of the structures in $\{\bP_n \mid n \geq 1\}$ is not contained in $\mathcal C$,
and that the disjoint union of $\bP_i$ and $\bP_j$, for $i \neq j$, is contained in ${\mathcal C}$. 
It follows that $\mathcal C$ is not a finite intersection of complements of CSPs (and, by Corollary~\ref{cor:main}, cannot be expressed in GSO).
\enex
\end{example}

\section{Canonical Datalog Programs}
\label{sect:main}
A remarkable fact about the expressive power of Datalog
for constraint satisfaction problems over finite domains is the existence of \emph{canonical Datalog programs} and the characterisation of Datalog expressivity in terms of an existential pebble game~\cite{FederVardi}; this has been generalised to CSPs for $\omega$-categorical structures~\cite{BodDalJournal}. In this section,
we build on these results to characterise Datalog expressibility of GSO queries. 

\subsection{The Canonical Datalog Program of Width $(\ell,k)$}
A Datalog program $\Pi$ is called \emph{sound} for a class of $\tau$-structures ${\mathcal C}$ if 
 $\llbracket \Pi \rrbracket \subseteq {\mathcal C}$. It is called a 
 \emph{canonical Datalog program of width $(\ell,k)$ for ${\mathcal C}$}, for $\ell,k \in {\mathbb N}$, if   if 
 \begin{itemize}
 \item $\Pi$ is sound for ${\mathcal C}$, and 
\item $\llbracket  \Pi' \rrbracket  \subseteq \llbracket \Pi \rrbracket$ for 
every Datalog program $\Pi'$ of width $(\ell,k)$
which is sound for ${\mathcal C}$. 
\end{itemize}

The following definition is due to Feder and Vardi~\cite{FederVardi} for finite structures and has been generalised to $\omega$-categorical structures by Bodirsky and Dalmau~\cite{BodDalJournal}. 

\begin{definition}\label{def:can}
Let $\bB$ be an $\omega$-categorical structure with a finite relational signature $\tau$. Let $\bB'$ be the expansion of $\bB$ by all primitively positively definable relations of arity at most $\ell$ and let $\tau'$ be the (finite) signature of $\bB'$. Then \emph{the canonical $(\ell,k)$-Datalog program for $\bB$} has IDBs $\tau' \setminus \tau$ and EDBs $\tau$. The empty $0$-ary relation serves as $\goal$. There is a finite number of inequivalent formulas $\psi(\bar x,\bar y)$ of the form 
$$(\psi_1(\bar x,\bar y) \wedge \cdots \wedge \psi_j(\bar x,\bar y)) \Rightarrow R(\bar x)$$
having at most $k$ variables and where $\psi_1,\dots,\psi_j$ are atomic 
$\tau'$-formulas and 
$R \in \tau' \setminus \tau$. 
For each of the inequivalent formulas
$\psi(\bar x,\bar y)$ such that $\bB' \models \forall \bar x , \bar y. \psi(\bar x, \bar y)$ 
we introduce a rule 
$$R( \bar x) {:}{-} \psi_1(\bar x,\bar y),\dots,\psi_j(\bar x,\bar y) .$$
\end{definition}

\begin{theorem}[Bodirsky and Dalmau~\cite{BodDalJournal}, Theorem 4]\label{thm:BD}
Let $\bB$ be a countable $\omega$-categorical $\tau$-structure. 
Let $\ell,k \in {\mathbb N}$ and let $\Pi$ be 
the canonical Datalog program of width $(\ell,k)$ for $\bB$. 
Then for every finite $\tau$-structure $\bA$ the following are equivalent:
\begin{itemize}
\item $\Pi$ derives \goal\ on $\bA$;
\item Spoiler has a winning strategy for the existential $(\ell,k)$-pebble game on $(\bA,\bB)$. 
\end{itemize}
Moreover, $\Pi$ is a canonical Datalog program of width $(\ell,k)$ 
for the complement of $\Csp(\bB)$. 
\end{theorem}

These results extend 
to finite unions of complements of CSPs for $\omega$-categorical structures (and therefore
to all homomorphism-closed GSO-expressible classes of structures 
by Corollary~\ref{cor:main}) because of the following well-known fact. 

\begin{lemma}\label{lem:unions}
If ${\mathcal C}_1$ and ${\mathcal C}_2$
are in Datalog, then so are ${\mathcal C}_1 \cup {\mathcal C}_2$ and ${\mathcal C}_1 \cap {\mathcal C}_2$. If 
$\Pi_1$ and $\Pi_2$ are Datalog programs of width $(\ell,k)$, then there is a Datalog program $\Pi$
of width $(\ell,k)$ for $\llbracket \Pi_1 \rrbracket \cup \llbracket \Pi_2 \rrbracket$
and for $\llbracket \Pi_1 \rrbracket \cap \llbracket \Pi_2 \rrbracket$. 
\end{lemma}
\begin{proof}
For union, let $\Pi$ be obtained
by taking the union of the rules of $\Pi_1$ and of $\Pi_2$, possibly after renaming IDB predicate names to make them disjoint except for $\goal$. 
For intersection, 
we proceed similarly, but we first rename the symbol $\goal$ in $\Pi_1$ to $\goal_1$ and the symbol $\goal$ in $\Pi_2$ to $\goal_2$. Finally we add the new rule $\goal \; {:}{-} \; \goal_1, \goal_2$ to the union of $\Pi_1$ and $\Pi_2$. It is clear that these constructions preserve the width. 
\end{proof}

\begin{theorem}\label{thm:canonical}
Let $\Phi$ be a GSO sentence such that $\llbracket \Phi \rrbracket$ is closed under homomorphisms. 
Let $\ell,k \in {\mathbb N}$. Then there exists a canonical Datalog program 
$\Pi$ of width $(\ell,k)$ for $\llbracket \Phi \rrbracket$.  
\end{theorem}
\begin{proof}
By Corollary~\ref{cor:main} there are GSO sentences $\Phi_1,\dots,\Phi_m$ and $\omega$-categorical structures $\bB_1,\dots,\bB_m$ such that
$\Phi$ is equivalent to $\Phi_1 \wedge \cdots \wedge \Phi_m$ and $\llbracket \neg \Phi_i \rrbracket = \Csp(\bB_i)$. Let $\Pi_i$ be a canonical Datalog program for $\Csp(\bB_i)$ which exists by Theorem~\ref{thm:BD}. Then Lemma~\ref{lem:unions} implies that there exists a Datalog program $\Pi$ such that $\llbracket \Pi \rrbracket = \llbracket \Pi_1\rrbracket \cap \cdots \cap  \llbracket \Pi_m \rrbracket$. 
It is clear that $\Pi$ is sound for $\llbracket \Phi \rrbracket$. 
To see that $\Pi$ is a canonical Datalog program for $\llbracket \Phi \rrbracket$, suppose that $\bA$ is such that some Datalog program $\Pi'$ of width $(\ell,k)$ which is sound for $\llbracket \Phi \rrbracket$
derives \goal\ on $\bA$. 
Since, for every $i \in \{1,\dots,m\}$, the program $\Pi'$ is also sound for $\llbracket \Phi_i \rrbracket$, and $\Pi_i$ is a canonical Datalog program for $\llbracket \Phi_i \rrbracket$, 
the program $\Pi_i$ derives \goal\ on $\bA$. 
Hence, $\bA \in \llbracket \Pi \rrbracket  = \llbracket \Pi_1 \rrbracket \cap \cdots \cap  \llbracket \Pi_m \rrbracket$. 
\end{proof} 

If a GSO sentence is in Datalog, then 
whether or not the problem describes a CSP 
can be characterised by a syntactic description of the Datalog program (Proposition~\ref{prop:CSP-connected-datalog}).  
We say that a Datalog rule $\psi_0 {:}{-} \psi_1,\dots,\psi_m$
is \emph{connected} if 
$\{\psi_1,\dots,\psi_m\}$ cannot be partitioned into two non-empty subsets with disjoint sets of variables. 
The following holds for Datalog in general. 

\begin{proposition}\label{prop:connected-datalog-CSP}
Let $\Pi$ be a connected Datalog program. Then $\llbracket \Pi \rrbracket$ equals the complement of a CSP. 
\end{proposition}
\begin{proof}
By Remark~\ref{rem:CSP}, 
it suffices to show that the complement of $\llbracket \Pi \rrbracket$ is closed under disjoint unions. So let $\bA$ and $\bB$ be structures 
where $\Pi$ does not derive the goal predicate. 
Let $\bA'$ and $\bB'$ be the structure computed by $\Pi$ on $\bA$ and on $\bB$, respectively. 
If the body of a rule of $\Pi$ holds on 
$\bA' \uplus \bB'$, then all variables must take values from $\bA'$, or all variables must take values from $\bB'$, because the rule is connected. 
Hence, $\Pi$ does not derive new facts on $\bA' \uplus \bB'$, and in particular does not derive the goal predicate. Hence $\bA \uplus \bB$ belongs to the complement of $\llbracket \Pi \rrbracket$ we well. 
\end{proof}

In the specific situation that a Datalog program describes a problem in GSO, 
Proposition~\ref{prop:connected-datalog-CSP} has a converse. 

\begin{proposition}\label{prop:CSP-connected-datalog}
Suppose that $\Pi$ is a Datalog program such that 
$\llbracket \Pi \rrbracket$ is in GSO and equals the complement of a CSP. 
Then there exists a \emph{connected} Datalog program $\Pi'$ defining $\llbracket \Pi \rrbracket$. 
In fact, if $\bB$ is the $\omega$-categorical structure such that $\llbracket \Pi \rrbracket$ is the complement of $\Csp(\bB)$ (which exists by Corollary~\ref{cor:omega-cat}), and 
if $\Pi$ has width $(\ell,k)$, then such a program $\Pi'$ can be obtained from the canonical Datalog program of width $(\ell,k)$ for $\bB$ by removing all rules that are not connected. 
\end{proposition}
\begin{proof}
Let $\Pi$ be the canonical Datalog program of width $(\ell,k)$ for $\bB$ (Definition~\ref{def:can}). Suppose that $\Pi$ contains a rule which is not connected, e.g., a rule of the form
 $\psi_0 {:}{-} \psi_1,\psi_2$ where $\psi_1$ and $\psi_2$ are conjunctions of atomic formulas with disjoint sets of variables. We claim that then 
 $\psi_0 {:}{-} \psi_1$ or $\psi_0 {:}{-} \psi_2$ must be a rule of $\Pi$ as well. Otherwise,  for each $i \in \{1,2\}$ there would be a substructure of $\bA_i$ of $\bB$ that satisfies $\psi_i \wedge \neg \psi_0$. Then the disjoint union of $\bA_1$ and $\bA_2$ satisfies $\psi_1 \wedge \psi_2 \wedge \neg \psi_0$, a contradiction to the definition of the canonical Datalog program of $\bB$. 
\end{proof}

\subsection{The Existential Pebble Game}
\label{sect:pebble} 
The existence of canonical Datalog programs for GSO can be be shown via a characterisation of Datalog based on existential pebble games. 
An informal definition of the existential pebble game was given in the introduction; we now present a formal definition of the concept of \emph{(positional) winning strategies for Duplicator} in this game. 

\begin{definition}[\cite{BodDalJournal}]
\label{def:winningstrat}
A \emph{winning strategy for Duplicator} for the existential $(\ell,k)$-pebble game on a pair of relational $\tau$-structures $(\bA,\bB)$ is a non-empty set $\mathcal H$ of partial homomorphisms from $\bA$ to $\bB$ such that 
\begin{enumerate}
\item $\mathcal H$ is closed under restrictions;
\item for all $h \in {\mathcal H}$ with $|\dom(h)| = d \leq \ell$ and for all $a_1,\dots,a_{k-d} \in A$ there is an extension $h' \in {\mathcal H}$ of $h$ such that $h'$ is also defined on $a_1,\dots,a_{k-d}$. 
\end{enumerate}
\end{definition}

We say that \emph{Duplicator wins the existential $(\ell,k)$-pebble game for 
$\llbracket \Phi \rrbracket$ on $\bA$ if there exists a countable $\tau$-structure $\bB$ such that 
$\Csp(\bB) \cap \llbracket \Phi \rrbracket = \emptyset$ and Duplicator has a winning strategy on $(\bA,\bB)$ as in Definition~\ref{def:winningstrat}, and otherwise we say that \emph{Spoiler wins the game}.} 
In all our arguments about the existential pebble game, we only need this definition; in particular, we do not need the presentation of the existential pebble game as we have presented it in the introduction. 

\begin{theorem}\label{thm:main}
Let $\Phi$ be a GSO sentence.
Then $\llbracket \Phi \rrbracket$ can be defined in Datalog if and only if 
\begin{enumerate}
\item $\llbracket \Phi \rrbracket$ is closed under homomorphisms, and 
\item there exist $\ell,k \in {\mathbb N}$ such that
for all finite structures $\bA$,  
Spoiler wins the existential $(\ell,k)$-pebble game for $\llbracket \Phi \rrbracket$ on $\bA$ if and only if $\bA \models \Phi$. 
\end{enumerate}
\end{theorem}
\begin{proof}
If $\llbracket \Phi \rrbracket$ is in Datalog, then it is closed under homomorphisms. Hence, for both directions of the proof we may apply   
Corollary~\ref{cor:main} to $\Phi$, which implies that there are GSO sentences $\Phi_1,\dots,\Phi_m$ and $\omega$-categorical structures $\bB_1,\dots,\bB_m$ such that
$\Phi$ is equivalent to $\Phi_1 \wedge \cdots \wedge \Phi_m$ and $\llbracket \neg \Phi_i \rrbracket = \Csp(\bB_i)$.  
First suppose that $\llbracket \Phi \rrbracket$ is in Datalog. That is, there exist $\ell,k \in {\mathbb N}$ and a Datalog program $\Pi$ of width $(\ell,k)$ such that
$\llbracket \Phi \rrbracket = \llbracket \Pi \rrbracket$. 
Suppose that 
$\bA$ is a finite $\tau$-structure such that $\bA \models \Phi$. Then Spoiler wins the $(\ell,k)$-game as follows. Suppose that 
$\bB$ is a countable structure such that $\Csp(\bB) \cap \llbracket \Phi \rrbracket = \emptyset$.
Then 
$\Csp(\bB) \cap \llbracket \Phi_i \rrbracket = \emptyset$
for some $i \in \{1,\dots,m\}$; otherwise, 
if there is a structure $\bA_i \in \Csp(\bB) \cap \llbracket \Phi_i \rrbracket$ for every $i \in \{1,\dots,m\}$, then the disjoint union of $\bA_1,\dots,\bA_m$ satisfies $\Phi_i$ since $\Phi_i$ is closed under homomorphisms, and is in
$\Csp(\bB)$ since $\Csp(\bB)$ is closed under disjoint unions; but this is in contradiction to our assumption that $\Csp(\bB) \cap \llbracket \Phi \rrbracket = \emptyset$. 
Hence, $\Csp(\bB) \subseteq \Csp(\bB_i)$ and there is a homomorphism $h$ from $\bB$ to $\bB_i$ (see~\cite[Lemma 2]{BodDalJournal}). Note that $\Pi$ is sound for the complement of $\Csp(\bB_i)$, and $\Pi$ derives \goal\ on $\bA$, and thus Theorem~\ref{thm:BD} implies that Spoiler wins the existential $(\ell,k)$-pebble game on $(\bA,\bB_i)$. But since $\bB$ homomorphically maps to $\bB_i$, this implies that Spoiler wins the existential $(\ell,k)$-pebble game on $(\bA,\bB)$.

Now suppose that $\bA \models \neg \Phi$. 
Hence, there exists $i \in \{1,\dots,m\}$ such that $\bA \models \neg \Phi_i$. 
Then Duplicator wins the $(\ell,k)$-game as follows. She starts by playing $\bB_i$. Then $\bA$ homomorphically maps to $\bB_i$, and Duplicator can win the existential $(\ell,k)$-pebble game on $(\bA,\bB_i)$ by always playing along the homomorphism. 
 
For the converse implication, suppose that
1.~and 2.~hold. 
By Theorem~\ref{thm:BD}, for every $i \in \{1,\dots,m\}$ 
there exists a canonical 
Datalog program $\Pi_i$ of width
$(\ell,k)$ for $\llbracket \Phi_i \rrbracket$. 
Then Lemma~\ref{lem:unions} implies that there exists a Datalog program $\Pi$ such that $\llbracket \Pi \rrbracket = \llbracket \Pi_1 \rrbracket \cap \cdots \cap  \llbracket \Pi_m \rrbracket$. 
Since each $\Pi_i$ is sound for $\llbracket \Phi_i \rrbracket$, it follows that $\Pi$ is sound for
$\llbracket \Phi \rrbracket$. 
Hence, it suffices to show that if $\bA$ is a finite $\tau$-structure such that $\bA \models \Phi$, then $\Pi$ derives \goal\ on $\bA$. Since $\bA \models \Phi_i$ for all $i \in \{1,\dots,m\}$, 
the assumption implies that 
Spoiler wins the existential $(\ell,k)$-pebble game on $(\bA,\bB_i)$.  
By Theorem~\ref{thm:BD}, it follows that $\Pi_i$ derives \goal\ on $\bA$. 
Hence, $\Pi$ derives \goal\ on $\bA$. 
\end{proof} 

An example application of Theorem~\ref{thm:main} for showing that a certain MSO sentence cannot be expressed in Datalog can be found in the proof of Theorem~\ref{thm:appl}. 

\section{Nested Guarded Queries}
\label{sect:GQ}
In this section we revisit logics that express problems that are contained in both GSO and Datalog, and then prove that all of them fail to express all problems that lie in the intersection of GSO and Datalog.  
\emph{Frontier-guarded Datalog} is a fragment of Datalog that is contained in GSO.
\emph{Guarded queries (GQ)} and 
the more expressive \emph{Nested guarded queries (GQ$^+$)} are extensions of frontier-guarded Datalog that have been introduced by Bourhis, Kr\"otzsch, and Rudolph~\cite{BourhisKroetzschRudolph};  the definitions will be recalled below. 
GQ is also strictly more expressive than \emph{Nested Monadically Defined Queries (Nemodeq)}, \emph{Monadically defined queries (Modeq)}~\cite{RK2013}, 
and the recently introduced \emph{almost monadic queries}~\cite{CateDalmauOprsal}. Both GQ and GQ$^+$ are contained in the intersection of Datalog and GSO. We show that GQ$^+$ does not contain all queries from the intersection of Datalog and MSO (and hence, neither do GQ, frontier-guarded Datalog, Nemodeq, Modeq, and almost monadic Datalog):
there are CSPs 
whose complement is both in MSO and in Datalog, but not in GQ$^+$ (Corollary~\ref{cor:gqplus}). 

\subsection{Frontier-Guarded Datalog}
A rule of a Datalog program is called \emph{frontier-guarded} if all variables of the head appear in a single EDB atom in the rule body.%
\footnote{The term \emph{frontier-guarded} was originally introduced for existential rules (also referred to as tuple-generating dependencies or Datalog$^\pm$), an expressive extension of Datalog. Frontier-guarded Datalog is then simply the syntactic intersection of frontier-guarded existential rules and Datalog. The name \emph{guarded} is reserved for rules where all variables of the rule (and not just the variables of the head) appear in a single EDB atom in the rule body.} 
 A Datalog program is called \emph{frontier-guarded} if all its rules are frontier-guarded. 
Note that every monadic Datalog program is frontier-guarded. Also the Datalog program from the proof of Proposition~\ref{prop:not-mso} is frontier-guarded. 
 
\begin{proposition}\label{prop:gdatalog-gso}
Every problem in frontier-guarded Datalog is in GSO. 
\end{proposition}
\begin{proof}
Let $\Pi$ be a Datalog program. 
Let $\Phi_\Pi$ be the SO sentence obtained by 
\begin{itemize}
\item existentially quantifying over the IDBs of $\Pi$, 
\item replacing each rule $\psi \; {:}{-} \; \phi_1,\dots,\phi_m$ of $\Pi$ by the conjunct $\forall \bar x (\psi \vee \neg \phi_1 \vee \cdots \vee \neg \phi_m)$ of $\Phi$, 
\item additionally adding the conjunct $\neg \goal$ to $\Phi_\Pi$. 
\end{itemize}
Clearly, $\Pi$ derives \goal\ on a finite structure $\bA$ if and only if $\bA$ satisfies $\neg \Phi_\Pi$.

We will now show that, if $\Pi$ is frontier-guarded, then $\Phi_\Pi$ can be expressed by a GSO sentence. 
To this end, we establish that $\Phi_\Pi$ has the same meaning under standard and under guarded semantics, from which the desired result follows via \cref{prop:guarded-semantics-equiv}.
Toward a contradiction, assume some structure $\bA$ satisfies $\Phi_\Pi$ under standard but not under guarded semantics (as $\Phi_\Pi$ is preceded only by existentially quantified relation symbols, this is the only possibility).
Assume that instantiating the IDBs with the relations $R_1, \ldots, R_k$ over $A$ witnesses that $\bA$ satisfies $\Phi_\Pi$ under standard semantics.
Obtain now the guarded relations $R'_1, \ldots, R'_k$ by removing all unguarded tuples from  $R_1, \ldots, R_k$, respectively. 
The desired contradiction is now obtained by arguing that instantiating the IDBs with the relations $R'_1, \ldots, R'_k$ witnesses that $\bA$ satisfies $\Phi_\Pi$ under guarded semantics:
First, we note that $\neg \goal$ is still satisfied. Second, for any rule $\psi \; {:}{-} \; \phi_1,\dots,\phi_m$ of $\Pi$, fixing any variable assignment of $\bar x$ with elements of $A$, we obtain that the truth of $\psi \vee \neg \phi_1 \vee \cdots \vee \neg \phi_m$ under the instantiation with $R_1, \ldots, R_k$ implies its truth under the instantiation with $R'_1, \ldots, R'_k$: On the one hand, should any $\neg \phi_i$ be true, this is immediate (both for EDB and IDB atoms). On the other hand, in case all $\neg \phi_i$ are false, then so is the one $\neg \phi_j$ where $\phi_j=\phi_j(\bar y)$ is the rule's ``frontier guard'', i.e., it consists of an EDB atom where $\bar y$ contains at least all the variables of $\bar z$ from $\psi(\bar z) = \psi$. Truth of $\phi_j(\bar y)$ under the chosen assignment implies that the tuple assigned to $\bar y$ is guarded and therefore also the tuple assigned to $\bar z$ must be guarded. Yet, then, truth of $\psi(\bar z)$ under the two IDB instantiations coincides, hence $\psi(\bar z)$ must hold.
\end{proof}

However, not every problem in GSO which is in Datalog can be expressed by a frontier-guarded Datalog program. 
To prove this we need the following definition (in contrast to Definition~\ref{def:can}, we work with a single parameter, since this is sufficient for our purposes). 

\begin{definition}
Let $\bB$ be an $\omega$-categorical structure with finite relational signature $\tau$ and $k \in {\mathbb N}$. Let $s$ be the maximal arity of $\tau$. Then the \emph{canonical frontier-guarded Datalog program of width $k$ for $\bB$} is the subset of the canonical Datalog program of width $(k,k)$ for $\bB$ which contains all rules that are frontier-guarded. 
\end{definition}

If $\Pi$ is the canonical frontier-guarded Datalog program of width $k$ for $\bB$ then it follows from Proposition~\ref{prop:gdatalog} below that 
$\llbracket  \Pi' \rrbracket  \subseteq \llbracket \Pi \rrbracket$ for 
every frontier-guarded Datalog program $\Pi'$ of width $(k,k)$
which is sound for the complement of $\Csp(\bB)$. 
We need to adapt the existential pebble game on $\bA$ and $\bB$ to the guarded setting as well. In our informal definition of the existential pebble game, we only change the rules for Spoiler: we additionally require that when Spoiler removes pebbles from $\bA$, then
all the pebbles that remain on $\bA$ must be \emph{guarded}, i.e., 
there must exist a tuple in a relation of $\bA$ such that these pebbles must be from the entries of this tuple. 
The resulting game will be called the \emph{existential guarded $k$-pebble game}. 
Formally, we again work exclusively with 
the concept of a winning strategy ${\mathcal H}$ for Duplicator (similarly as in Definition~\ref{def:winningstrat}). 
A partial homomorphism from $\bA$ to $\bB$ is called \emph{guarded} if there exists a tuple in a relation of $\bA$ whose entries contain $\dom(h)$. 

\begin{definition}
Let $k \in {\mathbb N}$ and let $\tau$ be a finite relational signature. 
A \emph{winning strategy for Duplicator for the existential guarded $k$-pebble game on two relational $\tau$-structures $\bA,\bB$}
is a set $\mathcal H$ of partial homomorphisms from $\bA$ to $\bB$ such that ${\mathcal H}$ is closed under restrictions, and for 
every $S \subseteq A$ with $|S| \leq k$
and every guarded $h \in {\mathcal H}$ with $\dom(h) \subseteq S$, there 
is an extension $h' \in {\mathcal H}$ of $h$ with domain $S$.
\end{definition}

The following can be shown analogously to the proof of Theorem 4 in~\cite{BodDalJournal}. 
We omit the proof because all the ideas are 
present in the proof 
of a more interesting theorem that we present in full detail in Section~\ref{sect:multi} (Theorem~\ref{thm:multi}). 

\begin{proposition}\label{prop:gdatalog}
Let $\bB$ be $\omega$-categorical with finite relational signature $\tau$ 
and let $\bA$ be a finite $\tau$-structure. 
Then for all $k \in {\mathbb N}$ the following statements are equivalent. 
\begin{itemize}
\item Every sound frontier-guarded Datalog program of width $k$ for the complement of $\Csp(\bB)$ does not derive \goal\ on $\bA$. 
\item The canonical frontier-guarded Datalog program of width $k$ for $\bB$   
does not derive \goal\ on $\bA$.
\item Duplicator has a winning strategy for the existential guarded $k$ pebble game on $\bA,\bB$. 
\end{itemize}
\end{proposition}

Recall from Example~\ref{expl:Q} that
$\Csp({\mathbb Q};<)$ is expressible in MSO; it is easy to see that its complement can be defined by a Datalog program (of width (2,3)). 

\begin{proposition}\label{prop:Q}
The complement of $\Csp({\mathbb Q};<)$ cannot be defined by a frontier-guarded Datalog program. 
\end{proposition}
\begin{proof}
Let $k \in {\mathbb N}$. 
By Proposition~\ref{prop:gdatalog}, it suffices to show that there exists a structure $\bA$ which has no homomorphism to $({\mathbb Q};<)$ but 
Duplicator has a winning strategy for the existential guarded $k$ pebble game on $\bA,({\mathbb Q};<)$. 
Let $\bA$ be a directed cycle of length at least $k+1$. Let ${\mathcal H}$ be the set of \emph{all} partial homomorphisms $h$ from $\bA$ to $({\mathbb Q};<)$ with a domain of size at most $k$.
We claim that $\mathcal H$ 
is a winning strategy for the existential guarded 
$(\ell,k)$-pebble game.  
Clearly, ${\mathcal H}$ is closed under restrictions. 
Now, let $h$ be a partial homomorphism whose domain is contained in a tuple from a relation of $\bA$, that is, contained in $\{a_1,a_2\}$ for two subsequent elements of $A$ on the cycle. 
Let $a_3,\dots,a_k$ be elements of $A$. Since $|A| \geq k+1$ there must be some $b \in A \setminus \{a_1,\dots,a_k\}$. Let $\prec$ be the order on $a_1,\dots,a_k$ in which these elements are traversed on the cycle, starting with $b$. 
Note that $\prec$ is a linear extension of $<^{\bA}$. 
We may then extend $h$ to a partial homomorphism $h'$ 
from $\bA$ to $({\mathbb Q};<)$ such that
$h(a_i) < h(a_j)$ if $a_i \prec a_j$, for all $i,j \in \{1,\dots,k\}$. Note that
$h' \in {\mathcal H}$; this finishes the proof. 
\end{proof} 

\subsection{Flag and Check}
A general mechanism, called \emph{flag and check}, to turn a query language into a potentially more powerful query language has been introduced by Bourhis, Kr\"otzsch, and Rudolph~\cite{BourhisKroetzschRudolph}, generalising earlier work of Rudolph and Kr\"otzsch~\cite{RK2013}. 
The idea is best described for some fundamental constraint satisfaction problems. It is easy to see 
that there is no $k$ such that $\Csp(\{0,1\};\neq)$ (i.e., graph 2-colorability) is defined by a Datalog program of width $(1,k)$. However, to decide whether a graph is 2-colorable, it suffices to show that there is no vertex $v$ which can reach $v$ via a path of odd length. Note that such a  
\emph{`check'} can be performed by a Datalog program of width $(1,2)$ once a \emph{`flag'} has been put on $v$. 
Similarly, there is no $k \in {\mathbb N}$ such that $\Csp({\mathbb Q};<)$ (i.e., digraph acyclicity) can be defined by a Datalog program of width $(1,k)$~\cite{Bodirsky}. However, to decide whether a graph contains a directed cycle, it suffices to find a vertex $v$ which can reach $v$ by a directed path with at least one edge. Again, after $v$ has been found and \emph{`flagged'} such a computation can be performed by a Datalog program of width $(1,2)$. 

We now present the formal definition of flag-and-check programs
of Bourhis, Kr\"otzsch, and Rudolph~\cite{BourhisKroetzschRudolph}. 

\begin{definition}\label{def:fnc}
Let $\tau$ be a finite set of relation and constant symbols. 
A \emph{flag-and-check $\tau$-program ($\tau$-FCP) of arity $m$} is a set of Datalog rules $\Pi$ with EDBs from $\tau$ and the IDBs 
$\{\goal,P_1,\dots,P_k\}$ where \goal\ is a distinguished predicate of arity $0$.
Moreover, on top of constants from $\tau$, $\Pi$ may use extra distinguished constant symbols $\lambda_1,\dots,\lambda_m$. 
If $\bA$ is a $\tau$-structure, then 
$\Pi^{\bA}$ is the set of all tuples $(a_1,\dots,a_m) \in A^m$ such that $\Pi$ derives \goal\ on the
$\tau \cup \{\lambda_1,\dots,\lambda_m\}$-expansion $\bA'$ of $\bA$ where $\lambda^{\bA'}_i = a_i$ for all $i \in \{1,\dots,m\}$. 
\end{definition}

Several formalisms in the literature 
are based on flag-and-check programs:
\begin{itemize}
\item A \emph{guarded $\tau$-query (GQ)} is a query with free variables $y_1,\dots,y_k$, for some $k \in {\mathbb N}$, 
of the form $\exists x_1,\dots,x_{\ell}. \Pi(\bar x, \bar y)$ where $\Pi$
is a frontier-guarded $\tau$-FCP of arity $m = k+\ell$, with the obvious semantics~\cite{BourhisKroetzschRudolph}. 
\item A \emph{monadically defined query (Modeq)} is  a query with free variables $y_1,\dots,y_k$, for some $k \in {\mathbb N}$, of the form $\exists x_1,\dots,x_{\ell}. \Pi(\bar x, \bar y)$ where 
$\Pi$ is a \emph{monadic $\tau$-FCP} of arity $m = k + \ell$, i.e., all IDBs of $\Pi$ have arity at most one~\cite{RK2013} (we follow the presentation in~\cite{BourhisKroetzschRudolph}). 
\end{itemize}

We omit the reference to $\tau$ in the notation if the reference is clear from the context. 
It has been shown that every problem in Modeq can be expressed in MSO~\cite{RK2013}. 
Similarly, one can show that GQ is contained in GSO. This result is implicit in~\cite{BourhisKroetzschRudolph}, which is why we present the proof here for the convenience of the reader.

\begin{proposition}\label{prop:GQ-GSO}
Every problem in GQ is in GSO. 
\end{proposition}
\begin{proof}
Let $k,\ell \in {\mathbb N}$ and 
let $\Pi$ be a guarded flag-and-check program of arity $k+\ell$. 
We have to find a GSO sentence which is equivalent to $\exists x_1,\dots,x_{\ell}. \Pi(\bar x,\bar y)$.  
Let $\Phi_{\Pi}$ be the GSO sentence
that is equivalent to $\Pi$ (Proposition~\ref{prop:gdatalog-gso}). 
Then we define $\Psi$
to be the GSO sentence obtained from $\Phi_{\Pi}$ as follows. 
\begin{itemize}
\item Replace each occurrence of $\lambda_i$, for $i \in \{1,\dots,m\}$, by a fresh variable $x_i$. 
\item Existentially quantify  $x_1,\dots,x_{\ell}$. \qedhere
\end{itemize}
\end{proof}

An example of a query that can be expressed in GQ, but not in Modeq can be found in~\cite{BourhisKroetzschRudolph} (Example 2). 
An example of a Datalog query in MSO which cannot be expressed in Modeq 
has already been found in~\cite{RK2013} (Example 6); in fact, the same query is not even expressible in GQ, as we will show in Section~\ref{sect:multi} (Example~\ref{prop:inexpr}), so we present it in more detail in the following. 

\begin{example}\label{expl:MSO-nemodeq}
Let $\tau$ be the  EDB  signature consisting  of the binary relation symbols $C$, $L$, and $D$. 
Let $\Phi_1(x,y)$ be the MSO $\tau$-formula
\begin{align*}
	\forall z \exists U_2,U_3 \forall u,v \big (  
	(C(u,z) & \Rightarrow U_3(u)) \\
	\wedge \quad (U_3(v) \wedge C(u,v) & \Rightarrow U_3(u)) \\
	\wedge \quad  (L(x,u) & \Rightarrow U_2(u)) \\
	\wedge \quad  (U_2(u) \wedge U_3(u) \wedge L(u,v) & \Rightarrow U_2(v)) 
	\wedge \; \neg U_2(y)  
	\big );
\end{align*}
Figure~\ref{fig:CLD} depicts  a $\tau$-structure that satisfies $\neg \Phi_1(x,y)$ for the elements of the structure labelled with $x$ and $y$.  
\begin{figure}
\begin{center}
\includegraphics[scale=.4]{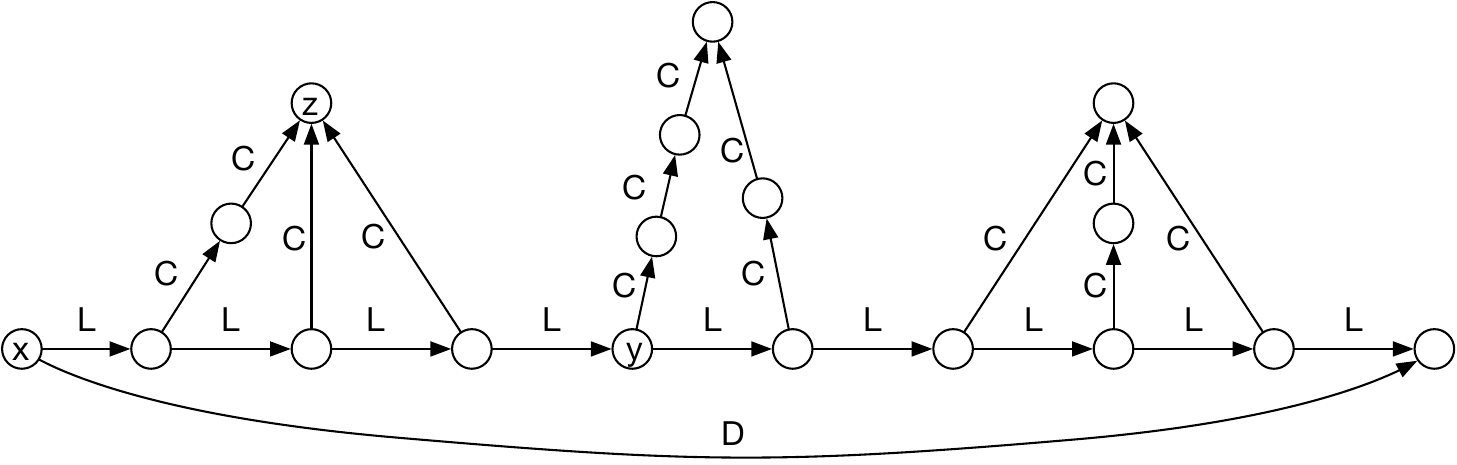}
\end{center}
\caption{An illustration of a structure that satisfies the sentence $\Psi$ in Example~\ref{expl:MSO-nemodeq}.} 
\label{fig:CLD}
\end{figure}
Let $\Phi_2(u,v)$ be the MSO formula 
\begin{align*}
\exists U_1 \forall x,y \big ( U_1(u) 
\wedge (U_1(x) \wedge \neg \Phi_1(x,y) \Rightarrow U_1(y) ) 
\wedge \neg (U_1(x) \wedge \neg \Phi_1(x,v)) 
\big ) . 
\end{align*}
Finally, let $\Psi$ be the MSO sentence
$\exists u,v \big (D(u,v) \wedge \neg \Phi_2(u,v) \big)$. 
See Figure~\ref{fig:CLD} for an illustration of a $\tau$-structure that satisfies $\Psi$. 
It follows from the results of Rudolph and Kr\"otzsch~\cite{RK2013} (Example~\ref{expl:ctd}) that $\Psi$ can be expressed by a Datalog program. \enex
\end{example}

\subsection{The Expressive Power of Modeqs} 
The Modeq formalism generalises natural concepts that have been discovered independently in the constraint satisfaction literature; this section can be skipped by impatient readers, and consulted by those interested in the relationship to constraint satisfaction. 

Certain flag-and-check programs have been studied in the context of the complexity of constraint satisfaction under the name \emph{peek arc consistency (PAC)}~\cite{DBLP:journals/tcs/BodirskyC10}, extending the famous (hyper-) arc consistency procedure (AC; see, e.g., \cite[Section 8.4]{Book}). In the following, $\tau$ denotes a finite relational signature. 

\begin{definition}
Let $\bB$ be a countable $\omega$-categorical $\tau$-structure with the orbits
$O_1,\dots,O_n$ and let $b_1,\dots,b_n$ be representatives from these orbits. 
Then the PAC procedure for $\Csp(\bB)$ derives  \goal\  on a $\tau$-structure $\bA$ if there exists an $a \in A$ such that for all $i \leq n$ the hyperarc consistency procedure for $\Csp(\bB,\{b_i\})$ derives  \goal\  on $(\bA,\{a\})$. 
\end{definition}
 
 For example, the PAC procedure for $\Csp(\{0,1\};\neq)$ derives 
 \goal\ on a finite structure $\bA$ if and only if
 $\bA$ has no homomomorphism to $(\{0,1\};\neq)$; similarly, the PAC procedure solves
$\Csp({\mathbb Q};<)$. 
 We prove that the PAC procedure can be expressed as a Modeq. 

\begin{lemma}
Let $\bB$ be a countable $\omega$-categorical $\tau$-structure with maximal arity $p$. Then there exists a Boolean Modeq 
$\Theta$ of width $(1,p)$
such that $\llbracket \Theta \rrbracket$ 
equals the class of all finite $\tau$-structures
where the PAC procedure for $\Csp(\bB)$ derives  
\goal. 
\end{lemma}

\begin{proof}
By the theorem of Ryll-Nardzewski, the automorphism group of $\bB$
has finitely many orbits $O_1,\dots,O_n$;
choose representatives $b_1,\dots,b_n$ from these orbits. 
For every $i \leq n$, let $\Pi_i$ be the Datalog program for the hyperarc consistency procedure for $\Csp(\bB,\{b_i\})$ (which has width $(1,p)$; see Section 8.4 in~\cite{Book}). 
Let $C_i$ be the unary IDB of $\Pi_i$ for the relation $\{b_i\}$, and let $\Pi'_i$ be the Datalog program obtained from $\Pi_i$ by adding the rule 
$C_i(\lambda_1) \; {:}{-}$ (without precondition).
Since expressibility by Datalog programs of width $(1,p)$ is closed under intersection (Lemma~\ref{lem:unions}), there is 
a Datalog program $\Pi$ of width $(1,p)$  expressing
$\Pi'_1 \wedge \cdots \wedge \Pi'_n$.
Also note that $\Theta := \exists x. \Pi$ holds on a $\tau$-structure $\bA$
if and only if there exists an element $a \in A$ such that 
for every $i \leq n$ the program $\Pi'_i$ derives
\goal\ on $(\bA,\{a\})$, which by definition is the case if and only if the PAC procedure derives \goal\  on $\bA$. This shows the statement. 
\end{proof}

The following example demonstrates that the expressivity of Modeq extends well beyond that of PAC. 

\begin{example}
The following is based on an example from~\cite{ACandFriends} of a structure such that the complement of the CSP of this structure cannot be solved by PAC, but can be solved by \emph{singleton arc consistency (SAC)}.\footnote{We do not need the notion of singleton arc consistency in this article and therefore refrain from giving a full definition and refer the interested reader to~\cite{ACandFriends}.} Below we will show that the mentioned class is even in Modeq. 
Let $\bB$ be the structure with domain
$\{-1,0,1,2\}$ and relations
\begin{align*}
R^{\bB} & \coloneqq \{-1,0,1,2\}^2 \setminus \{(-1,-1)\} \\
S^{\bB} & \coloneqq \{(0,1),(1,2),(2,0),(-1,-1)\}
\end{align*}
and let $\bA$ be the structure
with the domain $\{a,b,c,a',b',c'\}$ 
and the relations 
\begin{align*}
R^{\bA} & \coloneqq \{(a,a')\} \\
S^{\bA} & \coloneqq \{(a,b),(b,c),(a,c),(a',b'),(b',c'),(a',c')\}
\end{align*}
Then there is no homomorphism from $\bA$ to $\bB$, but the PAC procedure for $\Csp(\bB)$ does not derive \goal on $\bA$~\cite{ACandFriends}. 
However, the problem can be defined by the Modeq $\exists x_1,x_2.\, \Pi$ where $\Pi$ is the following monadic flag-and-check program of arity two. 
\begin{align*}
U_0(\lambda_1) & \; {:}{-}  & U_0(\lambda_2) & \; {:}{-}  \\
U_1(z) & \; {:}{-} \;  U_0(y) , S(y,z) 
& U_2(y) & \; {:}{-} \; U_0(z) , S(y,z) \\
U_2(z) & \; {:}{-} \;   U_1(y) , S(y,z) 
& U_0(y) & \; {:}{-} \; U_1(z) , S(y,z) \\
U_0(z) & \; {:}{-} \;  U_2(y) , S(y,z) 
& U_1(y) & \; {:}{-} \; U_2(z) , S(y,z) \\
\goal & \; {:}{-} \;   U_0(y), U_1(y), R(y,z), U_0(z), U_1(z)
\end{align*}
{\bf Claim.} $\exists x_1,x_2.\, \Pi$
 evaluates to true on a finite $\{R,S\}$-structure $\bA$ if and only if there is no homomorphism from $\bA$ to $\bB$. First 
suppose that there are $a_1,a_2 \in A$ such that 
$\Pi$ derives \goal\ on the 
$\{\lambda_1,\lambda_2\}$-expansion $\bA'$ of $\bA$ where $\lambda_i^{\bA_1} = a_i$ for $i \in \{1,2\}$. This means that there are $b_1,b_2 \in A$ such that $(b_1,b_2) \in R^{\bA}$,
and there are paths of net length 0 and 1 from $a_1$ to $b_1$,
 and paths of net length 0 and 1 from $a_2$ to $b_2$. In general, the existence of paths of net length $0$ and $1$ between two vertices implies that both have to be mapped to $-1$. Therefore, any homomorphism from $\bA$ to $\bB$ must 
map both $b_1$ and $b_2$ to $-1$, which is impossible
since $(b_1,b_2) \in R^{\bA}$. 

Now suppose that for any $a_1,a_2 \in A$ the program $\Pi$ does not derive \goal\  on the structure $\bA'$ defined above. That is, for every $(b_1,b_2) \in R^{\bA}$ there exists $i \in \{1,2\}$ such that the for all $u \in A$ that are connected to $b_i$ in the graph defined by $S^{\bA}$ there is no path of net length $0$ and no path of net length $1$ to $b_i$. Hence, there must exist a path of net length $2$ from $u$ to $b_i$. We then define $h(u) := b_i + 2 \mod 3$. For all other $u \in A$, we define $h(u) := -1$. It is straightforward to verify that $h$ is a homomorphism from $\bA$ to $\bB$. 
\enex
\end{example}

\subsection{Nested Queries}
The expressive power of Datalog fragments can sometimes be increased by considering \emph{nesting}, 
which is a concept introduced in~\cite{RK2013} and further studied in~\cite{BourhisKroetzschRudolph}. The idea is that  in nested queries we allow the use of other nested queries as if they were atomic formulas. 
Depending on the query language we start with, nesting may or may not increase the expressive power. Nested Datalog programs, for example, can be rewritten into equivalent Datalog programs without nesting. Similarly, MSO and GSO are closed under nesting. 
On the other hand, nested monadic queries (\emph{Nemodeq})
are strictly more expressive than monadic queries~\cite{RK2013}, and, similarly, 
nested guarded queries (\emph{GQ$^+$}) are strictly more expressive than guarded queries, as we will see below (Proposition~\ref{prop:inexpr}). 
We mention that the \emph{layered tree programs} studied in~\cite{CarvalhoDalmauKrokhin} are
in fact the queries obtained from linear monadic Datalog programs with at most one EDB per rule via nesting.

\begin{definition}[\cite{BourhisKroetzschRudolph}]
\label{def:nested}
Let $q,m \in {\mathbb N}$ and let $\tau$ be a finite relational signature. 
A \emph{$q$-nested (monadic, guarded) flag-and-check $\tau$-program of arity $m$} is defined inductively as follows. A $1$-nested (monadic, guarded) $\tau$-FCP of arity $m$ is the same as a (monadic, guarded) $\tau$-FCP of arity $m$ as defined in Definition~\ref{def:fnc}. A $q+1$-nested (monadic, guarded) $\tau$-FCP of arity $m$ is a (monadic, guarded) $\tau$-FCP $\Pi$ of arity $m$ that
may use in rule bodies $q$-nested (monadic, guarded) $\tau$-FCPs in addition to the symbols of $\tau$. 
These queries are called the \emph{subqueries of $\Pi$ at nesting level $q$};
the subqueries of $\Pi$ at nesting level $j$ are
inductively defined as the subqueries of the subqueries at nesting level $j+1$. 
In the case of guarded $q+1$-nested FCPs, the $q$-nested FCPs may not serve as guards. 
An FCP has \emph{width $m$} if all its subqueries have arity at most $m$. 

A \emph{($q$-nested, monadic, guarded, width $m$) $\tau$-query} is a query of the form $\exists \bar x. \Pi$ where $\Pi$ is a ($q$-nested, monadic, guarded, width $m$) $\tau$-FCP. 
\end{definition}

\begin{remark}\label{rem:disjunction} 
Note that all the classes of queries defined in Definition~\ref{def:nested} are closed under disjunction by appropriately taking the union of
the rule sets, just as in the proof of Lemma~\ref{lem:unions} for Datalog. 
\enrem
\end{remark}

The class of nested guarded queries is denoted by 
GQ$^+$, and the class of nested monadically defined queries by Nemodeq. It is known that every nested monadically defined query is equivalent to a Datalog program~\cite{RK2013}; the same is even true
for every nested guarded query~\cite[Figure 1]{BourhisKroetzschRudolph}. The following fact is implicit in~\cite{BourhisKroetzschRudolph}. 

\begin{theorem}\label{thm:gq-to-datalog}
Let $\tau$ be a finite relational signature with maximal arity $s$. 
Let $\Pi$ be a $q$-nested guarded FCP of arity $m$ and EDBs $\tau$ such that all rules have at most $k$ variables. 
Then there exists a 
Datalog program $\Pi'$ of width $(m+s,m+k)$ 
with a distinguished IDB $P_{\Pi}$ of arity $m$ 
such that for every $\tau$-structure $\bA$,
a tuple ${\bar a} \in A^m$ satisfies
$\Pi(\bar a)$
if and only if $\bar a \in P_{\Pi}^{\Pi'(\bA)}$.
The program $\Pi'$ can be computed in polynomial time from $\Pi$. 
\end{theorem}
\begin{proof}
The proof follows closely the corresponding statement for Nemodeq~\cite{RK2013}. For $q=1$, 
the Datalog program $\Pi'$ contains for each rule in $\Pi$ a new rule obtained
by 
\begin{itemize}
\item replacing each constant $\lambda_i$ by a new variable $x_i$, 
\item replacing each occurrence of $\goal$ by $P_{\Pi}(x_{1},\dots,x_{m})$
where $P_{\Pi}$ is a new IDB of arity $m$;
\item replacing each atomic formula $R(\bar y)$ by a new atomic formula $R'(x_1,\dots,x_{m},\bar y)$ where $R'$ is a new IDB.
\end{itemize}
For $q > 1$, 
the translation is defined recursively: 
let $\Pi_1,\dots,\Pi_s$ be the subqueries
of $\Pi$ at nesting level $q-1$. Let $\Pi'_1,\dots,\Pi'_s$ be the Datalog programs that can be associated to these FCPs by the inductive assumption; we may assume that the IDBs for all these programs are pairwise disjoint. Then $\Pi'$ is defined to be the union over the rules of $\Pi_1,\dots,\Pi_s$ together 
with the rules of $\Pi$ where we replace
each subquery $\Pi_i$ at nesting level $q$ by $P_{\Pi_i}$. 

To verify that the Datalog program $\Pi'$ defined in this way satisfies the required properties, let $\bA$ a $\tau$-structure and suppose that $\bar a \in A^m$ satisfies $\Pi(\bar a)$. 
This means that $\Pi$ derives $\goal$ 
on the $\tau \cup \{\lambda_1,\dots,\lambda_m\}$-expansion $\bA'$ of $\bA$ where $\lambda_i^{\bA'} = a_i$ for all $i \in \{1,\dots,m\}$,
that is, $\goal^{\Pi(\bA')} = \{()\}$. 
It can be shown by induction over
the evaluation of Datalog programs (see Remark~\ref{rem:op-dat}) that this is the case if and only if
$\bar a \in P_\Pi^{\Pi'(\bA)}$. 
\end{proof}

We therefore use some of the terminology that we introduced for Datalog also for Nemodeq and for GQ$^+$. 
It is also known that Nemodeq is contained in MSO~\cite{RK2013}. 

\begin{example}\label{expl:ctd}
Example~\ref{expl:MSO-nemodeq} can be expressed in Nemodeq~\cite{RK2013}, and hence is in Datalog and in MSO. 
\enex
\end{example}

\begin{proposition}\label{prop:gqpgso}
Every problem in GQ$^+$ is contained in GSO. 
\end{proposition}
\begin{proof}
This is an immediate consequence of
Proposition~\ref{prop:GQ-GSO},  because
GSO is closed under nesting. 
\end{proof}

\subsection{The Nested Guarded Game}
\label{sect:multi}
In this section we present a modification of the existential guarded $k$-pebble game that allows us to capture the expressive power of GQ and GQ$^+$. 
In particular, we will prove that the query in Example~\ref{expl:MSO-nemodeq} is not in GQ. We call this game the \emph{nested guarded game}; it will also be used  in the next section to prove that there are problems in the intersection of Datalog and GSO that cannot be expressed in GQ$^+$. 

Again we start with an informal description of the game. Let $\tau$ be a finite relational signature. The game has three parameters, $q$, $m$, and $k$. 
There are two players, Spoiler and Duplicator, that play on a pair of $\tau$-structures $(\bA,\bB)$. Each player has $k$ labelled pebbles, out of which $m$ pebbles are blue, the others are red. 
Spoiler places the blue pebbles on $\bA$, Duplicator answers by placing her blue 
pebbles on elements of $\bB$. 
Then the two players play the 
existential guarded $(k - m)$-pebble game
with the red pebbles, with the difference 
that Duplicator loses if the map between all the pebbled vertices (blue and red alike) is not a partial homomorphism. At most $q$ times, 
Spoiler can relocate the red pebbles (without the guard restriction).  
We present a formal definition of winning strategies for Duplicator in the nested guarded game; all statements about the game only involve such winning strategies for Duplicator. 

\begin{definition}\label{def:duplicator}
Let $k,q,m \in {\mathbb N}$. A \emph{winning strategy for Duplicator for the $q$-nested width $m$ guarded $k$-pebble game on two $\tau$-structures $(\bA,\bB)$} is a sequence ${\mathcal H}_0,\mathcal H_1,\dots,{\mathcal H}_q$ of non-empty sets of partial homomorphisms from $\bA$ to $\bB$ such that 
\begin{enumerate}
\item each of ${\mathcal H}_0,{\mathcal H}_1,\dots,{\mathcal H}_q$ is closed under restriction;
\item if $h \in {\mathcal H}_i$, for $i \in \{1,\dots,q\}$, is such that there exists a tuple in a relation of $\bA$ whose entries contain $\dom(h)$, 
then ${\mathcal H}_i$ contains for every $S \subseteq A$ with $|S| \leq k$ and $\dom(h) \subseteq S$ an extension of $h$ with domain $S$;
\item if $h \in {\mathcal H}_{i}$ for $i \in \{1,\dots,q\}$ 
and $S \subseteq A$ with $|S| \leq k$ 
is such that $\dom(h) \subseteq S$ and $|\dom(h)| \leq m$, 
then ${\mathcal H}_{i-1}$ contains an extension of $h$ with domain $S$.
\end{enumerate}
\end{definition}

To prove the connection between 
GQ$^+$ and 
the $q$-nested guarded $k$-pebble game
we need an appropriate notion of canonical Boolean $q$-nested guarded queries. 

\begin{remark} 
We present the game characterisation of  GQ$^+$ only for such GQ$^+$ sentences whose sets of finite models are  complements of CSPs of 
$\omega$-categorical structures $\bB$; however, 
in analogy to Section~\ref{sect:main},
one can define a game that characterises the general case, using 
Corollary~\ref{cor:main}. Since this is notationally and technically heavier, but does not add new ideas or new applications in the following, we only present the special case. 
\enrem
\end{remark}

\begin{definition}[Canonical nested guarded queries]
Let $\bB$ be an $\omega$-categorical structure with a finite relational signature $\tau$ and let $R \subseteq B^m$ be a relation with a primitive positive definition over $\bB$. 
The \emph{canonical $q$-nested guarded query $\Pi_R$ of width $(m,k)$ for $(\bB,R)$} is defined by induction over $q$ as follows. 
For $q = 1$, let $o^1 = (o^1_1,\dots,o^1_m),\dots,o^s = (o^s_1,\dots,o^s_m) \in B^m$ 
be representatives of all orbits of $m$-tuples of $\Aut(\bB)$ that are not contained in $R$. 
For $i \in \{1,\dots,s\}$ let $\Pi_{o^i}$ be the canonical frontier-guarded Datalog program of width $k$ of the $\omega$-categorical structure 
$(\bB;o^i_1,\dots,o^i_m)$.\footnote{Observe that if $(s_1,\dots,s_m)$ and $(s_1',\dots,s_m')$ have the same orbit in $\Aut(\bB)$, then 
$(\bB,s_1,\dots,s_m)$ and $(\bB,s_1',\dots,s_m')$ have the same canonical frontier-guarded Datalog program.} We transform $\Pi_{o^i}$ into an $m$-ary FCP by replacing literals of the form $o^i_j = x$ by 
$\lambda_j  = x$. Let $\Pi_R$ be obtained by taking 
the conjunction of all the resulting FCPs over all $i \in \{1,\dots,s\}$ (which is again an FCP; this can be shown analogously to   Lemma~\ref{lem:unions}
). 

For $q > 1$, we suppose inductively that for all $m \geq 0$ and relations $S$ with a primitive positive definition over $\bB$ 
the canonical $q-1$ nested guarded query $\Xi_S$ of width $(m,k)$ for 
$(\bB,S)$ is already defined.  
Replace in the canonical $1$ nested guarded query of width $(m,k)$ 
for $(\bB,R)$ every occurrence of an EDB $S$ by $\Xi_S$. We define $\Pi_R$ to be the resulting $q$-nested query. 

The \emph{canonical $q$ nested guarded query $\Pi$ of width $(m,k)$ for $\bB$} is obtained from the canonical $q$ nested guarded query $\Pi_R$ of width $(m,k)$ for $(\bB,R)$, where $R$ is the empty relation of arity $0$, by existentially quantifying all free variables. 
\end{definition}

\begin{definition}[$q$-nested $m$-bounded $k$-variable logic]
Let $\tau$ be a relational signature and let $q,m,k \in {\mathbb N}$. 
Then $\tau$-sentences in \emph{$q$-nested $m$-bounded $k$-variable logic $L^{m,k}_q$} are built inductively from atomic $\tau$-formulas 
and the following operation:
if $\phi_1,\dots,\phi_{s}$ are formulas from $L^{m,k}_q$, then 
$$\exists \bar y \; (\phi_1 \wedge \cdots \wedge \phi_{s})$$ is in 
$L^{m,k}_q$ if 
\begin{itemize} 
\item 
there exists an $i \in \{1,\dots,s\}$ and an atomic formula $\psi$ in $\phi_i$ whose variables contain all the variables that are not existentially quantified in $\bar y$, 
\item each of the formulas $\phi_1,\dots,\phi_s$ 
is from $L^{m,k}_{q-1}$.
\end{itemize}
\end{definition}


The following result and its proof are inspired by Theorem~\ref{thm:BD} and its proof. 

\begin{theorem}\label{thm:multi}
Let $\bB$ be an $\omega$-categorical structure 
with finite relational signature $\tau$ 
and let $\bA$ be a finite $\tau$-structure. 
Then for all $q,m,k \in {\mathbb N}$ the following statements are equivalent. 
\begin{enumerate}
\item Every $q$-nested guarded query of width $(m,k)$ which is sound for the complement of $\Csp(\bB)$ is false on $\bA$. 
\item The canonical $q$-nested guarded query of width $(m,k)$ for the complement of $\Csp(\bB)$ 
is false  
on $\bA$.
\item Duplicator has a winning strategy for the $q$-nested width $m$ guarded $k$-pebble game on $(\bA,\bB)$. 
\item Every sentence from $L^{m,k}_{q}$ which holds in $\bA$ also holds in $\bB$. 
\end{enumerate}
\end{theorem}

\begin{proof}
For the implication from 1.\ to 2.,
it suffices to prove that the canonical $q$-nested guarded query $\Pi$ of width $(m,k)$ is sound for the complement of $\Csp(\bB)$. Let $\bA$ be a finite $\tau$-structure. 
Let $R \subseteq B^m$ be primitively positively definable over $\bB$. 
 If $q=1$  
then $\Pi(t_1,\dots,t_m)$ holds in $\bA$ if and only if for every 
tuple $(s_1,\dots,s_m) \in B^m \setminus R$ the canonical frontier-guarded Datalog program for $(\bB,s_1,\dots,s_m)$ derives $\goal$.  
This in turn means that there exists no homomorphism from $\bA$ to $\bB$ which maps $(t_1,\dots,t_m)$ to $(s_1,\dots,s_m)$, because the canonical frontier-guarded Datalog program, which is a subset of the canonical Datalog program of width $k$, is sound by definition. By induction over the nesting depth $q$ it follows that if $\bA$ satisfies the canonical $q$ nested guarded query of with $(m,k)$ for $\bB$, then there is no homomorphism from $\bA$ to $\bB$, and hence $\Pi$ is sound for $\Csp(\bB)$. 

For the implication from 2.\ to 3., let $\bB'$ be the expansion of $\bB$ by all primitive positive definable relations of arity at most $m$; since $\bB$ is $\omega$-categorical, the structure $\bB'$ still has a finite relational signature $\tau'$. We compute a winning strategy ${\mathcal H}_0,{\mathcal H}_1,\dots,{\mathcal H}_q$ for Duplicator for the $q$ nested width $m$ guarded $k$ pebble game on $(\bA,\bB)$ 
 using the canonical $q$-nested guarded query of width $(m,k)$ for $(\bB',R)$ for every relation $R$ of arity $m$ which has a primitive positive definition in $\bB$. 
Define $\bA'_0$ be the expansion of $\bA$ with the same signature as $\bB'$ that contains
for every $m$-ary relation symbol $R \in \tau' \setminus \tau$ the empty relation, i.e., $R^{\bA'_0} = \emptyset$. 
For $i \in \{0,\dots,q-1\}$, suppose that $\bA_0',\dots,\bA_i'$  have already been defined. 
Let $\bA'_{i+1}$ be the expansion of $\bA$ with the same signature as $\bB'$ that contains for every $m$-ary relation symbol $R \in \tau' \setminus \tau$ 
the 
relation computed by 
the canonical $i$-nested guarded query of width $(m,k)$ for $(\bB',R)$
on $\bA'_{i-1}$. 
For $i \in \{0,\dots,q\}$, let ${\mathcal H}_i$ be the set of all partial homomorphisms $f$ from $\bA_i'$ to $\bB'$ with domain of size at most $k$. 
Note that since the canonical $q$ nested guarded query of width $(m,k)$ for $\bB$ is false on $\bA$, 
the relation symbol $F$ of arity $0$ which denotes the empty relation over $\bB$ (which is primitively positively definable) denotes the empty relation over $\bA'_q$ as well, and hence ${\mathcal H}_0,{\mathcal H}_1,\dots,{\mathcal H}_q$ are non-empty. Clearly, these sets are closed under restriction. 

We claim that
the sequence ${\mathcal H}_0,{\mathcal H}_1,\dots,{\mathcal H}_q$ also satisfies items 2.\ and 3.\ in the definition of winning strategies (Definition~\ref{def:duplicator}). 
Indeed, for item 2., let $h \in {\mathcal H}_i$, for $i \in \{1,\dots,q\}$, be such that there exists a tuple in a relation of $\bA$ whose entries contain $\dom(h) = \{a_1,\dots,a_j\}$, and let $S \subseteq A$ be a superset of $\dom(h)$ of size at most $k$. Consider the following frontier-guarded rule in the canonical Datalog program for $(\bB,h(a_1),\dots,h(a_j))$:
the body of the rule is the canonical query $\phi$ of the substructure of $\bA'_q$ induced on $S$. The head of the rule is $R(a_1,\dots,a_j)$ where $R \in \tau'$ is such that $R^{\bB'}$ is the projection of the relation defined by $\phi$ in $\bB'$ to $a_1,\dots,a_j$. By the assumption that there exists a tuple in a relation of $\bA$ whose entries contain $\{a_1,\dots,a_j\}$, this rule is indeed frontier guarded. Since the rule applies to $\bA'$ we obtain that $(a_1,\dots,a_j) \in R^{\bA'_q}$. 
By the definition of ${\mathcal H}_i$, we have that $(h(a_1),\dots,h(a_j)) \in R^{\bB'}$. By the definition of $R^{\bB'}$ we can find elements in $B'$ for the other variables of $\phi$ so that $\phi$ is satisfied, and this shows that $h$ can be extended to a partial homomorphism from $\bA'_q$ to $\bB'$ defined on all of $S$. 


We show item 3.\ by induction over $i \in \{1,\dots,q\}$. 
Let $h \in {\mathcal H}_{i}$ be such that $\dom(h)$ has size at most $m$ and let $S \subseteq A$ be a superset of $\dom(h) = \{a_1,\dots,a_j\}$ of size $j \leq k$. We have to show that ${\mathcal H}_{i-1}$ contains an extension of $h$ with domain $S$. 
Let $\phi$ be the primitive positive formula obtained from the canonical query of the substructure $\bA_S$ of $\bA'_i$ induced on $S$ by existentially quantifying all variables except the ones in the domain of $h$. Then $\phi$ is equivalent to a primitive positive formula over the signature of $\bB$;  let $R$ be the relation symbol of $\bB'$ for the relation defined by $\phi$ in $\bB$. 
For any $(t_1,\dots,t_j) \in B^j \setminus R^{\bB'}$ 
the canonical $q$-nested query derives $\goal$ 
since the canonical query of $\bA_S$ together with $a_1 = t_1,\dots,a_j = t_j$ is unsatisfiable in $(\bB,t_1,\dots,t_j)$. Hence, $(a_1,\dots,a_j) \in R^{\bA'}$. 
In this case, $(h(a_1),\dots,h(a_j)) \in R^{\bB'}$ since $h$ is a partial homomorphism on $\{a_1,\dots,a_j\}$. Therefore, $(h(a_1),\dots,h(a_j))$ satisfies $\phi$
and the witnesses for the existentially quantified variables of $\phi$ provide an extension of $h$ to a partial homomorphism from $\bA'_{i+1}$ to $\bB$ which is defined on all of $S$. 


3.\ implies 4.: 
We show by induction over the syntactic structure of $L^{m,k}_q$ formulas that 
if $\phi(v_1,\dots,v_m)$ is an $L^{m,k}_q$ formula, then for all $h \in {\mathcal H}_q$ and all elements $a_1,\dots,a_m$ from the domain of $h$, if $\bA$ satisfies $\phi(a_1,\dots,a_m)$, then
$\bB$ satisfies $\phi(h(a_1),\dots,h(a_m))$. For $m=0$ this implies 4. The base case of the induction is obvious because atomic formulas are preserved by homomorphisms. Now suppose that 
$\phi$ is of the form $\exists \bar y (\phi_1 \wedge \cdots \wedge \phi_s)$ where 
$\phi_1,\dots,\phi_s$ are formulas from $L^{m,k}_q$. Let $\bar c$ be a tuple providing witnesses for the variables $\bar y$ that shows that $\phi(a_1,\dots,a_m)$ holds in $\bA$. 

If there exists an $i \in \{1,\dots,s\}$ and an atomic formula $\psi(\bar z)$ in $\phi_i$ whose variables contain all the variables that are not existentially quantified in $\bar y$, then $h(\bar z)$ is a tuple in a relation of $\bA$ that contains the domain of $h$. Hence, $h$ has an extension $h'$ whose domain contains the entries of $\bar c$, because $h \in {\mathcal H}_q$. 
Then $h'(\bar c)$ provides the witnesses for $\bar y$ that show that $\phi(h(a_1),\dots,h(a_m))$ holds in $\bB$, because $\phi_i(h(a_1),\dots,h(a_m),h'(\bar c))$ holds by inductive assumption for every $i \leq s$. 
If each of the formulas $\phi_1,\dots,\phi_s$ is even from $L^{m,k}_{q-1}$, then we consider the extension $h'$ of $h$ from ${\mathcal H}_{q-1}$ whose domain also contains the entries of $\bar c$; such an extension exists because $h \in {\mathcal H}_q$. By the inductive assumption, $\bA$ satisfies $\phi_i(h'(a_1),\dots,h'(a_m),h'(\bar c))$ for every $i \leq s$. This concludes the induction. 

Finally, 4.\ implies 1.: suppose that there is a $q$ nested guarded query $\phi$ of width $(m,k)$
which is sound for the complement of $\Csp(\bB)$ and true on $\bA$. We translate the evaluation of $\phi$ on $\bA$ 
into a sentence from $L^{m,k}_q$ which is true on $\bA$ but false on $\bB$. 
Suppose that $\phi$ is of the form $\exists \bar x. \Pi$ where $\Pi$ is a $q$-nested guarded FCP, 
and let  $\bar a$ be elements of $A$ for the variables in $\bar x$ that show that $\phi$ is true in $\bA$. 

For each IDB $R$ of $\Pi$ of arity $k$ and every 
$k$-tuple $\bar a$ such that $\Pi$ derives $R(\bar a)$ on $\bA$ we define an $L^{m,k}_q$-formula $\psi(x_1,\dots,x_k)$ which holds on $\bar a$ in $\bA$ 
and such that $\Pi$ derives $R(\bar x)$ on the canonical database of $\psi(\bar x)$ (after transforming it into prenex normal form) 
by induction over the evaluation of $\Pi$ on $\bA$. 
Let $\theta(x_1,\dots,x_k,\bar y)$ be 
 the body of the rule of $\Pi$ that derived $R(\bar a)$ on $(\bar a,\bar b)$. 
For every atomic formula $S(\bar z)$ of $\theta$ we may assume by inductive assumption that there exists an $L^{m,k}_q$-formula $\chi(\bar z)$ which holds on the respective entries $\bar c$ of $(\bar a,\bar b)$ in $\bA$ and such that $\Pi$ derives $S(\bar c)$ on the canonical database of $\chi(\bar z)$. Then we define $\psi$ to be the conjunction of all these formulas where all variables $\bar y$ are existentially quantified. 
This formula clearly holds on $\bar a$ in $\bA$,
and $\Pi$ derives $R^(\bar x)$ on the canonical database of $\psi$. 
The rule of $\Pi$ that derives $R(\bar a)$ must be guarded, and hence contains  an atomic formula whose variables contain $x_1,\dots,x_k$, so the formula $\psi$ is indeed an $L^{m,k}_q$-formula. 
Then the formula that we obtain for the IDB $R$ of $\Pi$ that denotes $B^k$ in $\bB'$ by existentially quantifying all variables is true on $\bA$ but false on $\bB$. 
\end{proof} 

\begin{corollary}\label{cor:inexpr}
Let $\bB$ be an $\omega$-categorical structure with a finite relational signature $\tau$ and let $q,m \geq 1$.  
Then the complement of $\Csp(\bB)$ cannot be expressed by a $q$-nested width $m$ guarded query if and only if for all $k \geq 1$ there exists a finite $\tau$-structure $\bA \notin \Csp(\bB)$ such that Duplicator has a winning strategy for the $q$-nested width $m$ guarded $k$-pebble game on $(\bA,\bB)$. 
\end{corollary}
\begin{proof}
For the backwards direction, let 
$\Pi$ be a $q$-nested guarded query of width $(m,k)$ which is sound for the complement of $\Csp(\bB)$. 
By assumption, there exists 
$\bA \notin \Csp(\bB)$ such that Duplicator has a winning strategy for the $q$ nested width $m$ guarded $k$ pebble game on $(\bA,\bB)$. By the implication from 3.\ to 1.\ in Theorem~\ref{thm:multi}, the query $\Pi$ is false on $\bA$. But then $\Pi$ does not express the complement of $\Csp(\bB)$. 

For the forwards direction, we show the contraposition. We assume that there exists a $k \geq 1$ such that for every 
$\bA \notin \Csp(\bB)$ Duplicator has no winning strategy 
for the $q$ nested width $m$ guarded $k$ pebble game on $(\bA,\bB)$. It suffices to prove that the canonical $q$ nested guarded query $\Pi$ of width $(m,k)$ for the complement of $\Csp(\bB)$ expresses the complement of $\Csp(\bB)$. 
Let $\bA$ be an instance of $\Csp(\bB)$. 
If $\bA$ is satisfiable then $\Pi$ is false on $\bA$, because $\Pi$ is sound for the complement of $\Csp(\bB)$ as shown in the proof of the 
implication 1.\ implies 2.~in
Theorem~\ref{thm:multi}. If $\bA$ is unsatisfiable then the contraposition of the implication from 2.\ to 3.\ in  Theorem~\ref{thm:multi} implies that 
$\Pi$ is true on $\bA$. Hence, $\Pi$ expresses  the complement of $\Csp(\bB)$. 
\end{proof}


To illustrate the power of our game characterisation of GQ, we show that there are problems in the intersection of Datalog and MSO that are inexpressible in GQ. 

\begin{proposition}\label{prop:inexpr}
The class of finite structures described by the MSO sentence $\Psi$ from Example~\ref{expl:MSO-nemodeq} is (in Datalog, but) not expressible in GQ. 
\end{proposition}
\begin{proof}
It is easy to see that $\Psi$ describes a class $\mathcal C$ of finite structures whose complement is closed under disjoint unions, and hence is the complement of a CSP by Remark~\ref{rem:CSP}. By Corollary~\ref{cor:main} 
there exists an $\omega$-categorical structure $\bB$ 
such that the complement of $\Csp(\bB)$ equals
${\mathcal C}$. 

To prove the inexpressibility in GQ, we use Corollary~\ref{cor:inexpr} for the special case $q=1$. Let $m,k \geq 1$. 
Let $\phi_1(x_0,x_{k})$ be the formula 
$$\exists x_1,\dots,x_{k-1} \bigwedge_{i \in \{0,\dots,k-1\}} C(x_i,x_{i+1}).$$ 
Let $\phi_2(y_0,y_k)$ be 
formula $\exists z,y_1,\dots,y_{k-1} \bigwedge_{i \in \{1,\dots,k-1\}} \big (\phi_1(y_i,z) \wedge L(y_i,y_{i+1}) \big).$
Let 
$\phi$ be the formula
$\bigwedge_{i \in \{0,\dots,m\}} \phi_2(z_i,z_{i+1}) \wedge D(z_0,z_{m+1})$, rewritten in prenex normal form, and let $\bA$ be the canonical database of $\phi$. 

\medskip 
{\bf Claim 1.} $\bA$ satisfies $\Psi$, and hence is an unsatisfiable instance of $\Csp(\bB)$. 
First note that $(z_0,z_1),(z_1,z_2),\dots,(z_m,z_{m+1})$ are precisely the pairs in $A^2$
that do not satisfy the formula $\Phi_1$ from definition of $\Psi$ in Example~\ref{expl:MSO-nemodeq}. Then 
the vertices $z_0,z_{m+1}$ play the role of $u$ and $v$: to see that $\Phi_2(u,v)$ does not hold in $\bA$, let $U_1 \subseteq A$ be such that it satisfies the first two conjuncts of $\Phi_2$ for all $x,y \in A$. Then $U_1$ must contain $z_0$ by the first conjunct, and $z_1$ by the second conjunct, and inductively by the second conjunct it must contain all of $z_0,z_1,\dots,z_{m+1}$, and hence $v$. But then the last disjunct $\neg (U_1(x) \wedge \neg \Phi_1(x,v))$ is not satisfied for $x=z_m$. 

\medskip 
{\bf Claim 2.} Duplicator has a winning strategy  
${\mathcal H}_0,{\mathcal H}_1$ for the $1$-nested width $m$ guarded $k$ pebble game on $\bA$ and $\bB$: we set ${\mathcal H}_0$ to be the set of all partial homomorphisms $h$ from $\bA$ to $\bB$ such that $\dom(h) \leq m$. Set ${\mathcal H}_1$ to be the set of all partial homomorphisms $h$ from $\bA$ to $\bB$ such that if $\dom(h)$ contains the variable $z$ from a conjunct $\phi_2(y_0,y_k)$ of $\phi$, and $\phi_2'(\bar u)$ is the formula obtained from the rewriting of $\phi_2$ into prenex normal form and existentially quantifying all variables except for the ones that are in the domain of $h$, 
then $h(\bar u)$ satisfies $\phi_2'$ in $\bB$. 
Note that $\dom(h)$ cannot contain the variable $z$ for all of the conjuncts of $\phi$ of the form $\phi_2(y_0,y_k)$, and hence ${\mathcal H}_1$ is non-empty. It is straightforward to verify that 
${\mathcal H}_0,{\mathcal H}_1$ satisfies
the three items 
 from 
Definition~\ref{def:duplicator}. 

The statement is an immediate consequence of
Corollary~\ref{cor:inexpr} for $q=1$ and the two claims. 
\end{proof}

\subsection{Separation in GSO}
\label{sect:sep}
In this section we present an example of a GSO sentence which is in Datalog,
but cannot be expressed in GQ$^+$. 
The GSO sentence describes the complement of the CSP of the following structure $\bB$, which is fixed throughout this section. 
Define $\bB := ({\mathbb N};D,R)$ where 
\begin{align*}
D & \coloneqq \{(x,y) \in {\mathbb N}^2 \mid x \neq y\} \\
 R & \coloneqq \{(x,y,z)  \in {\mathbb N}^3 \mid x=y \Rightarrow y = z\}.
 \end{align*}
Clearly, the structure $\bB$ is $\omega$-categorical.
The structure $\bB$ plays a prominent role in the classification of structures preserved by all permutations~\cite{BodChenPinsker} and the complexity of quantified constraint satisfaction problems for equality constraints~\cite{qecsps,MartinZhuk}. 
The proof of the following proposition uses 
ideas from Example~\ref{expl:eq-mso}. 

\begin{proposition}\label{prop:eq-gso}
There is a GSO sentence that expresses
$\Csp(\bB)$. 
\end{proposition}
\begin{proof}
Let $\tau = \{D,R\}$ be the signature of $\bB$. 
Let $S$ be a new ternary and 
$X$ a new unary relation symbol. 
Let $\chi_X$ be the formula which specifies
that 
\begin{itemize}
\item $X$ is non-empty, 
\item if $x,y,z$ are such that $S(x,y,z)$, then $y \in X$ if and only if $z \in X$, and
\item $X$ is minimal with respect to the properties above (the details are similar as in  Example~\ref{expl:eq-mso}). 
\end{itemize} 
Let $\phi$ be the $\tau \cup \{S,X\}$-sentence which 
expresses that there exists $S \subseteq R$ such that for all $X$ that satisfy $\chi_X$ we have
\begin{enumerate}
\item for all $x,y$, if  $D(x,y)$ then $x \notin X$ or $y \notin X$, and 
\item for all $x,y,z$, if $R(x,y,z)$ but not $S(x,y,z)$, 
then $x \notin X$ or $y \notin X$.  
\end{enumerate}
We claim that a finite structure $\bA$ satisfies $\phi$ if and only if $\bA$ has a homomorphism to $\bB$. 
If $h$ is a homomorphism from $\bA$ to $\bB$, 
then we define $S \subseteq R$ to be the relation that contains $(x,y,z) \in R$ if and only if $h(x) = h(y)$. Let $X \subseteq A$ be such that $\chi_X$ is satisfied. 
We have to show that 1.\ and 2.\ above are satisfied. Let $\bA'$ be the $\{E,D\}$-structure with domain $A$ and the relations
\begin{align*}
E^{\bA'} & \coloneqq \bigcup_{(x,y,z) \in S} \{(x,y),(y,z)\} \\
D^{\bA'} & \coloneqq D^{\bA} \cup \{(x,y) \mid \exists z (R(x,y,z) \wedge \neg S(x,y,z)) \} 
\end{align*}
Note that $h$ is a homomorphism from 
$\bA'$ to
structure $({\mathbb N};E,D)$ 
from Example~\ref{expl:eq-mso}.
Recall from Example~\ref{expl:eq-mso} that 
any homomorphism from $\bA'$ to $({\mathbb N};E,D)$ must be constant on $X$. 
Hence, if $(x,y) \in D$ then $h(x) \neq h(y)$ and hence we must have $x \notin X$ or $y \notin X$. 
To see that 2.~is satisfied as well, note that if
$(x,y,z) \in R \setminus S$, then $(x,y) \in D^{\bA'}$, and hence $h(x) \neq h(y)$ by the definition of $S$, which again implies that $x \notin X$ or $y \notin X$. 

Now suppose that $\bA$ satisfies $\phi$. Let
$S \subseteq R$ be the ternary relation witnessing this. Define the structure $\bA'$ as above. 
As in Example~\ref{expl:eq-mso} we may argue that there exists a homomorphism from $\bA'$ to 
$({\mathbb N};E,D)$, which is a homomorphism from $\bA$ to $\bB$ as well. 
\end{proof}

\begin{proposition}\label{prop:eq-gq}
The complement of $\Csp(\bB)$ cannot be expressed in GQ$^+$. 
\end{proposition}
\begin{proof}
Let $q,m \geq 1$. 
We use Corollary~\ref{cor:inexpr} to show that $\Csp(\bB)$ cannot be expressed by a $q$-nested width $m$ guarded query. Let $k \geq 1$. 
We construct an unsatisfiable instance $\bA$ of $\Csp(\bB)$ 
such that Duplicator wins the $q$ nested width $m$ guarded $k$-pebble game on $(\bA,\bB)$.  
Without loss of generality we may assume that $m \leq k$. 
It will be convenient to specify $\bA$ as the canonical database of a primitive positive 
$\{D,R\}$-formula $\phi$.
Let $\phi_0(x,y)$ be the formula $R(x,x,y)$, and for $i > 0$  let $\phi_i(x,y)$ be the formula 
$$\exists x_0,\dots,x_k \left (x = x_0 \wedge  \bigwedge_{j = 0}^{k-1} \phi_{i-1}(x_j,x_{j+1}) \wedge R(x_0,x_k,y) \right ).$$


Let $\phi$ be the formula $\phi_{s}(x,y) \wedge D(x,y)$ where 
$s := q+1$. 
Let $\bA$ be the canonical database of $\phi$ after transforming 
$\phi$ into prenex normal form. 
Note that every variable of $\phi_i(x,y)$ except $x$ and $x_0$ appears exactly once as the final argument in a conjunct that involves $R$. 



\medskip 
{\bf Claim.} $\bA$ is an unsatisfiable instance of $\Csp(\bB)$. This follows from the fact that $\phi_i(x,y)$ implies $x=y$, for every $i \leq s$, which can be shown by a straightforward induction over $i$. 

\medskip 
{\bf Claim.} 
Duplicator has a winning strategy 
in the $q$ nested width $m$ guarded $k$-pebble game 
on $(\bA,\bB)$. 

Let ${\mathcal H}_0$ be the set of all partial homomorphisms $h$ from 
substructures of $\bA$ to $\bB$ 
with domain size at most $k$, and 
let ${\mathcal H}_i$, for all $i \in \{1,\dots,q\}$, 
be the set of all partial homomorphisms $h$ from 
substructures of $\bA$ to $\bB$ 
with domain size at most $k$ 
such that $h(u)=h(v)$, for $u,v \in \dom(h)$, if 
and only if both $u$ and $v$ 
\begin{itemize}
\item are part of 
the same subformula of $\phi$ of the form $\phi_{i-1}$, or 
\item are elements of
different subformulas of $\phi$ of the form $\phi_{i-1}$ but $\dom(h)$ contains all the variables
$x_r,\dots,x_s,x_{s+1}$ from $\phi_i$ for some $0 \leq r < s < k$ and $u$ appears in the subformula 
$\phi_{i-1}(x_{r},x_{r+1})$ and $v$ appears in the subformula $\phi_{i-1}(x_{s},x_{s+1})$. 
\end{itemize}
Since $\dom(h) \leq k$ it is impossible that $r=0$ and $s=k-1$; hence, if $x_0$ and $x_k$ are from $\dom(h)$, then 
we will have $h(x_0) \neq h(x_k)$. 

Clearly, ${\mathcal H}_i$ is closed under restrictions and non-empty. 
Now suppose that $h \in {\mathcal H}_i$ is such that there exists a tuple $a$ in a relation of $\bA$ whose entries contain $\dom(h)$. Let $S \subseteq A$ be a superset of $\dom(h)$ of cardinality at most $k$. 
If this relation is $D$, then
$a$ equals $(x,y)$, 
we may clearly extend $h$ to some 
$h' \in  {\mathcal H}_i$ with $h'(x) \neq h'(y)$. 
Otherwise, the relation must be $R$ and $a$ equals $(x_0,x_k,y)$ for some variables $x_0,x_k,y$ from a subformula of $\phi$ of the form $\phi_j(x,y)$ for some $j \leq i$.  
Note that $h(x_0) = h(x_k) = h(y)$ if $j < i$, 
and $h(x_0) \neq h(x_k)$ and $h(x_k) \neq h(y)$ otherwise. It is easy to see that in both cases we may extend $h$ to some $h' \in {\mathcal H}_i$ with $\dom(h') = S$. 
Finally,  if $h \in {\mathcal H}_{i}$ for $i \in \{1,\dots,q\}$ and $S \subseteq A$ with $|S| \leq k$ is such that $\dom(h) \subseteq S$,  
then ${\mathcal H}_{i-1}$ contains an extension $h'$ of $h$ with domain $S$: 
if $x \in S$ appears in a copy of $\phi_{i-1}$ 
whose variables intersect $\dom(h)$, we set $h'(x)$ to the value taken by these variables under $h$.
If $\dom(h)$ contains all the variables $x_r,\dots,x_s,x_{s+1}$ from $\phi_i$ for some $0 \leq r < s < k$ and $x$ appears in the subformula $\phi_{i-1}(x_s,x_{s+1})$ and some variable of $\phi_{i-1}(x_r,x_{r+1})$ appears in $\dom(h)$, then
we set $h'(x)$ to the value taken by this variable under $h$.  Otherwise, pick $h'(x)$ to be a new value. 
Hence, ${\mathcal H}_0,{\mathcal H}_1,\dots,{\mathcal H}_q$ is indeed a winning strategy for Duplicator for the $q$ nested width $m$ guarded $k$ pebble game on $(\bA,\bB)$. 
\end{proof}

\begin{proposition}\label{prop:ex-datalog}
The complement of $\Csp(\bB)$ can be expressed in Datalog. 
\end{proposition}
\begin{proof}
It is easy to see that the complement of $\Csp(\bB)$ can be defined by a Datalog program of width $(2,3)$, because for a given $\tau$-structure $\bA$ it suffices to 
\begin{itemize}
\item 
compute the smallest reflexive, symmetric, and transitive binary relation $T$ on all elements participating in $R$ such that for every $(x,y,z) \in R^{\bA}$ with $(x,y) \in T$ we have $(y,z) \in T$:\\
	\begin{minipage}{0.35\textwidth}
		\begin{align*}
			T(y,z) & \; {:}{-}  \; T(x,y), R(x,y,z)  \\
			T(y,x) & \; {:}{-}  \; T(x,y)  \\
			T(x,z) & \; {:}{-}  \; T(x,y), T(y,z)  
		\end{align*}
	\end{minipage}
	\begin{minipage}{0.3\textwidth}
		\begin{align*}
			T(x,x) & \; {:}{-}  \; R(x,y,z)  \\
			T(y,y) & \; {:}{-}  \; R(x,y,z)  \\
			T(z,z) & \; {:}{-}  \; R(x,y,z),  
		\end{align*}
	\end{minipage}
\\
\item 
and then to check whether $T$ and $D$ contain a common pair of elements or $D$ has a loop:\\[1ex]
$\qquad\ \  \goal  \; {:}{-} \; T(x,y), D(x,y)  \qquad \qquad \ 	\goal  \; {:}{-} \; D(x,x).$ \qedhere
\end{itemize} 
\end{proof}

\begin{corollary}\label{cor:gqplus}
The intersection of GSO and Datalog contains problems that cannot be expressed in 
GQ$^+$. 
\end{corollary}
\begin{proof} 
Since GSO is closed under negation, 
Proposition~\ref{prop:eq-gso} implies that 
the complement of $\Csp(\bB)$ is in GSO, and  Proposition~\ref{prop:eq-gq} shows that
$\Csp(\bB)$ is not in GQ$^+$. 
\end{proof} 

We will see a strengthening of this result in 
Corollary~\ref{cor:gqplus+}.

\subsection{Separation in MSO}
\label{sect:sep-mso}
In this section, we will prove that there are problems in MSO that can be expressed in Datalog, but that cannot be expressed in Nemodeq. The MSO sentence describes 
the complement of the CSP of a structure $\bC$ which is closely related to the structure $\bB$ from Section~\ref{sect:sep}. 
Throughout this section, let $\bC$ be the 
$\{D,R\}$-structure 
$(\Dom;D^\bC,R^\bC)$ where 
\begin{align*}
    \Dom & \coloneqq \mathbb{N} \ \dot{\cup} \ 
    {\mathbb N}^3 \\   
	D^\bC & \coloneqq \{(x,y) \in \Dom^2 \mid x,y \in \mathbb{N},\  x \neq y\} \\
	R^\bC & \coloneqq \{(x,y,z,t)  \in \Dom^4 \mid x,y,z \in \mathbb{N}, \ t=\atuple{x,y,z}, \  x=y \Rightarrow y = z\}.
\end{align*}

That is, the domain $\Dom$ contains two sorts of elements: all natural numbers \emph{and} all triples of natural numbers.\footnote{For the triples \emph{inside} $\Dom$, we will use the notation $\atuple{n_1,n_2,n_3}$, to better distinguish them from tuples \emph{over} $\Dom$ which we will continue to enclose in round parentheses $(\ldots)$.} In words, the binary relation $D^\bC$ expresses disequality between natural numbers, whereas the quaternary relation $R^\bC$ consists of quadruples where the first three components are natural numbers while the last component is the triple built from the first three components. 
Among all such quadruples, $R^\bC$ contains those where the first three components are identical whenever the first two components are. An alternative formulation would be to say that $R^\bC$ contains all quadruples of the form $(x,y,z,\atuple{x,y,z})$ except those where $x=y$ and $y \neq z$.

\begin{proposition}\label{prop:datalog} 
The complement of $\Csp(\bC)$ is in Datalog. 
\end{proposition}
\begin{proof} 
The idea of the proof is analoguous to that of Proposition~\ref{prop:ex-datalog}. Consider the following Datalog program $\Pi$\\ 
\begin{minipage}{0.45\textwidth}
	\begin{align*}
		T(y,z) & \; {:}{-}  \; T(x,y), R(x,y,z,t)  \\
		T(y,x) & \; {:}{-}  \; T(x,y)  \\
		T(x,z) & \; {:}{-}  \; T(x,y), T(y,z) \\ 
		\goal & \; {:}{-} \; T(x,y), D(x,y) \\
		\goal & \; {:}{-} \; T(t,t), R(x,y,z,t) 
	\end{align*}
\end{minipage}
\begin{minipage}{0.3\textwidth}
\begin{align*}
T(x,x') & \; {:}{-} \; R(x,y,z,t), R(x',y',z',t) \\
T(y,y') & \; {:}{-} \; R(x,y,z,t), R(x',y',z',t) \\
T(z,z') & \; {:}{-} \; R(x,y,z,t), R(x',y',z',t) \\
T(x,x) & \; {:}{-}  \; D(x,y)  \\
T(y,y) & \; {:}{-}  \; D(x,y)  
\end{align*}
\end{minipage}\\[2ex]
Now assume $\bA$ is such that $\Pi$ does not derive $\goal$ and let $\bA' = \Pi(\bA)$.
Let $$A^* = \{a_1,a_2,a_3 \in A \mid (a_1,a_2,a_3,a_4) \in R^{\bA}\} \cup \{a_1,a_2 \in A \mid (a_1,a_2) \in D^{\bA}\}.$$ Thanks to $\bA' = \Pi(\bA)$ and $\Pi(\bA)$ not deriving $\goal$, we observe the following: 
\begin{enumerate}
\item $T^{\bA'}$ is an equivalence relation on $A^*$ and is disjoint from $D^\bA$,  
\item $A^*$ is disjoint from $\{a_4 \in A \mid (a_1,a_2,a_3,a_4) \in R^{\bA}\}$,
\item for any $(a_1,a_2,a_3,a) \in R^{\bA}$ and $(a'_1,a'_2,a'_3,a) \in R^{\bA}$, we have that $a_i$ and $a'_i$ are in the same $T^{\bA'}$-equivalence class for any $i \in \{1,2,3\}$.   
\end{enumerate}
We now define a homomorphism $h$ from $\bA$ to $\bC$. Pick an injective function $h^* \colon {A^*}_{\!\!/T^{\bA'}} \to \mathbb{N}$ mapping $T^{\bA'}$-equivalence classes to natural numbers. For every $a \in A^*$, we let $h(a)=h^*([a]_{T^{\bA'}})$ and for every $a \in A \setminus A^*$, we let $h(a) \in \mathbb{N}^3$ be such that, whenever $(a_1,a_2,a_3,a) \in  R^\bA$ for some $a_1,a_2,a_3 \in A^*$ then $h(a) = \langle h(a_1),h(a_2),h(a_3) \rangle$ (thanks to item 3 above, this choice is always possible). It can now be readily checked that $h$ is a homomorphism from $\bA$ to $\bC$.     
\end{proof} 


We will now argue that $\Csp(\bC)$ can be expressed in MSO.
We first 
provide a characterisation for a given $\bA$ being in $\Csp(\bC)$.

\begin{lemma}\label{lem:sepchar}
A finite $\{D,R\}$-structure $\bA$ is in $\Csp(\bC)$ if and only if the following conditions are satisfied:
\begin{enumerate}
\item[(1)] There is a partition $(A',A'')$ of $A$ such that $A'$ consists of exactly those elements of $A$ occurring in the fourth position of some tuple from $R^{\bA}$, while $A''$ contains at least all elements of $A$ occurring in $D^{\bA}$ or in any of the the first three positions of some tuple from $R^{\bA}$.
\item[(2)] The binary relation $D^{\bA}$ is disjoint from the smallest equivalence relation $\approx$ over $A''$ that satisfies the following:
\begin{itemize}
\item[(2a)] For any $(a_1,b_1,c_1,d) \in R^{\bA}$ and $(a_2,b_2,c_2,d) \in R^{\bA}$, we have $a_1 \approx a_2$ as well as $b_1 \approx b_2$ and $c_1 \approx c_2$.
\item[(2b)] For any $(a,b,c,d) \in R^{\bA}$ with $a \approx b$, we also have $b \approx c$.   
\end{itemize}
\end{enumerate}
\end{lemma}
\begin{proof}
For the ``if'' part, we provide a homomorphism $h:A \to \Dom$ defined by
$$
a \mapsto \left\{
\begin{array}{ll}
f([a]_\approx) & \text{ if } a \in A''\\
\atuple{f([b]_\approx),f([c]_\approx),f([d]_\approx)} & \text{ whenever } (b,c,d,a) \in R^{\bA}\\ 
\end{array}
\right.
$$
where $f \colon {A''}_{\!\!/\approx} \to \mathbb{N}$ is an arbitrary injection from the $\approx$-equivalence classes of $A''$ to the natural numbers. Note that (2a) ensures well-definedness of the second case. 

For the ``only if'' part, assume a homomorphism $h$ from $\bA$ to $\bC$. Without loss of generality we may assume $h$ to be such that it maps any $a \in A$ not participating in any of the relations to an arbitrary element of $\mathbb{N}$.
We then let $A'' = h^{-1}(\mathbb{N})$ and $A' = A \setminus A''$, which by assumption satisfy Condition (1). We define an equivalence relation $\dot{\approx}$ over $A''$ by letting $a\ \dot{\approx}\ b$ iff $h(a) = h(b)$. We observe that $\dot{\approx}$ must be disjoint from $D^\bA$ and that ${\approx} \subseteq \dot{\approx}$. Therefore $\approx$ must be disjoint from $D^\bA$, warranting Condition (2).    
\end{proof}

We will construct an MSO sentence that describes $\Csp(\bC)$ in a stepwise fashion.
As a ``witness'' of the existence of a homomorphism $h$ from $\bA$ to $\bC$, we will employ a set variable $X$ which is supposed to correspond to the set
$$ 
\{a \in A \mid h(a) = \atuple{n,n,n}\in \mathbb{N}^3 \}. 
$$
If $X$ represented this set, $\bA$ would necessarily need to satisfy the MSO formula
$$\varphi_\mathrm{well}(X) \coloneqq \varphi_\mathrm{sort}(X) \wedge \varphi_\mathrm{triag}(X) \wedge \varphi_\mathrm{compat}(X)$$
whose components are specified as follows:

The first conjunct $\varphi_\mathrm{sort}$ prevents ``type clashes'' between numbers and triples (where the set variable $T$ is used to identify the elements that $h$ maps to triples): 
\begin{align*}
	\varphi_\mathrm{sort}(X) \coloneqq \exists T. X\subseteq T \wedge \forall x,y,z,t & \big(R(x,y,z,t) \Rightarrow \neg T(x) \wedge \neg T(y) \wedge \neg T(z) \wedge T(t)\big) \\
	& \wedge \big(D(x,y) \Rightarrow \neg T(x) \wedge \neg T(y)\big).
\end{align*}

The second conjunct $\varphi_\mathrm{triag}$ ensures that $X$ covers those $a \in A$ with $h(a)=\atuple{n_1,n_2,n_3}$ where $n_1=n_2=n_3$ is warranted because of $n_1 = n_2$ being known: 
$$\varphi_\mathrm{triag}(X) \coloneqq \forall x,y,z,t. R(x,y,z,t) \wedge \varphi_\mathrm{Eq}(x,y,X) \Rightarrow X(t)$$
where $\varphi_\mathrm{Eq}(x,y,X)$ is an MSO formula with two additional free first-order variables which (presuming the choice of $X$ as above) holds for every pair $(a_1,a_2) \in A^2$ for which the above Datalog program would derive the $Eq$ predicate:
$$\varphi_\mathrm{Eq}(x,y,X) \coloneqq \forall Y. \big(Y(x) \wedge \varphi_\mathrm{propag}(Y,X) \Rightarrow Y(y)\big)$$
with the following formula $\varphi_\mathrm{propag}(Y,X)$ enforcing that $Y$ is a set which is closed under ``derived equalities'':
\begin{align*}
\forall x_1,\ldots,x_6,t & \big( R(x_1,x_2,x_3,t) \wedge R(x_4,x_5,x_6,t) \big)
 \Longrightarrow \bigwedge_{\mathclap{1\leq i\leq 3}} \big( Y(x_i) \Leftrightarrow Y(x_{i+3}) \big) \\ 
  \wedge \ \forall x_1,x_2,x_3,t & \big( R(x_1,x_2,x_3,t) \wedge X(t) \big)
\Longrightarrow \bigwedge_{\mathclap{1 \leq i < j \leq 3}}\big( Y(x_i) \Leftrightarrow Y(x_j) \big).
\end{align*}

Finally, the third conjunct $\varphi_\mathrm{compat}(X)$ ensures that the ``derivable equalities'' are not in conflict with the disequalities imposed through $D$:
$$\varphi_\mathrm{compat}(X) \coloneqq \forall x,y \big (\varphi_\mathrm{Eq}(x,y,X) \Rightarrow \neg D(x,y) \big ).$$

\begin{proposition}\label{prop:eq-mso}
A finite $\{D,R\}$-structure $\bA$ is in $\Csp(\bC)$ if and only if it satisfies the MSO sentence
$\exists X. \varphi_\mathrm{well}(X)$.
\end{proposition}
\begin{proof}
For the ``only if'' direction, we make a formal argument along the intuitions provided above. Assume a homomorphism $h$ from $\bA$ to $\bC$. We show that the assignment $X \mapsto A^=$ with $$A^= \coloneqq \{a \in A \mid (n,n,n,h(a)) \in R^{\bA}, n\in \mathbb{N} \}$$ validates $\varphi_\mathrm{well}(X)$: 

It validates $\varphi_\mathrm{sort}(X)$ because one can choose $T \mapsto h^{-1}(\mathbb{N}^3)$ whereas in order to satisfy the premises of any implication inside $\varphi_\mathrm{sort}(X)$, all first order variables but $t$ must be instantiated with elements of $h^{-1}(\mathbb N)$. 

Validation of $\varphi_\mathrm{triag}(X)$ can be obtained from the fact that, for any $(a_x, a_y, a_z, a_t) \in R^{\bA}$ satisfying $h(a_x)=h(a_y)$ it must also hold that $h(a_y)=h(a_z)$ by virtue of $h$ being a homomorphism into $\bC$ -- where it remains to show that $\varphi_\mathrm{Eq}(a_x,a_y,A^=)$ indeed implies $h(a_x)=h(a_y)$. To see that this is true, consider the case where $Y$ is mapped to the set $A_x = h^{-1}(h(a_x))$. Then the implication's consequence indeed ensures $h(a_x)=h(a_y)$, provided that the premise is found to be true. The premise's first conjunct holds by construction, whence it remains to show that $\varphi_\mathrm{propag}(A_x)$ also holds true, which -- in view of our chosen assignment for $X$ -- again follows from the fact that $h$ is a homomorphism into $\bC$ .

Finally, validation of $\varphi_\mathrm{compat}(X)$ is established by recalling that $\varphi_\mathrm{Eq}(a_x,a_y,A^=)$ implies $h(a_x)=h(a_y)$ (see above) and therefore $(h(a_x),h(a_y)) \not\in D^{\bC}$ by definition.

\medskip

For the ``if'' direction, assume that some assignment $X \mapsto A_X \subseteq A$ validates $\varphi_\mathrm{well}(X)$ for $\bA$. We will use \Cref{lem:sepchar} to argue that $\bA$ is in $\Csp(\bC)$. Let $A'' = \{a_4 \mid (a_1,a_2,a_3,a_4) \in R^\bA\}$ and let $A' = A \setminus A''$.
Then, by virtue of $\varphi_\mathrm{sort}$, the partition $(A',A'')$ satisfies Condition (1) of  \Cref{lem:sepchar}. Satisfaction of the Condition (2) would be an immediate consequence of $\varphi_\mathrm{compat}(X)$ if one could establish that the binary relation $Eq_{A_X} \coloneqq \{(a_1,a_2) \mid \varphi_\mathrm{Eq}(a_1,a_2,A_X)\}$ contains $\approx$. This is achieved by arguing that $Eq_{A_X}$ is an equivalence relation and satisfies the two subitems in \Cref{lem:sepchar} defining $\approx$. The former is a result of the structure of  $\varphi_\mathrm{Eq}$ as an undirected reachability check.
The latter is enforced through $\varphi_\mathrm{propag}$ and $\varphi_\mathrm{triag}$ which together ensure the two subconditions (2a) and (2b).  
\end{proof}

The proof of the following proposition is similar to the proof of Proposition~\ref{prop:eq-gq}. 

\begin{proposition}\label{prop:eq-c-gq}
	The complement of $\Csp(\bC)$ cannot be expressed in GQ$^+$. 
\end{proposition}
\begin{proof}
Let $q,m,k \in {\mathbb N}$ be such that $q,m \geq 1$ and $k \geq m$. We construct an unsatisfiable instance $\bA$ of $\Csp(\bC)$ such that Duplicator wins the $q$ nested width $m$ guarded $k$-pebble game on $(\bA,\bC)$. Let $\phi'_i$ be as defined in the proof of Proposition~\ref{prop:eq-gq}, 
but replacing $R(x,y,z)$ by $R(x,y,z,t)$ where $t$ is a fresh variable. Let $\bA$ be the canonical database of $\phi := \phi_{q+1}(x,y) \wedge D(x,y)$ after transforming it into prenex normal form. Similarly as in the proof of 
Proposition~\ref{prop:eq-gq}, it is straightforward to show that $\bA$ does not have a homomorphism to $\bC$. We need to show that Duplicator has a winning strategy in the $q$ nested width $m$ guarded $k$-pebble game on $(\bA, \bC)$. 
Let ${\mathcal H}_0$ be the set of all partial homomorphisms $h$ from substructures of $\bA$ to $\bC$ with domain size at most $k$. Let ${\mathcal H}_i$, for all $i \in \{1,\dots,q\}$, be the set of all partial homomorphisms $h$ from substructures of $\bA$ to $\bC$ with domain size at most $k$ defined as follows. 
If $u,v \in \dom(h)$ are variables of $\phi$ that do \emph{not} appear as the last argument of any atomic formula $R(x,y,z,t)$ in $\phi$, then we require that $h(u) = h(v)$  
 if and only if both $u$ and $v$ are part of the same subformula of $\phi_{q+1}(x,y) \wedge D(x,y)$ of the form $\phi_{i-1}$, or $\dom(h)$ contains all the variables $x_r,\dots,x_s,x_{s+1}$ from $\phi_i$ for some $0 \leq r < s < k$ and $u$ appears in the subformula $\phi_{i-1}(x_r,x_{r+1})$ and $v$ appears in the subformula $\phi_{i-1}(x_s,x_{s+1})$. 
If $t,t' \in \dom(h)$ are such that $\phi$ contains conjuncts $R(x,y,z,t)$ and $R(x',y',z',t')$, then these conjuncts are unique by the construction of $\phi$; in this case, $h(t) = h(t')$ if and only if $h(x)=h(x')$, $h(y)=h(y')$, and $h(z) = h(z')$. Otherwise, 
$h(t)$ must be distinct from all other points in the image of $h$. 

 Since $\dom(h) \leq k$ it is impossible that $r=0$ and $s = k-1$; hence, 
if $x_0$ and $x_k$ are from $\dom(h)$, then $h(x_0) \neq h(x_k)$. Clearly, ${\mathcal H}_i$ is closed under restrictions and non-empty. Now suppose that $h \in {\mathcal H}_i$ is such that there exists a tuple $a$ in a relation of $\bA$ whose entries contain $\dom(h)$. Let $S \subseteq A$ be a superset of $\dom(h)$ of cardinality at most $k$. If this relation is $D$, then $a$ equals $(x,y)$ and we may clearly extend $h$ to some $h' \in {\mathcal H}_i$ with $h'(x) \neq h'(y)$. 
Otherwise, the relation must be $R$ and $a$ equals $(x_0,x_k,y,t)$ for some variables $x_0,x_k,y,t$ from a subformula $\phi$ of the form $\phi_j(x,y)$ for some $j \leq i$. Note that $h(x_0) = h(x_k) = h(y)$ if $j < i$, and $h(x_0) \neq h(x_k)$ and $h(x_k) \neq h(y)$ otherwise. It is easy to see that in both cases we may extend $h$ to some $h' \in {\mathcal H}_i$ with domain $S$. 
Finally, if $h \in {\mathcal H}_i$ for $i \in \{1,\dots,q\}$ and $S \subseteq A$ with $|S| \leq k$ is such that $\dom(h) \subseteq S$, then ${\mathcal H}_{i-1}$ contains an extension $h'$ of $h$ with domain $S$: 
if $x \in S$ appears in a copy of $\phi_{i-1}$ whose variables intersect $\dom(h)$, we set $h'(x)$ to the value taken by these variables under $h$. If $\dom(h)$ contains all the variables $x_r,\dots,x_{s},x_{s+1}$ from $\phi_i$ for some $0 \leq r < s < k$ and $x$ appears in the subformula $\phi_{i-1}(x_s,x_{s+1})$ and some variable of $\phi_{i-1}(x_r,x_{r+1})$ appears in $\dom(h)$, then we set $h'(x)$ to the value taken by this variable under $h$. Otherwise, pick $h'(x)$ to be a new value. Hence, ${\mathcal H}_0,{\mathcal H}_1,\dots,{\mathcal H}_q$ is indeed a winning strategy for Duplicator for the $q$ nested width $m$ guarded $k$ pebble game on $(\bA,\bC)$. 
\end{proof}

Note that the following is a strengthening of Corollary~\ref{cor:gqplus}. 

\begin{corollary}\label{cor:gqplus+}
	The intersection of MSO and Datalog contains problems that cannot be expressed in GQ$^+$. 
\end{corollary}
\begin{proof}
 We already know that 
 the complement of $\Csp(\bC)$ is in Datalog (Proposition~\ref{prop:datalog}), but not in GQ$^+$ (Proposition~\ref{prop:eq-c-gq}). 
	Proposition~\ref{prop:eq-mso} shows that $\Csp(\bB)$ is in MSO; then, as MSO is closed under negation, the complement of $\Csp(\bB)$ is MSO-expressible as well.
\end{proof}
\color{black}

\section{A coNP-complete CSP in MSO}
\label{sect:coNP}
How tame is the intersection of MSO and Datalog? We have seen that every CSP which is in MSO and in Datalog can be expressed as the CSP of an $\omega$-categorical structure (Corollary~\ref{cor:omega-cat}).  
It would be interesting to know 
whether this can be strengthened further by 
showing that these problems can even be expressed as CSPs for a particularly well-behaved class of structures, namely for 
\emph{reducts of finitely bounded homogeneous structures} (the definition of this class can be found below). 
It has been conjectured that every CSP 
for a structure from this class is 
 in P or NP-complete~\cite{BPP-projective-homomorphisms}; we will refer to this conjecture as the \emph{infinite-domain dichotomy conjecture}. 
 This conjecture generalises the finite-domain complexity dichotomy that was conjectured by Feder and Vardi~\cite{FederVardi} and proved by Bulatov~\cite{BulatovFVConjecture} and by Zhuk~\cite{ZhukFVConjecture}.


\begin{definition}\label{def:homog}
A relational structure is called 
\emph{homogeneous} if every 
isomorphism between finite substructures can be extended to an automorphism.
\end{definition}

Homogeneous structures with a \emph{finite} relational signature (and all of their reducts) are $\omega$-categorical. 
Reducts $\bB$ of homogeneous
structures with finite relational signature are a particularly well-behaved subclass of the class of $\omega$-categorical structures. See~\cite{BKOPP} for an example of 
a result that only holds for such $\bB$, but not for $\omega$-categorical structures in general. 
However, there are uncountably many countable homogeneous digraphs with pairwise distinct CSP, and it follows that there are homogeneous digraphs with undecidable CSPs~\cite{BodirskyNesetrilJLC}. 

To define a class of structures whose CSPs have better computational properties we need the following concept. 
The \emph{age} of a structure $\bB$ is the class of all finite structures that embed into $\bB$. 
A structure $\bB$ is called
\emph{finitely bounded} if there exists a 
finite set $\mathcal F$ of finite structures 
such that a finite structure $\bA$ belongs to the age of $\bB$ if and only if no structure in ${\mathcal F}$ embeds into $\bA$. 

\begin{example}
The structure $\bB$ from Section~\ref{sect:sep} is finitely bounded and homogeneous.
\enex
\end{example}

Clearly, the CSP of finitely bounded structures $\bB$ is in NP. 
It is also easy to see that every CSP of a reduct of a finitely bounded structure is in NP as well. 


\begin{question}\label{quest:fin-bd-hom}
Is it true that every CSP in MSO is 
the CSP for a reduct of a finitely bounded homogeneous structure? 
\end{question} 

In this section, we do not answer this question, but at least we can show that the class of CSPs in MSO is (under complexity-theoretic assumptions) larger than the class of CSPs for reducts of finitely bounded structures (see Section~\ref{sect:csps}). 
In particular, the class of CSPs that can be expressed in MSO does not fall into the scope of the mentioned infinite-domain CSP dichotomy conjecture.

A \emph{tournament} is a directed graph such that for any pair of vertices $u$ and $v$, there is an edge from $u$ to $v$ or an edge from $v$ to $u$, but not both. 
Let $\mathcal T = \{\bT_2,\bT_3,\dots\}$ be the set of \emph{Henson tournaments}: 
the tournament $\bT_n$, for $n \geq 2$, has vertices $0,1,\dots,n+1$ and the following edges:
\begin{itemize}
\item $(i,i+1)$ for $i \in \{0,\dots,n\}$;
\item $(0,n+1)$;
\item $(j,i)$ for $i+1 < j$ and $(i,j) \neq (0,n+1)$. 
\end{itemize}
The class $\mathcal C$ of all finite loopless digraphs that do not embed any of the digraphs from $\mathcal T$
is an amalgamation class, and hence there
exists a homogenous structure $\bH$ with age $\mathcal C$. It has been shown in~\cite{BodirskyGrohe} that $\Csp(\bH)$ is coNP-complete. 


\begin{proposition}\label{prop:couterexpl}
$\Csp(\bH)$ can be expressed in MSO. 
\end{proposition}
\begin{proof}
We have to find an MSO sentence $\Phi$  that holds
on a given digraph $(V;E)$ if and only if
$(V;E)$ does not embed any of the tournaments from $\mathcal T$. 
We specify an MSO $\{X,E\}$-sentence $\Phi$, for a unary relation symbol $X$, 
that is
true on a finite $\{X,E\}$-structure $\bS$ if and only if
$(X^\bS;E^\bS)$ is isomorphic to $\bT_n$, for some $n \geq 2$. 
In $\Phi$, we existentially quantify over 
\begin{itemize}
\item two vertices $s,t \in X$ (that stand for the
vertex $0$ and the vertex $n+1$ in $\bT_n$). 
\item a partition of $X \setminus \{s\}$ 
into two sets $A$ and $B$ (they stand for the set of even and the set of odd numbers in $\{1,\dots,n+1\}$). 
\end{itemize}
The formula $\Phi$ has the following conjuncts:
\begin{enumerate}
\item a first-order formula that states that 
$E$ defines a tournament on $X$; 
\item a first-order formula that expresses that $E$ is a linear order on $A$ with maximal element $a$; 
\item a first-order formula that expresses that $E$ is a linear order on $B$ with maximal element $b$; 
\item $E(s,t)$, $E(s,a)$, $E(a,b)$, and $E(x,s)$ for all $x \in X \setminus \{a,t\}$;  
\item a first-order formula that states that if there is an edge from an element $x \in A$ to an element $y \in B$ then there is precisely one
element $z \in A$ such that $(y,z),(z,x) \in E$, 
unless $y=t$; 
\item a first-order formula that states that if there is an edge from an element $x \in B$ to an element $y \in A$ then there is precisely one
element $z \in B$ such that $(y,z),(z,x) \in E$, 
unless $y=t$. 
\end{enumerate}
We claim that the MSO sentence $\forall x \colon \neg E(x,x) \wedge \forall X \colon \neg \Phi$ 
holds on a finite digraph if and only if the digraph  is loopless and does not embed $\bT_n$, for all $n \geq 3$. 
The forwards implication easily follows from the observation that if $(X;T)$ is isomorphic to $\bT_n$, for some $n \geq 2$, then $\phi$ holds; this is straightforward from the construction of $\Phi$: we set $s$ to $0$, $t$ to $n+1$, we set $A$ to the set of even and $B$ to the set of odd numbers in $\{1,\dots,n+1\}$. 
Conversely, suppose that $\Phi$ holds. Then
$(X;T)$ is a tournament. We construct an isomorphism $f$ from $(X;T)$ to $\bT_{|X|-1}$ as follows. Define $f(s) \ceq 0$, $f(a) \ceq1$, and $f(b) = 2$. 
Since $E(a,b)$, by item 5 there exists exactly one $a' \in A$ such that $E(b,a')$ and $E(a',a)$. 
Define $f(a') \ceq 3$. 
If $a' = t$ then we have found an isomorphism with $\bT_2$. Otherwise, the partial map $f$ defined so far is an embedding into $\bT_n$ for some $n \geq 3$. 
Item 6 and $E(b,a')$ imply that there exists exactly one $b' \in B$ such that $E(a',b')$ and $E(b',b)$,
and we define $f(b') \ceq 4$. Continuing in this manner, we eventually define $f$ on all of $X$ and find an isomorphism with $\bT_{|X|-1}$. 
\end{proof}

This shows that, unless NP = coNP, $\Csp(\bH)$ cannot be expressed as $\Csp(\bB)$ for some reduct of a finitely bounded  structure, because such CSPs are in NP.  

\begin{corollary}
If NP $\neq$ coNP, then there
are CSPs in MSO that cannot be expressed as CSPs of finitely bounded structures. 
\end{corollary}

We do not know how to show this statement without complexity-theoretic assumptions, even if we just want to rule out that $\Csp(\bH)$ can be expressed as $\Csp(\bB)$ for some reduct of a finitely bounded  \emph{homogeneous} structure. 
However, using the tools from Section~\ref{sect:pebble} we can rule out that the complement of $\Csp(\bH)$ can be expressed in Datalog. Note that this insight is obtained without resorting to any complexity-theoretic assumptions.  

\begin{theorem}\label{thm:appl}
Let $\Phi$ be the MSO sentence from the proof of Proposition~\ref{prop:couterexpl}. 
Then 
$\llbracket \Phi \rrbracket$  
cannot be expressed in Datalog. 
\end{theorem} 
\begin{proof}
We use Theorem~\ref{thm:main};
it suffices to find for all $\ell,k$ a finite model $\bA$ of $\Phi$ such that Duplicator wins the existential 
$(\ell,k)$-pebble game for $\llbracket \Phi \rrbracket$ 
on $\bA$. Simply choose $\bA$ to be the tournament $\bT_{k+1}$; then clearly $\bA \models \neg \Phi$. Note that $\Csp(\bH) \cap \llbracket \Phi \rrbracket = \emptyset$, 
so it suffices to find a winning strategy ${\mathcal H}$ for Duplicator 
for the existential $(\ell,k)$-pebble game on $(\bA,\bH)$. 

Let $\mathcal H$ be the set of all partial embeddings $h$ from $\bA$ to $\bB$ such that $\dom(h) \leq k$. 
Then $\mathcal H$ is clearly 
 non-empty and closed under restrictions. 
To verify that $\mathcal H$ is a winning strategy for Duplicator we still need to show that for all $h \in {\mathcal H}$ with $\dom(h) = \{a_{k-d+1},\dots,a_k\}$ and for all $a_1,\dots,a_{k-d} \in A$ there is an extension $h' \in {\mathcal H}$ of $h$ such that $h'$ is also defined on $a_1,\dots,a_{k-d}$. Let $\bA'$ be the substructure of $\bA$ with the domain $\{a_1,\dots,a_k\}$. 
Since $|A'| < |A|$, there exists an embedding of $\bA'$ into $\bH$; we identify the elements of $A'$ with their image under this embedding, so that $\bA'$ is now a substructure of $\bH$, and $h$ is an isomorphism between finite substructures of $\bH$. 
 By the homogeneity of $\bH$, there exists an automorphism $\alpha$ of
$\bH$ which extends $h$. Then the restriction of $\alpha$ to $a_1,\dots,a_k$ is the desired extension $h'$ of $h$ in ${\mathcal H}$. 
\end{proof}


\section{Conclusion, Open Problems, and Prospect}
In this article, we provided a game-theoretic characterisation of the problems in Guarded Second-order Logic that are equivalent to a Datalog program. We also proved the existence of canonical Datalog programs for GSO sentences whose models are closed under homomorphisms. 
To prove these results, we showed 
that every class of finite $\tau$-structures in GSO whose complement is closed under homomorphisms is a finite union of CSPs. 

Our results imply that the so-called universal-algebraic 
approach, which has eventually led to the classification of finite-domain CSPs in Datalog~\cite{BoundedWidthJournal}, can be applied to study problems that are simultaneously in GSO and in Datalog (see~\cite{Collapses}). 
They might also pave the way towards a syntactic characterisation of GSO $\cap$ Datalog and of  MSO $\cap$ Datalog;
however, we showed that the class of nested guarded queries from~\cite{BourhisKroetzschRudolph} is not expressive enough for this purpose. 
See Figure~\ref{fig:summary} for an overview of the considered logics, their inclusions, and references to examples that show that the inclusions are strict. 

\begin{figure}
\begin{center}
\begin{tikzpicture}
    \matrix (A) [matrix of nodes, row sep=1.2cm]
    { 
        SO & & \\
	Datalog & GSO &  \\
        & GQ$^+$ & MSO \\
        & GQ & Nemodeq \\
        & Frontier-Guarded Datalog & Modeq \\
        & & Monadic Datalog \\
    };
    \draw (A-1-1) -- (A-2-2) 
    node[midway, above] {\quad\quad\red{Expl.~\ref{expl:datalog-completement}}};
    \draw (A-1-1)--(A-2-1) 
  node[pos=.3, left] {\red{CSP$(K_3)$}}
  node[pos=.7, right] {Thm.~\ref{thm:gq-to-datalog}};
    \draw (A-2-1)--(A-3-2)
    node[midway, above] {\red{\quad\quad Expl.~\ref{expl:datalog-completement}}};
    \draw (A-2-2)--(A-3-2) 
    node[pos=.7, right] {Prop.~\ref{prop:gqpgso}};
    \draw (A-2-2)--(A-3-3) 
     node[midway, above] {\quad\red{Prop.~\ref{prop:not-mso}}};
   \draw (A-3-2)--(A-4-3) 
     node[midway, above] {\quad\red{Prop.~\ref{prop:not-mso}}};
   \draw (A-4-2)--(A-5-3)
     node[midway, above] {\quad\red{Prop.~\ref{prop:not-mso}}};
   \draw (A-5-2)--(A-6-3)
     node[midway, above] {\quad\red{Prop.~\ref{prop:not-mso}}};
    \draw (A-3-2)--(A-4-2) 
    node[midway, left] {\red{Prop.~\ref{prop:inexpr}}};
    \draw (A-4-2)--(A-5-2)
    node[midway, left] {\red{Prop.~\ref{prop:Q}}};
    \draw (A-3-3)--(A-4-3)
    node[pos=.3, right] {\red{CSP$(K_3)$}}
    node[pos=.7, right] {\cite{RK2013}};
    \draw (A-4-3)--(A-5-3)
    node[midway, right] {\red{Prop.~\ref{prop:inexpr}}};
    \draw (A-5-3)--(A-6-3)
    node[midway, right] {\red{Prop.~\ref{prop:Q}}};
\end{tikzpicture}
\end{center}
\caption{A Hasse Diagram of the logical formalisms considered in this article, ordered by expressive power. The red labels give references to propositions or examples in this text that show that the respective inclusion is strict.} 

\label{fig:summary}
\end{figure}

We state some concrete open problems in the context. 
\begin{enumerate}
\item Is there a \emph{finite} structure whose CSP is in Datalog, but not in Nemodeq? 
Also the expressive power of frontier-guarded Datalog, GQ, and GQ$^+$ for finite-domain CSPs appears to be unexplored. 
\item It is known that \emph{singleton linear arc consistency (SLAC)}~\cite{Kozik16} captures the intersection of finite-domain CSPs with Datalog. It would be interesting to define an appropriate notion of  singleton linear arc consistency for 
CSPs for $\omega$-categorical structures; can it be used to precisely characterise the intersection of MSO and Datalog? 
\item Is there a CSP of a reduct of a finitely bounded homogeneous structure which is not expressible in GSO?
\item Question~\ref{quest:fin-bd-hom}: Is every CSP in MSO $\cap$ Datalog the CSP of a reduct of a finitely bounded homogeneous structure? 
\item Is every CSP in existential MSO also in 
MMSNP? 
\end{enumerate}

We are also confident that our results will advance the understanding of CSPs (the complements of) which are obtained as the homomorphism-closure of the set of some theory's finite models. For example, the homomorphism-closures of the model sets of guarded- and guarded-negation-theories have recently been found to be GSO-expressible~\cite{BodFelKnaRud} so, by virtue of our results, we immediately know they must be (complements of) $\omega$-categorical CSPs.

\bibliography{local.bib}

\end{document}